\documentclass[journal,onecolumn]{IEEEtran}
\usepackage[dvipsnames]{xcolor}

\usepackage{etoolbox}
\newtoggle{longversion}
\toggletrue{longversion}
\newtoggle{isitlongversion}
\usepackage{bbm}
\usepackage{url}
\usepackage{enumitem}
\usepackage{hyperref}
\hypersetup{
  colorlinks   = true, 
  urlcolor     = black!50!red, 
  linkcolor    = black!50!brown, 
  citecolor   = black!50!brown 
}
\usepackage{pgfplots}
\pgfplotsset{compat=1.15}
\usepackage{amsmath}
\usepackage{amssymb}

\usepackage{amsthm,thm-restate}
\usepackage{verbatim}
\usepackage{bm}
\usepackage{color,graphicx}

\usepackage{graphicx}
\usepackage{mdwtab}
\usepackage{subfigure}
\usepackage{mathtools,tikz}
\usetikzlibrary{circuits.ee.IEC}
\mathtoolsset{showonlyrefs} 

\tikzset{
  declare function={
    atan3(\a,\b)=ifthenelse(atan2(0,1)==90, atan2(\a,\b), atan2(\b,\a));},
  kinky cross radius/.initial=+.125cm,
  @kinky cross/.initial=+, kinky crosses/.is choice,
  kinky crosses/left/.style={@kinky cross=-},kinky crosses/right/.style={@kinky cross=+},
  kinky cross/.style args={(#1)--(#2)}{
    to path={
      let \p{@kc@}=($(\tikztotarget)-(\tikztostart)$),
          \n{@kc@}={atan3(\p{@kc@})+180} in
      -- ($(intersection of \tikztostart--{\tikztotarget} and #1--#2)!%
             \pgfkeysvalueof{/tikz/kinky cross radius}!(\tikztostart)$)
      arc [ radius     =\pgfkeysvalueof{/tikz/kinky cross radius},
            start angle=\n{@kc@},
            delta angle=\pgfkeysvalueof{/tikz/@kinky cross}180 ]
      -- (\tikztotarget)}}}

\usepackage{hhline}
\usepackage{multirow}
\usepackage{pdfpages}
\usepackage{enumitem}  
\iftoggle{longversion}{
\newcommand{\longonly}[1]{#1} 
\newcommand{\isitlongonly}[1]{} 
\newcommand{\isitshortonly}[1]{} 
\newcommand{\shortonly}[1]{} 
}
{
\newcommand{\longonly}[1]{} 
\newcommand{\shortonly}[1]{#1} 

\iftoggle{isitlongversion}
{
\onecolumn
\newcommand{\isitlongonly}[1]{#1} 
\newcommand{\isitshortonly}[1]{} 

}
{
\twocolumn
\setlength{\abovedisplayskip}{0.3em}
\setlength{\belowdisplayskip}{0.3pt}
\setlength{\abovedisplayshortskip}{0pt}
\setlength{\belowdisplayshortskip}{0pt}
\newcommand{\isitlongonly}[1]{} 
\newcommand{\isitshortonly}[1]{#1} 

}
}

\IEEEoverridecommandlockouts
\usepackage{cite}
\usepackage{amsmath,amssymb,amsfonts,amsthm,xspace}
\usepackage{algorithmic}
\usepackage{graphicx}
\usepackage{textcomp}

\usetikzlibrary{arrows,shapes.geometric}
\allowdisplaybreaks
\newtheorem{remark}{Remark}

\newtheorem{thm}{Theorem}
\newtheorem{defn}{Definition}
\newtheorem{example}{Example}
\newtheorem{lemma}[thm]{Lemma}
\newtheorem{corollary}[thm]{Corollary}
\newtheorem{prop}[thm]{Proposition}
\newcommand{\subparagraph}[1]{\par {\em\underline{#1:}}}
\newtheorem*{Lemma}{Lemma}
\newtheorem*{Theorem}{Theorem}
\newtheorem*{Example}{Example 1 (continued)}
\newtheorem*{Eexample}{Example 3 (continued)}
\usepackage{pgffor}
\foreach \x in {a,...,z}{%
\expandafter\xdef\csname vec\x \endcsname{\noexpand\ensuremath{\noexpand\bm{\x}}}
}

\foreach \x in {A,...,Z}{%
\expandafter\xdef\csname vec\x \endcsname{\noexpand\ensuremath{\noexpand\bm{\x}}}
}

\foreach \x in {A,...,Z}{%
\expandafter\xdef\csname c\x \endcsname{\noexpand\ensuremath{\noexpand\mathcal{\x}}}
}

\foreach \x in {A,...,Z}{%
\expandafter\xdef\csname bb\x \endcsname{\noexpand\ensuremath{\noexpand\mathbb{\x}}}
}

\newcommand{\defineqq}{\ensuremath{\stackrel{\textup{\tiny def}}{=}}}

\usetikzlibrary{shapes,arrows,positioning,decorations.pathreplacing,calc}

\def\msg{\ensuremath{m}} 
\def\msgh{\ensuremath{\hat{\msg}}} 
\def\nummsg{\mbox{$N$}} 

\def\auth{\ensuremath{\mathrm{auth}}}

\def\numrand{\mbox{$L$}} 
\def\codeset{\Gamma} 
\newcommand{\fo}{\ensuremath{f_{\one}}\xspace}
\newcommand{\ft}{\ensuremath{f_{\two}}\xspace}
\newcommand{\Fo}{\ensuremath{F_{\one}}\xspace}
\newcommand{\Ft}{\ensuremath{F_{\two}}\xspace}
\newcommand{\rand}{\ensuremath{\textup{rand}}\xspace}

\newcommand{\mch}{\ensuremath{W_{Z|XY}}\xspace}

\newcommand{\indep}{\raisebox{0.05em}{\rotatebox[origin=c]{90}{$\models$}}\xspace}

\newcommand{\na}{\ensuremath{\text{hon}}\xspace}
\newcommand{\malone}{\ensuremath{\text{mal \one}}}
\newcommand{\maltwo}{\ensuremath{\text{mal \two}}}

\newcommand{\spoofable}{\text{spoofable}\xspace}

\newcommand{\spoofability}{\text{spoofability}\xspace}

\newcommand{\inb}[1]{\left\{#1\right\}}
\newcommand{\inp}[1]{\left(#1\right)}
\newcommand{\insq}[1]{\left[#1\right]}

\def\reliable{\ensuremath{\mathrm{reliable}}}
\def\MAC{\ensuremath{\mathrm{MAC}}}
\def\AVMAC{\ensuremath{\mathrm{AV-MAC}}}

\def\rand{\ensuremath{\mathrm{rand}}}

\newcommand{\mo}{\ensuremath{m_{\one}}\xspace}
\newcommand{\mt}{\ensuremath{m_{\two}}\xspace}
\newcommand\independent{\protect\mathpalette{\protect\independenT}{\perp}}
\def\independenT#1#2{\mathrel{\rlap{$#1#2$}\mkern2mu{#1#2}}}

\def\one{\ensuremath{\mathsf{A}}\xspace} 
\def\two{\ensuremath{\mathsf{B}}\xspace} 

\def\oneb{\ensuremath{\mathbf{a}}\xspace} 
\def\twob{\ensuremath{\mathbf{b}}\xspace} 

\newcommand{\bmac}{\ensuremath{\text{byzantine-MAC}}\xspace}
\newcommand{\bmacs}{\ensuremath{\text{byzantine-MACs}}\xspace}
\newcommand{\red}[1]{{\textcolor{red}{#1}}}
\newcommand{\blue}[1]{{\textcolor{black}{#1}}}
\newcommand{\olive}[1]{{\textcolor{black}{#1}}}

\newcommand{\todo}[1]{{\textcolor{blue!50!red}{{\bf To Do:} #1}}}

\begin{document}

\title{Byzantine Multiple Access Channels --- Part II: Communication With Adversary Identification\footnote{This work was presented in part at the 2021 IEEE International Symposium on Information Theory~\cite{NehaBDP21}.} 
\thanks{N. Sangwan and V. Prabhakaran acknowledge support of
the Department of Atomic Energy, Government of India,
under project no. RTI4001. N. Sangwan was
additionally supported by the Tata Consultancy Services (TCS)
foundation through the TCS Research Scholar Program. The work of M. Bakshi was supported in part by the Research Grants Council of the Hong Kong Special Administrative Region, China, under Grant GRF 14300617, and in part by the National Science Foundation under Grant No. CCF-2107526. The work
of B. K. Dey was supported by Bharti Centre for
Communication in IIT Bombay. V. Prabhakaran was additionally supported by the Science \& Engineering Research Board, India
through project MTR/2020/000308.}}

\author{Neha~Sangwan,~\IEEEmembership{Member,~IEEE,}
        Mayank~Bakshi,~\IEEEmembership{Member,~IEEE,}
        Bikash~Kumar~Dey,~\IEEEmembership{Member,~IEEE,}
        and~Vinod~M.~Prabhakaran,~\IEEEmembership{Member,~IEEE}%
\thanks{\blue{N.~Sangwan is with University of California San Diego, CA, USA. She was with the School of Technology
and Computer Science, Tata Institute of Fundamental Research, Mumbai
400~005, India.}}
\thanks{M.~Bakshi is with Arizona State University, Tempe, AZ, USA.}%
\thanks{B.~K.~Dey is with Indian Institute of Technology Bombay, Mumbai
400~076, India.}
\thanks{\blue{V.~M.~Prabhakaran is with the School of Technology
and Computer Science, Tata Institute of Fundamental Research, Mumbai
400~005, India.}}}


\maketitle

\begin{abstract}
We introduce the problem of determining the identity of a byzantine user (internal adversary) in a communication system. We consider a two-user discrete memoryless multiple access channel where either user may deviate from the prescribed behaviour. 
\blue{Since small deviations may be indistinguishable from the effects of channel noise, it might be overly restrictive to attempt to detect all deviations.}
In our formulation, we only require detecting deviations which impede the decoding of the non-deviating user's message. When neither user deviates, correct decoding is required. When one user deviates, the decoder must  either output a pair of messages of which the message of the {non-deviating} user is correct or identify the deviating user. {The users and the receiver do not share any randomness.} The results include a characterization of the set of channels where communication is feasible, and an inner and outer bound on the capacity region. {We also show that whenever the rate region has non-empty interior, the capacity region is same as the capacity region under randomized encoding, where each user shares independent randomness with the receiver. We also give an outer bound for this randomized coding capacity region.}
\end{abstract}


\section{Introduction}\label{sec:intro}
In many modern wireless communication applications (e.g., the Internet of Things), devices with varying levels of security are connected over a shared communication medium.  Compromised devices may allow an adversary to disrupt the communication of other devices. This motivates the question we study in this paper -- is it possible to design a communication system in which malicious actions by compromised devices can be detected so that such devices can be isolated or taken offline?

We consider a two-user Multiple Access Channel (MAC) where either user may deviate from the prescribed behaviour\footnote{\blue{In this work, we do not consider generalization to more than two users. The general case was studied for reliable communication (a different decoding guarantee than the present problem) over \bmac in Part I \cite{NehaBDP23}.}}. The deviating user (if any) is fixed for the entire duration of the transmission. We will refer to this channel as a \bmac in the rest of the paper. Owing to the noisy nature of the channel, it may be impossible or overly restrictive to attempt to detect all deviations. Indeed, it suffices to detect only such deviations which impede the correct decoding of the other user's message. We formulate a communication problem for the \bmac with the following decoding guarantee (Fig.~\ref{fig:authcomMAC}): the decoder outputs either a pair of messages {\em or} declares one of the users to be deviating. When both users are honest, the decoder must output the correct message pair with high probability (w.h.p.); when exactly one user deviates, w.h.p., the decoder must either correctly detect the deviating user or output a message pair of which the message of the other (honest) user is correct (see Section~\ref{sec:model}). No guarantees are made if both users deviate. Thus, we require that a deviating user cannot cause a decoding error for the other user without getting caught. We call this problem as that of {\em communication with adversary identification}.
The focus of this paper is on the case where the decoder does not share any randomness with either of the encoders. We consider the average probability of error criterion, that is, an honest user sends a codeword uniformly at random. The identity of this codeword is not known to the other user.
\begin{figure}[h]\centering
\begin{tikzpicture}[scale=0.4]
	\draw (2,0) rectangle ++(3,2) node[pos=.5]{ user \two};
	\draw (2,2.4) rectangle ++(3,2) node[pos=.5]{ user \one};
	\draw (10-2,0) rectangle ++(3.5,4.4) node[pos=.5]{ $W_{Z|XY}$};
	\draw (16.5-2,1.2) rectangle ++(3.5,2) node[pos=.5]{ Decoder};
	\draw[->] (1,1) node[anchor=east]{ $\msg_{\two}$} -- ++ (1,0) ;
	\draw[->] (5,1) -- node[above] { $\vecY$} ++ (3,0);
	
	\draw[->] (1,3.4) node[anchor=east]{ $\msg_{\one}$} -- ++ (1,0) ;
	\draw[->] (5,3.4) -- node[above] { $\vecX$} ++ (3,0);

	\draw[->] (13.5-2,2.2) -- node[above] { $\vecZ$} ++ (3,0);
	\draw[->] (19-1,2.2) -- ++ (1,0) node[anchor=west]{\begin{array}{c} \msgh_{\one},\msgh_{\two}\\ \text{or}\\ \oneb\text{ or }\twob\end{array}};
\end{tikzpicture}
\caption{{\em Communication with adversary identification} in a \bmac: Reliable decoding of both the messages is required when neither user deviates. When a user (say, user \two) deviates, the decoded message should either be correct for the honest user or the decoder should identify the deviating user (by outputting $\twob$) with high probability.}\label{fig:authcomMAC}
\end{figure}

For comparison, consider the stronger guarantee  of {\em reliable communication} where the decoder outputs a pair of messages such that the message(s) of non-deviating user(s) is correct w.h.p. \cite{NehaBDPITW19,NehaBDP23}.\footnote{\blue{\cite{NehaBDP23} is  Part I of this work titled ``Byzantine Multiple Access Channels - Part I: Reliable Communication''. \cite{NehaBDPITW19} is the corresponding conference version.}} While achieving this clearly satisfies the requirements of the present model, it might be too demanding as we discuss in the following example.

\begin{example}[Binary erasure MAC{~\cite[pg.~83]{YHKEG}}]\label{ex:1} Binary erasure MAC is a deterministic MAC model with binary inputs $X, Y$ and ternary output $Z=X+Y$ where $+$ is real addition and $Z\in \inb{0, 1, 2}$. In this channel, a deviating user can run an independent copy of the honest user's encoder and inject a spurious message which will appear equally plausible to the decoder as the honest user's actual message (see Example~\ref{ex:1} in Section~\ref{sec:comparison}). Thus, it is not possible to drive the probability of error to zero for any message set of size at least two under the reliable communication guarantee. 
However, consider a simple scheme which uses all strings of Hamming weight 1 as codebook of user \one\ and all strings of Hamming weight $n-1$ as codebook of user \two.
We use the following decoder: 
If the sum of entries of the output string is $n+2$ or more,  user \one\ is malicious. 
On the other hand, if it is $n-2$ or less, then user \two\ is malicious. Note that the output will be correct with probability 1 in these cases. When the sum of entries of the output string is $n+1$ but there is no location with symbol 0, then the decoder declares user \two\ to be malicious, otherwise, user \one\ is malicious. This is because when user \one\ is malicious and sends a string of hamming weight 2, user \two, being honest, sends one of the codewords uniformly at random, the identity of which is unknown to user \one.
Thus, the probability that there is no 0 in the output string is $2/n$ which vanishes with $n$. 
Similarly, when the sum of entries of the output string is $n-1$, the decoder declares user \one\ to be malicious when there is no 2 in the output string, and declares user \two\ to be malicious otherwise. 
Along similar lines, we can also argue that when the sum of entries of the output string is $n$, with a non-vanishing probability, there is a 2 and a 0 in the output string. 
Such an output only corresponds to two unique input strings of users which are both valid codewords. Hence, the output would be correct. 
Thus, we can get a vanishing probability of error using this scheme, though with no rate. 
In fact, we show that for this channel, the capacity region of communication with adversary identification is the same as the (non-adversarial) capacity region of the binary erasure MAC. \blue{In Section~\ref{sec:comparison}, we  show the feasibility of communication with adversary identification over this channel. We obtain the capacity region in Section~\ref{sec:example_tightness}. }
\end{example}

{Another decoding guarantee, that is weaker than} the present model, allows the decoder  to declare adversarial interference (in the presence of malicious user(s)) without identifying the adversary. We called this model {\em authenticated communication} and characterized its  feasibility condition and capacity region in\cite{NehaBDPISIT19}. The feasibility condition is called {\em overwritability}, a notion which was introduced by Kosut and Kliewer for network coding \cite{KK2} and arbitrarily varying channels \cite{KosutKITW18}. 

{The present model lies between the models for reliable communication and authenticated communication in a byzantine MAC. However, obtaining results here appears to be significantly more challenging.}
On the one hand, for reliable communication over the two-user MAC, we may treat the channel from each user to the decoder as an arbitrarily varying channel (AVC) \cite{BBT60} with the other user's input as state. Hence, the users may send their messages using the corresponding AVC codes~\cite{CsiszarN88}. Thus, the rectangular region defined by the capacities of the two AVCs is achievable\footnote{In fact, this rectangular region defined by the capacities of the two AVCs is the reliable communication capacity region since a deviating user can act exactly like the adversary in the AVC of the other user. See \cite[Section~1.3]{NehaBDP23}. Note that the AVCs for binary erasure MAC have zero AVC capacity.\label{ftn:reliable}}.  
On the other hand, for authenticated communication over the two-user MAC, our achievable strategy in \cite{NehaBDPISIT19} involved an unauthenticated communication phase using a non-adversarial MAC code followed by separate (short) authentication phases for each user’s decoded message. 
Failure to authenticate a user's message implies the presence of an adversary (though not its identity since the user whose message is being authenticated might have deviated to cause the authentication failure).
In both 
the cases above, the decoder, when it accounts for the byzantine nature of the users, deals with the users one at a time. 
However, similar decoding strategies seem to be insufficient for adversary identification.
Determining the identity of a deviating user requires dealing with the byzantine nature of both users simultaneously, thereby complicating the decoder design (see Section~\ref{sec:feasibility}).

We characterize the infeasibility of communication with adversary identification using a condition on the channel we call {\em \spoofability} (see Fig.~\ref{fig:spoof1}). 
In a spoofable channel, a deviating user may mount an attack which can be confused with an attack of the other user and which introduces a spurious message that can be confused with the actual message of the (other) honest user.
When the channel is not \spoofable, a deterministic code in the style of \cite{CsiszarN88} can provide positive rates to both the users (Theorem~\ref{thm:main_result}). We give an inner and an outer bound to the deterministic coding capacity region. 
Our outer bound is in terms of the capacity of an arbitrarily varying multiple access channel{\cite{AhlswedeC99}~(Theorem~\ref{thm:outer_bd}). 
Further, in Section~\ref{sec:comparison}, a comparison is made between \spoofability and the feasibility conditions for the reliable communication and authenticated communication models. \blue{Using example~\ref{ex:3}, we also demonstrate that the capacity regions of the three models in \bmac: {\em 1)} authenticated communication, {\em 2)} communication with adversary identification and {\em 3)} reliable communication are strictly different.}

In Section~\ref{sec:rand_capacity}, we draw connections of the present model to the case when  the users share independent randomness with the receiver. Analogous to the {\em dichotomy phenomenon} for arbitrarily varying channels~\cite{Ahlswede78} and arbitrarily varying multiple access channels~\cite{Jahn81}, we show that whenever the capacity region for the deterministic case has a non-empty interior, it is the same as the capacity region for the randomized case. We also give an outer bound on the capacity region for randomized codes.

\paragraph*{Related works} 
The model falls in the general class of adversarial channels. 
There is a long line of works in the information theory literature on communication in the  presence of external adversaries (see \cite{survey} for a survey), for example, an arbitrarily varying channel~\cite{BBT60,Ahlswede78,CsiszarN88} or an arbitrarily varying-MAC~\cite{Jahn81,Gubner,AhlswedeC99}. In these models the channel law can be arbitrarily varied by an adversary during transmission from a given set of allowed channel laws. While our model is different from these models, the technical formulation is heavily inspired. For example, similar to our deterministic coding with an average probability of error criterion, the most well studied model in arbitrarily varying channels  uses a  deterministic codebook with an average probability of error criterion where a user sends a codeword uniformly at random. 
The identity of the sent codeword(s) is hidden from the adversary. Our randomized coding model is also inspired from the corresponding randomized coding model in arbitrarily varying multiple access channels where each user shares independent randomness with the decoder which is private from the other user. Our techniques also borrow from the achievability and converse techniques in the arbitrarily varying channels literature. Further, they expand and add to the existing set of available techniques for these problems. 
In addition to communication over arbitrarily varying channels, there has also been some recent work on authenticated communication over channels in which an external adversary may be present. In the presence of the adversary, the decoder may declare adversarial interference instead of decoding  \cite{KosutKITW18,BKKGYu,Graves16,BeemerCNS20} (In a 2-user MAC model in \cite{BeemerCNS20} when declaring the presence of an adversary, the decoder is required to decode at least one user's message.). 

These models are different from the present model, where a legitimate user of the channel \blue{may be} adversarial, and when declaring the presence of an adversary, we also require the decoder to output its identity. Such users are often called {\em byzantine} users.
Communication in systems with byzantine users has also received some attention \cite{Jaggi7,KTong,KK2,Yener,NehaBDPITW19,NehaBDPISIT19,NehaBDP23}. 
Network coding with byzantine attacks on nodes and edges has been studied \cite{Jaggi7,KTong,KK2}. He and Yener  \cite{Yener} considered a Gaussian two-hop network with an eavesdropping and byzantine adversarial relay where the receiver is required to decode with message secrecy and detect byzantine attack.
The present model, on the other hand, considers byzantine users in a multiple access channel. This model was previously considered in~\cite{NehaBDP23,NehaBDPITW19,NehaBDPISIT19} but with different decoding guarantees. 
From a cryptographic point of view, message authentication codes where the users have pre-shared keys and
communicate over noiseless channels have been extensively
studied~\cite{SimmonsCRYPTO84,MaurerIT00,Gungor16}. Message authentication over noisy
channels has also been considered~\cite{LaiEPIT09,Jiang14,Gungor16,TuLIT18}. \blue{\cite{9805432} and \cite{9807270} study authentication with the help of a shared secret key over an adversary controlled channel, where adversary has the ability to replace the observation at the output. \cite{9887973} presents keyless authentication over an AWGN channel where an adversary, having side information about the message and the trasmission, can inject an additive signal into the output. Adversarial attacks and defenses in deep learning based wireless systems is another emerging area of research. See \cite{9609969,e24081047}  and surveys \cite{9653662,9887796,10263803} for details. }

\paragraph*{Summary of contributions}
These are the main contributions of this work.
\begin{itemize}
	\item We introduce the problem of communication with adversary identification in a \bmac and characterize the class of \bmacs which allow positive rates using deterministic codes under the average probability of error criterion \blue{(Theorem~\ref{thm:main_result} in Section~\ref{sec:feasibility})}. 
	\item We also provide inner and outer bounds to the capacity region \blue{(Theorems~\ref{thm:inner_bd} and \ref{thm:outer_bd} in Section~\ref{sec:capacity})}. 
	\item For a \bmac, \blue{in Section~\ref{sec:comparison}}, we compare the feasibility condition for communication with adversary identification with the corresponding \blue{feasibility conditions in the models of {\em reliable communication }and {\em authenticated communication}}. \blue{We also show a strict separation between the capacity regions under these three models using an example (see Example~\ref{ex:3} and Fig.~\ref{fig:ex3}).}
	\item  \blue{In Section~\ref{sec:rand_capacity},} we show that, for communication with adversary identification, whenever positive rates can be provided to both users under deterministic coding, the capacity region of deterministic coding is the same as the randomized coding capacity region. This is like the {\em dichotomy phenomenon} in~\cite{Ahlswede78}~and~\cite{Jahn81} for arbitrarily varying channels and arbitrarily varying multiple access channels respectively. We also give an outer bound on the randomized coding capacity region. 
\end{itemize}

\section{System Model}\label{sec:model}

\vspace{0.2 cm}
\paragraph*{Notation}
For a set $\cS\subseteq \bbR^{k}$, let $\mathsf{conv}(\cS)$ and $\mathsf{int}(\cS)$ denote its convex closure and interior respectively. For a set $\cH$, let $\cH^c$ denote its complement. For a set $\cA$, $\textsf{Unif}(\cA)$ denotes the uniform distribution over $\cA$. We denote random variables by capital letters, like $X, Y$ and  $\tilde{Y}$, and their corresponding alphabets by calligraphic letters, for example, $\cX, \cY$ and $\tilde{\cY}$. Let $\vecx\in \cX^n$ (resp. $\vecX$ distributed over $\cX^n$) denote the $n$-length vectors (resp. $n$-length random vectors). We denote the distribution of a random variable $Y$ by $P_Y$ and use the notation $Y\sim P_Y$ to indicate this.  
For an alphabet $\cX$, we define the set of all empirical distributions of $n$-length sequences in $\cX^n$ by $\cP^n_{\cX}$. For a distribution $P_X\in \cP^n_{\cX}$, let $T^n_{X}$ denote the set of all $n$-length sequences $\vecx\in \cX^n$ with empirical distribution $P_X$.
For a vector $\vecy\in \cY^n$, the statement $\vecy\in T^n_{Y}$ is sometimes used to implicitly define $P_Y$ as the empirical distribution of $\vecy$ and a random variable $Y$ distributed according to $P_Y$. For $P_{XY}\in \cP^n_{\cX\times\cY}$ and $\vecx\in T^n_{X}$, we define $T^n_{Y|X}(\vecx)\defineqq\inb{\vecy|(\vecx, \vecy)\in T^n_{XY}}$. 
For a natural number $n$, we denote the set $\inb{1, 2, \ldots, n}$ by $[1:n]$. For a real number $a$, $\exp\inp{a}$ denotes $2^{a}$ and $\log(a)$ denotes $\log_2(a)$, that is, $\exp$ and $\log$ are with respect to base 2. 
For a conditional distribution $P_{Y|X}$, we denote its $n$-fold product (memoryless channel) by $P^n_{Y|X}$. For a vector $\vecx\in \cX^n$, the term $P^n_{Y|X}(\cdot|\vecx)$ denotes the output distribution on $\cY^n$ when $\vecx$ is fed as input to the memoryless channel $P^n_{Y|X}$. \blue{For a set $\cS\subseteq\cY^n$, we use the notation $P^n_{Y|X}(\cS|\vecx)$ to denote the conditional probability of $\cS$ conditioned on $\vecx$. }For a 2-user multiple access channel $W_{Z|X Y}$ from input alphabets $\cX$ and $\cY$ to output alphabet $\cZ$, we will sometimes use $W$ to simplify the notation. Its $n$-fold product will be denoted by $W^n$.
 For a two-user MAC $W$, we will use $\cC_{\MAC}(W)$ (or simply $\cC_{\MAC}$) to denote its (non-adversarial) capacity region.  

Consider a two-user discrete memoryless Multiple Access Channel (MAC)  as shown in Fig.~\ref{fig:authcomMAC}. User $\one$ has input alphabet $\cX$ and user $\two$ has input alphabet $\cY$. The output alphabet of the channel is $\cZ$. The sets $\cX$, $\cY$ and $\cZ$ are finite. We study communication in a MAC where either user may deviate from the communication protocol by sending any sequence of its choice from its input alphabet. {While doing so, the deviating user is unaware of other user's input}. Further, the deviating user (if any) is fixed for the entire duration of the transmission. We will refer to this channel model  as a {\em \bmac}. 
\begin{defn}[Adversary identifying code]\label{defn:code}
An $(\nummsg_{\one},\nummsg_{\two},n)$  {\em deterministic adversary identifying code} for a \bmac consists of the following: 
\begin{enumerate}[label=(\roman*)]
\item Two message sets, $\mathcal{M}_i = \{1,\ldots,\nummsg_i\}$, $i=\one,\two$,
\item Two deterministic encoders, $f_{\one}:\mathcal{M}_{\one}\rightarrow \mathcal{X}^n$ and $f_{\two}:\mathcal{M}_{\two}\rightarrow\mathcal{Y}^n$, and
\item A deterministic decoder, $\phi:\mathcal{Z}^n\rightarrow(\mathcal{M}_{\one}\times\mathcal{M}_{\two})\cup\{\oneb,\, \twob\}.$ 
\end{enumerate}
\end{defn}
\blue{Definition~\ref{defn:code} defines a deterministic adversary identifying code. The achievability is shown for this code (Lemma~\ref{thm:detCodesPositivity} and Theorem~\ref{thm:inner_bd}). Our feasibility converse, on the other hand, holds even for stochastic codes where the encoders are allowed to privately randomize (Lemma~\ref{thm:converse}). We also provide results for randomized adversary identifying code where the encoders share independent randomness with the decoder (Definition~\ref{defn:rand-code}) in  Section~\ref{sec:rand_capacity}.  
}

For notational convenience, let us define the decoder $\phi_{\one}$ for user~\one's message as, for $\vecz\in\cZ^n$, 
\begin{align}
\phi_{\one}(\vecz)&=    
                \begin{cases}
                        m_{\one}&\text{if } \phi(\vecz)=(m_{\one},m_{\two})\\
                        \oneb& \text{if }\phi(\vecz) = \oneb\\
                        \twob& \text{if }\phi(\vecz) = \twob,
                \end{cases} 
\label{eq:feas_conv_dec1}\\
	\intertext{and the decoder $\phi_{\two}$ for user~\two's message as, for $\vecz\in\cZ^n$,}
\phi_{\two}(\vecz)&=
                \begin{cases}
                        m_{\two}&\text{if } \phi(\vecz)=(m_{\one},m_{\two})\\
                        \oneb& \text{if }\phi(\vecz) = \oneb\\
                        \twob& \text{if }\phi(\vecz) = \twob.
                \end{cases} \label{eq:feas_conv_dec2}
\end{align}  
The decoder outputs the symbol \oneb to declare that user \one is adversarial. Similarly, an output of \twob is used to declare that user \two is adversarial. The {\em average probability of error} $P_{e}(f_{\one},f_{\two},\phi)$ is the maximum of the average probabilities of error in the following three cases: (1) both users are honest, (2) user \one is adversarial, and (3) user \two is adversarial. When both users are honest, an error occurs if the decoder does not output the pair of correct messages. Let $\cE_{\mo,\mt} = \inb{\vecz:\phi(\vecz)\neq(\mo, \mt)}$ denote the corresponding error event.  The average error {probability} when both users are honest is
\begin{align}
&P_{e,\na}\hspace{-0.25em} \defineqq\frac{1}{N_{\one}\cdot N_{\two}} \sum_{\substack{(\mo, \mt)\in\\ \mathcal{M}_{\one}\times\mathcal{M}_{\two}}}W^n\inp{\cE_{\mo,\mt}|f_{\one}^{(n)}(\mo), f_{\two}^{(n)}(\mt)}.\label{eq:na}
\end{align}
\blue{Recall that $W^n\inp{\cE_{\mo,\mt}|f_{\one}^{(n)}(\mo), f_{\two}^{(n)}(\mt)}$ is used to denote the conditional probability of the set $\cE_{\mo,\mt}$ conditioned on the input $\inp{f_{\one}^{(n)}(\mo),f_{\two}^{(n)}(\mt)}$. }When user \one is adversarial, an error occurs unless the decoder's output is either the symbol $\oneb$ or a pair of messages of which the message of user \two is correct. The corresponding error event is $\cE^{\two}_{\mt} \defineqq \inb{\vecz:\phi_{\two}(\vecz)\notin\{\mt, \oneb\}}$. 
The average probability of error when user \one is adversarial is 
\begin{align}
&P_{e,\malone} \defineqq \max_{\vecx\in\cX^n} \left(\frac{1}{N_{\two}}\sum_{m_{\two}\in \mathcal{M}_{\two}}W^n\inp{\cE^{\two}_{\mt}|\vecx, f_{\two}(\mt)}\right).\label{eq:mal1}
\end{align} Similarly, for $\cE^{\one}_{\mo} \defineqq \inb{\vecz:\phi_{\one}(\vecz)\notin\{\mo,\twob\}}$, 
the average probability of error when user \two is adversarial is 
\begin{align}
&P_{e,\maltwo} \defineqq \max_{\vecy\in\cY^n}\left(\frac{1}{N_{\one}}\sum_{m_{\one}\in \mathcal{M}_{\one}}W^n\inp{\cE^{\one}_{\mo}|f_{\one}(\mo), \vecy}\right).\label{eq:mal2}
\end{align}  
We define the \emph{average probability of error} as
\begin{align*}
P_{e}(f_{\one},f_{\two},\phi)\defineqq \max{\inb{P_{e,\na},P_{e,\malone},P_{e,\maltwo}}}.
\end{align*}
We note that $\cE_{\mo, \mt}  = \cE^{\one}_{\mo}\cup\cE^{\two}_{\mt}$. Thus,
\begin{align}
&P_{e,\na}\hspace{-0.25em} =
\frac{1}{N_{\one}\cdot N_{\two}} \sum_{(\mo, \mt)\in \mathcal{M}_{\one}\times\mathcal{M}_{\two}}W^n\inp{\cE^{\one}_{\mo}\cup\cE^{\two}_{\mt}|f_{\one}(\mo), f_{\two}(\mt)}\nonumber\\
&\leq\frac{1}{N_{\one}\cdot N_{\two}} \sum_{(\mo, \mt)\in \mathcal{M}_{\one}\times\mathcal{M}_{\two}}\Big(W^n\inp{\cE^{\one}_{\mo}|f_{\one}(\mo), f_{\two}(\mt)}+W^n\inp{\cE^{\two}_{\mt}|f_{\one}(\mo), f_{\two}(\mt)}\Big)\nonumber\\
&=\frac{1}{ N_{\two}}\sum_{\mt\in\cM_{\two}}\inp{\frac{1}{N_{\one}} \sum_{\mo\in \mathcal{M}_{\one}}W^n\inp{\cE^{\one}_{\mo}|f_{\one}(\mo), f_{\two}(\mt)}}\nonumber\\
&\qquad \qquad+\frac{1}{N_{\one}}\sum_{\mo\in\cM_{\one}}\inp{ \frac{1}{N_{\two}}\sum_{\mt\in\mathcal{M}_{\two}}W^n\inp{\cE^{\two}_{\mt}|f_{\one}(\mo), f_{\two}(\mt)}}\nonumber\\
&\leq P_{e,\malone} +P_{e,\maltwo}.\label{honest_error_ub}
\end{align}	
So, if $P_{e,\malone}$ and $P_{e,\maltwo}$ are small, $P_{e,\na}$ is also small.
\begin{remark}\label{remark:det_rand_attacks}
Note that the probability of error under a randomized attack \blue{(private randomization at the adversary)} is the weighted average of the probabilities of errors under the different deterministic attacks and hence maximized by a deterministic attack. Thus, $P_{e,\maltwo}$ is an upper bound on the probability of error for any attack by user \two, deterministic or random. 
Similarly, $P_{e,\malone}$ is an upper bound for any attack by user \one.  Thus, the probability of error under deterministic attacks is same as that under randomized attacks. 
\end{remark}
\begin{defn}[Achievable rate pair and capacity region for communication with adversary identification]\label{defn:capacity}
 $(R_{\one},R_{\two})$ is an {\em achievable rate pair for communication with adversary identification} if there exists a sequence of $(\lfloor2^{nR_{\one}}\rfloor,\lfloor2^{nR_{\two}}\rfloor,n)$ adversary identifying codes $\{f_{\one}^{(n)},f_{\two}^{(n)},\phi^{(n)}\}_{n=1}^\infty$ such that $\lim_{n\rightarrow\infty}P_{e}(f_{\one}^{(n)},f_{\two}^{(n)},\phi^{(n)})=0.$ The {\em  capacity region of communication with adversary identification} $\cC$ is the closure of the set of all such achievable rate pairs. 
\end{defn}
{\begin{remark}
Note that the capacity region of a MAC where both users are honest, denoted by $\cC_{\MAC}$,  is an outer bound on the  capacity region of communication with adversary identification $\cC$, that is, $\cC\subseteq \cC_{\MAC}$. This is because when both users are honest, an adversary identifying code guarantees reliable decoding for both users.
\end{remark}}
\blue{In the next subsection, we give a brief summary of relevant results from the AV-MAC literature, which will be useful in the later part of the paper.}
\subsection{Arbitrarily Varying MAC}\label{sec:AVMAC}
\blue{\begin{defn}[see \cite{Jahn81}]\label{defn:AVMAC}
An AV-MAC $\cW = \{W(\cdot|\cdot,\cdot,s)| s\in \cS\}$ is a family of MACs parameterized by the set of state symbols $\cS$ where MACs $W(\cdot|\cdot,\cdot,s)$ are randomized  $\cX\times\cY\rightarrow\cZ$ maps. 
\end{defn}
The state of an AV-MAC can vary arbitrarily during the transmission. This can also be interpreted as an adversary choosing a distribution over the state symbols (and hence a channel from the convex hull of $\cW$) for each symbol sent into the channel by the senders. 
Two popular settings in which the problem is studied are:
\begin{enumerate}
        \item Randomized coding \cite{Jahn81}: Each sender shares individually private randomness with the  receiver which is hidden from the other sender(s) and the adversary. 
                \item Deterministic coding with average error criterion \cite{Gubner,AhlswedeC99}: The senders and receiver do not share any randomness. Each sender chooses a codeword uniformly at random using private randomness hidden from the adversary and sends the corresponding codeword.
\end{enumerate}
In both cases, because of the randomness (shared or private), the adversary is unaware of the actual symbols sent into the channel by each sender. Similarly, the senders are unaware of the choice of adversarial state vector. We refer the reader to \cite{Jahn81,Gubner,AhlswedeC99} for formal definitions of code, probability of error and capacity region.}

\blue{For distributions $P_X$, $P_Y$ on $\cX$ and $\cY$ respectively, define the set $\cR(P_X, P_Y)$ of all pairs $(R_1, R_2)$ satisfying
\begin{align}
0\leq R_1&\leq \inf I(X;Z|Y)\\
0\leq R_2&\leq \inf I(Y;Z|X)\\
R_1+R_2&\leq \inf I(XY;Z)
\end{align} where each infimum is taken over all $(X, Y, S, Z)\sim P_{XYSZ}$ such that $P_{XYSZ}(x, y, s, z) = P_{X}(x)P_Y(y)P_S(s)W(z|x,y,s)$. Let $\cR(\cW)$ be the closed convex hull $\bigcup_{(P_X, P_Y)}\cR(P_X, P_Y)$. 
For an AV-MAC $\cW$, let $\cC^{}_{\AVMAC}(\cW)$ denote the deterministic coding capacity region and $\cC^{\rand}_{\AVMAC}(\cW)$ denote the randomized coding capacity region. We next state the following result due to \cite{Jahn81}.
\begin{thm}[\cite{Jahn81}]
$\cC^{\rand}_{\AVMAC}(\cW) = \cR(\cW)$ and $\cC^{}_{\AVMAC}(\cW) = \cC^{\rand}_{\AVMAC}(\cW)$ if there exists $R_1, R_2\in \cC^{}_{\AVMAC}(\cW)$ such that $R_1, R_2>0$.
\end{thm}
Gubner~\cite{Gubner} gave three symmetrizability conditions for an AV-MAC (given below) and conjectured that these conditions characterize the feasibility of existence of $R_1, R_2\in \cC^{}_{\AVMAC}(\cW)$ such that $R_1, R_2>0$. This conjecture was later shown to be true by Ahlswede and Cai \cite{AhlswedeC99}.
The symmetrizability conditions for an AV-MAC $\cW$ are as follows:
\begin{enumerate}
        \item $\cW $ is symmetrizable-$(\cX, \cY)$ if there exists a stochastic $U:\cX\times\cY\rightarrow \cS$ such that
        \begin{align}\label{eq:sym_xy}
        \sum_s W(z|x, y, z)U(s|x',y')=W(z|x', y', z)U(s|x,y)
        \end{align} for all $x,x'\in \cX,  y,y'\in \cY$ and $z\in \cZ$.
        \item $\cW $ is symmetrizable-$\cX$ if there exists a stochastic $U:\cX\rightarrow \cS$ such that
        \begin{align}\label{eq:sym_x}
        \sum_s W(z|x, y, z)U(s|x')=W(z|x', y, z)U(s|x)
        \end{align} for all $x,x'\in \cX,  y\in \cY$ and $z\in \cZ$.
        \item $\cW $ is symmetrizable-$\cY$ if there exists a stochastic $U:\cY\rightarrow \cS$ such that
        \begin{align}\label{eq:sym_y}
        \sum_s W(z|x, y, z)U(s|y')=W(z|x, y', z)U(s|y)
        \end{align} for all $x\in \cX,  y,y'\in \cY$ and $z\in \cZ$.
\end{enumerate}
\begin{thm}[{\cite[Lemmas 3.5 and 3.6]{Gubner}}]\label{thm:gubner_sym}
If $\cW$ is symmetrizable-$\cX$, then $R_1 = 0$ for all $R_1, R_2\in \cC^{}_{\AVMAC}(\cW)$. Similarly, if $\cW$ is symmetrizable-$\cY$, then $R_2 = 0$ for all $R_1, R_2\in \cC^{}_{\AVMAC}(\cW)$.
\end{thm}
\begin{remark}\label{remark:gubner}
The proofs of \cite[Lemmas 3.5 and 3.6]{Gubner} are given in \cite[Lemmas~3.10 and 3.11]{Gubnerthesis}. Following the proofs, it can be shown that for any code with $R_1>0$ and $R_2>0$, the probability of error is at least $1/4$ if the channel is either symmetrizable-$\cX$ or  symmetrizable-$\cY$.
\end{remark}
\begin{thm}[{\cite[Theorem 1]{AhlswedeC99}}]
For an AV-MAC $\cW$, there exist $R_1, R_2\in \cC^{}_{\AVMAC}(\cW)$ such that $R_1, R_2>0$ if $\cW$ is not symmetrizable-$\cX$, not symmetrizable-$\cY$ and not symmetrizable-$(\cX, \cY)$.
\end{thm}}

\section{Feasibility of communication with adversary identification}\label{sec:feasibility}
\begin{defn}[Spoofable\ \bmac]\label{defn:spoof}
A \bmac \mch is \one-{\em \spoofable} if there exist conditional distributions ${Q_{Y|\tilde{X}\tilde{Y}}}$ and ${Q_{X|\tilde{X}X'}}$ such that $\forall\,x', \tilde{x}\in \cX,\, \tilde{y}\in \cY,\, z\in \cZ,$
\begin{align}\label{eq:spoof1}
&\sum_{y}Q_{Y|\tilde{X}\tilde{Y}}(y|\tilde{x},\tilde{y})\mch(z|x',y) \nonumber\\
&= \sum_{y}Q_{Y|\tilde{X}\tilde{Y}}(y|x',\tilde{y})\mch(z|\tilde{x},y) \nonumber\\
& = \sum_{x}Q_{X|\tilde{X}X'}(x|\tilde{x},x')\mch(z|x,{\tilde{y}}).
\end{align}

A \bmac \mch is \two-{\em \spoofable} \longonly{({see Fig.~\ref{fig:spoof2}.})} if there exist conditional distributions ${Q_{X|\tilde{X}\tilde{Y}}}$ and ${Q_{Y|\tilde{Y}Y'}}$ such that $\forall\,\tilde{x}\in \cX,\, \tilde{y}, y'\in \cY, \, z\in \cZ,$
\begin{align}\label{eq:spoof2}
&\sum_{x}Q_{X|\tilde{X}\tilde{Y}}(x|\tilde{x},\tilde{y})\mch(z|x,y') \nonumber\\
&= \sum_{x}Q_{X|\tilde{X}\tilde{Y}}(x|\tilde{x},y')\mch(z|x,\tilde{y}) \nonumber\\
& = \sum_{y}Q_{Y|\tilde{Y}Y'}(y|\tilde{y},y')\mch(z|\tilde{x},y).
\end{align}
A \bmac is {\em \spoofable} if it is either \one-\spoofable or \two-\spoofable. 

\end{defn}
	
\begin{figure*}[h]
\centering
\subfigure[]{
\begin{tikzpicture}[scale=0.5]
	\draw (1.7-0.8-0.3,2.9-0.1) rectangle ++(2.3,1.2) node[pos=.5]{\footnotesize ${Q^n_{Y|\tilde{X}\tilde{Y}}}$};
	\draw (4.1,4) rectangle ++(1.5,1.5) node[pos=.5]{ $W^n$};	
	\draw[ ->] (0.4-0.3,3.0) node[anchor=east]{ $\tilde{\vecy}\phantom{'}$} -- ++ (0.5,0) ;
	\draw[ ->] (0.4-0.3,3.7+0.1) node[anchor=east]{ $\tilde{\vecx}\phantom{'}$} -- ++ (0.5,0) ;
	\draw[->] (3.4,5.1) -- ++(0.7,0);
	\draw[->] (3.4,4.4) -- ++(0.7,0);
	\draw[-] (3.4,4.4) -- ++ (0,-1);
	\draw[-] (3.4,5.1) -- ++ (0,1);
	\draw[-] (2.9, 3.4)  --  ++(0.5,0);
	\draw[-] (0.4-0.3,6.1) node[anchor=east]{ $\vecx'$} -- ++ (3.3,0);
	\draw[->] (5.4+1-0.3-0.5,4.75) --  ++ (0.5,0)node[anchor= west] { $\vecz$};
	\end{tikzpicture}\label{fig:spoof1a} 
}
\subfigure[]{
\begin{tikzpicture}[scale=0.5]
	\draw (1.7-0.8-0.3,2.9-0.1) rectangle ++(2.3,1.2) node[pos=.5]{\footnotesize ${Q^n_{Y|\tilde{X}\tilde{Y}}}$};
	\draw (4.1,4) rectangle ++(1.5,1.5) node[pos=.5]{ $W^n$};	
	\draw[ ->] (0.4-0.3,3.0) node[anchor=east]{ $\tilde{\vecy}\phantom{'}$} -- ++ (0.5,0) ;
	\draw[ ->] (0.4-0.3,3.7+0.1) node[anchor=east]{ $\vecx'$} -- ++ (0.5,0) ;
	\draw[->] (3.4,5.1) -- ++(0.7,0);
	\draw[->] (3.4,4.4) -- ++(0.7,0);
	\draw[-] (3.4,4.4) -- ++ (0,-1);
	\draw[-] (3.4,5.1) -- ++ (0,1);
	\draw[-] (2.9, 3.4)  --  ++(0.5,0);
	\draw[-] (0.4-0.3,6.1) node[anchor=east]{ $\tilde{\vecx}\phantom{'}$} -- ++ (3.3,0);
	\draw[->] (5.4+1-0.3-0.5,4.75) --  ++ (0.5,0)node[anchor= west] { $\vecz$}; 
\end{tikzpicture}\label{fig:spoof1b} 
}
\subfigure[]{
\begin{tikzpicture}[scale=0.5]
	\draw (1.7-1.1,5.6-0.1) rectangle ++(2.5,1.2) node[pos=.5]{\footnotesize ${Q^n_{X|\tilde{X}X'}}$};
	\draw (4.1,4) rectangle ++(1.5,1.5) node[pos=.5]{ $W^n$};	
	\draw[ ->] (0.1,5.8-0.1) node[anchor=east]{ ${\vecx'}$} -- ++ (0.5,0) ;
	\draw[ ->] (0.1,6.4+0.1) node[anchor=east]{ $\tilde{\vecx}\phantom{'}$} -- ++ (0.5,0) ;
	\draw[->] (3.4,5.1) -- ++(0.7,0);
	\draw[->] (3.4,4.4) -- ++(0.7,0);
	\draw[-] (3.4,4.4) -- ++ (0,-1);
	\draw[-] (3.4,5.1) -- ++ (0,1);
	\draw[-] (3.1, 6.1)  --  ++(0.3,0);
	\draw[-] (0.1,3.4) node[anchor=east]{ $\tilde{\vecy}\phantom{'}$} -- ++ (3.3,0);
	\draw[->] (5.4+1-0.3-0.5,4.75) --  ++ (0.5,0)node[anchor= west] { $\vecz$};
\end{tikzpicture}\label{fig:spoof1c} 
}\caption{When \eqref{eq:spoof1} holds for a \bmac $W$, for $(\vecx', \tilde{\vecx}, \tilde{\vecy})\in \cX^n\times\cX^n\times\cY^n$, the output distributions in the three cases above will be the same.} \label{fig:spoof1}
\end{figure*}

\longonly{
\begin{figure*}[h]
\centering
\subfigure[]{
\begin{tikzpicture}[scale=0.5]
	\draw (1.7-0.8-0.3,5.6) rectangle ++(2.3,1) node[pos=.5]{\footnotesize ${Q_{X|\tilde{X}\tilde{Y}}}$};
	\draw (4.1,4) rectangle ++(1.5,1.5) node[pos=.5]{ $W$};	
	\draw[ ->] (0.1,5.8) node[anchor=east]{ $\tilde{y}\phantom{'}$} -- ++ (0.5,0) ;
	\draw[ ->] (0.1,6.4) node[anchor=east]{ $\tilde{x}\phantom{'}$} -- ++ (0.5,0) ;
	\draw[->] (3.4,5.1) -- ++(0.7,0);
	\draw[->] (3.4,4.4) -- ++(0.7,0);
	\draw[-] (3.4,4.4) -- ++ (0,-1);
	\draw[-] (3.4,5.1) -- ++ (0,1);
	\draw[-] (2.9, 6.1)  --  ++(0.5,0);
	\draw[-] (0.1,3.4) node[anchor=east]{ $y'$} -- ++ (3.3,0);
	\draw[->] (5.4+1-0.3-0.5,4.75) --  ++ (0.5,0)node[anchor= west] { $z$};
\end{tikzpicture}\label{fig:spoof2a} 
}
\subfigure[]{
\begin{tikzpicture}[scale=0.5]
	\draw (1.7-0.8-0.3,5.6) rectangle ++(2.3,1) node[pos=.5]{\footnotesize ${Q_{X|\tilde{X}\tilde{Y}}}$};
	\draw (4.1,4) rectangle ++(1.5,1.5) node[pos=.5]{ $W$};	
	\draw[ ->] (0.1,5.8) node[anchor=east]{ ${y'}$} -- ++ (0.5,0) ;
	\draw[ ->] (0.1,6.4) node[anchor=east]{ $\tilde{x}\phantom{'}$} -- ++ (0.5,0) ;
	\draw[->] (3.4,5.1) -- ++(0.7,0);
	\draw[->] (3.4,4.4) -- ++(0.7,0);
	\draw[-] (3.4,4.4) -- ++ (0,-1);
	\draw[-] (3.4,5.1) -- ++ (0,1);
	\draw[-] (2.9, 6.1)  --  ++(0.5,0);
	\draw[-] (0.1,3.4) node[anchor=east]{ $\tilde{y}$} -- ++ (3.3,0);
	\draw[->] (5.4+1-0.3-0.5,4.75) --  ++ (0.5,0)node[anchor= west] { $z$};
\end{tikzpicture}\label{fig:spoof2b} 
}
\subfigure[]{
\begin{tikzpicture}[scale=0.5]
	\draw (1.7-1.1,2.9) rectangle ++(2.5,1) node[pos=.5]{\footnotesize ${Q_{Y|\tilde{Y}Y'}}$};
	\draw (4.1,4) rectangle ++(1.5,1.5) node[pos=.5]{ $W$};	
	\draw[ ->] (0.1,3.1) node[anchor=east]{ ${y'}$} -- ++ (0.5,0) ;
	\draw[ ->] (0.1,3.7) node[anchor=east]{ $\tilde{y}\phantom{'}$} -- ++ (0.5,0) ;
	\draw[->] (3.4,5.1) -- ++(0.7,0);
	\draw[->] (3.4,4.4) -- ++(0.7,0);
	\draw[-] (3.4,4.4) -- ++ (0,-1);
	\draw[-] (3.4,5.1) -- ++ (0,1);
	\draw[-] (3.1, 3.4)  --  ++(0.3,0);
	\draw[-] (0.1,6.1) node[anchor=east]{ $\tilde{x}\phantom{'}$} -- ++ (3.3,0);
	\draw[->] (5.4+1-0.3-0.5,4.75) --  ++ (0.5,0)node[anchor= west] { $z$}; 
\end{tikzpicture}\label{fig:spoof2c} 
}
\caption{A \bmac $W$ is \two-\spoofable if for each $\tilde{x},\, \tilde{y},\, y',  \,z$ the conditional output distributions $P(z|\tilde{x},\tilde{y},y')$ in \ref{fig:spoof2a}, \ref{fig:spoof2b} and \ref{fig:spoof2c} are the same.} \label{fig:spoof2}
\end{figure*}
}

Our definition of \spoofable channels is motivated by a scenario where the decoder cannot differentiate between two different likely transmitted codewords of user \one, while at the same time it cannot blame user \two\ for the situation since the situation appears to be possible due to an adversarial action of either user \one\ or user \two. To see this, let $(\vecx', \tilde{\vecx}, \tilde{\vecy})\in \cX^n\times\cX^n\times\cY^n$. When a \bmac is \one-\spoofable, {\em i.e.}
\eqref{eq:spoof1} holds, the output distributions in the following three cases are the same (see Fig.~\ref{fig:spoof1}): 
\begin{enumerate}[label = (\alph*)]
\item User \one sends $\vecx'$ and an adversarial user \two sends $\vecY\sim Q^n_{Y|\tilde{X}\tilde{Y}}(\cdot|\tilde{\vecx},\tilde{\vecy})$, i.e., $\vecY$ is distributed as the output of the memoryless channel $Q_{Y|\tilde{X}\tilde{Y}}$ on inputs $\tilde{\vecx}$ and $\tilde{\vecy}$; 
\item User \one sends $\tilde{\vecx}$ and an adversarial user \two sends $\vecY\sim Q^n_{Y|\tilde{X}\tilde{Y}}(\cdot|\vecx',\tilde{\vecy})$; 
\item User \two sends $\tilde{\vecy}$ and an adversarial user \one sends $\vecX\sim Q^n_{X|\tilde{X}X'}(\cdot|\tilde{\vecx},\vecx')$.
\end{enumerate}
In Lemma~\ref{thm:converse}~(Appendix~\ref{sec:proof_thm1}), we use the above property of spoofable channels to show that for an \one-\spoofable \bmac, user-\one cannot send even one bit reliably. Similarly, for a \two-\spoofable \bmac, user-\two cannot send one bit reliably.\footnote{In fact, we show the result for the case when users are allowed to use privately randomized encoders.} For an \one-\spoofable channel, our proof considers any given code $(f_{\one}, f_{\two}, \phi)$ and messages and messages $M_{\one}$, $M'_{\one}$ and $M_{\two}$ which are independent and uniformly distributed over their alphabets. By replacing $(\vecx', \tilde{\vecx}, \tilde{\vecy})$ with $(f_{\one}(M_{\one}), f_{\one}(M_{\one}'),f_{\two}(M_{\two}))$ in the above argument, we note that the output distributions in the following three cases are the same: 
\begin{enumerate}[label=(\alph*)]
\item User \one is honest and sends $f_{\one}(M_{\one})$ and user \two is adversarial and attacks with $\vecY\sim Q^n_{Y|\tilde{X}\tilde{Y}}(\cdot|f_{\one}(M_{\one}'),f_{\two}(M_{\two}))$; 
\item  User \one is honest and sends $f_{\one}(M'_{\one})$ and user \two is adversarial and attacks with $\vecY\sim Q^n_{Y|\tilde{X}\tilde{Y}}(\cdot|f_{\one}(M_{\one}),f_{\two}(M_{\two})$; 
\item User \two is honest and sends $f_{\two}(M_{\two})$ and user \one is adversarial and attacks with $\vecX\sim Q^n_{X|\tilde{X}X'}(\cdot|f_{\one}(M_{\one}),f_{\one}(M_{\one}'))$. 
\end{enumerate}
Thus, the decoder cannot determine the adversarial user reliably, nor can it differentiate between $M_{\one}$ and $M'_{\one}$ as the input of user \one.

Our first result states that non-\spoofability characterizes the \bmacs in which users can work at positive rates of communication with adversary identification.  
\begin{thm}\label{thm:main_result}
If a \bmac is \one-\spoofable (resp. \two-\spoofable), communication with adversary identification from user-\one (resp. user-\two) is impossible. Specifically, for any $(N_{\one}, N_{\two}, n)$ adversary identifying code with $N_{\one}\geq 2$ (resp. $N_{\two}\geq 2$), the probability of error is at least $1/12$.  If a \bmac is neither \one-\spoofable nor \two-\spoofable, then its capacity region has a non-empty interior ($\mathsf{int}(\cC)\neq \emptyset$), {\em i.e.}, both users can communicate reliably with adversary identification at positive rates. 
\end{thm}
\begin{corollary}
$\mathsf{int}(\cC) = \emptyset$ if and only if a \bmac is \spoofable.
\end{corollary}
\begin{remark}\label{remark:det_comm_one_spoofable}
When exactly one user is \spoofable, Theorem~\ref{thm:main_result} leaves open the question whether the other user can communicate reliably. A similar case is also open for Arbitrarily Varying Multiple Access MAC (AV-MAC) (see \cite{AhlswedeC99}).\footnote{Pereg and Steinberg \cite{PeregS19} considered this when encoders have private randomness, a setting we do not address in this paper.} 
\end{remark}

\begin{figure*}[!h]
\centering
\subfigure[]{
\begin{tikzpicture}[scale=0.4]
	\draw[] (4.1+4,4-6) rectangle ++(1.3,1.5) node[pos=.5]{\footnotesize $W$};
	\draw[->] (3.4+4,5.25-6) -- ++(0.7,0);
	\draw[red,->] (3.4+4,4.25-6) -- ++(0.7,0);
	\draw[red,-] (3.4+4,4.25-6) -- ++ (0,-1);
	\draw[-] (3.4+4,5.25-6) -- ++ (0,1);
	\draw[red,-] (2.4+4, 3.25-6) node[left]{\scriptsize \textcolor{red}{${\vecy}$}} --  ++(1,0) ;
	\draw[-] (2.4+4,6.25-6) node[anchor=east]{\scriptsize $f_{\one}({m}_{\one})$} -- ++ (1,0);
	\draw[->] (5.4+4,4.75-6) --  ++ (0.5,0)node[anchor= west] {\scriptsize $\vecz$};
\end{tikzpicture}
\label{fig:dec2a} 
}
\subfigure[]{
\begin{tikzpicture}[scale=0.4]
	\draw[] (4.1-4,4-6) rectangle ++(1.3,1.5) node[pos=.5]{\footnotesize $W$};
	\draw[red,->] (3.4-4,5.25-6) -- ++(0.7,0);
	\draw[->] (3.4-4,4.25-6) -- ++(0.7,0);
	\draw[-] (3.4-4,4.25-6) -- ++ (0,-1);
	\draw[red,-] (3.4-4,5.25-6) -- ++ (0,1);
	\draw[-] (2.4-4, 3.25-6) node[left]{\scriptsize $f_{\two}(\tilde{m}_{\two})$} --  ++(1,0) ;
	\draw[red, -] (2.4-4,6.25-6) node[anchor=east]{\scriptsize \textcolor{red}{$\vecx$}} -- ++ (1,0);
	\draw[->] (5.4-4,4.75-6) --  ++ (0.5,0)node[anchor= west] {\scriptsize $\vecz$};
\end{tikzpicture}
\label{fig:dec2bb} 
}
\subfigure[]{
\begin{tikzpicture}[scale=0.4]
	\draw[] (4.1+4,4-6) rectangle ++(1.3,1.5) node[pos=.5]{\footnotesize $W$};
	\draw[->] (3.4+4,5.25-6) -- ++(0.7,0);
	\draw[red,->] (3.4+4,4.25-6) -- ++(0.7,0);
	\draw[red,-] (3.4+4,4.25-6) -- ++ (0,-1);
	\draw[-] (3.4+4,5.25-6) -- ++ (0,1);
	\draw[red,-] (2.4+4, 3.25-6) node[left]{\scriptsize \textcolor{red}{$\tilde{\vecy}$}} --  ++(1,0) ;
	\draw[-] (2.4+4,6.25-6) node[anchor=east]{\scriptsize $f_{\one}(\tilde{m}_{\one})$} -- ++ (1,0);
	\draw[->] (5.4+4,4.75-6) --  ++ (0.5,0)node[anchor= west] {\scriptsize $\vecz$};
\end{tikzpicture}
\label{fig:dec2c} 
}
\subfigure[]{
\begin{tikzpicture}[scale=0.4]
\draw[teal] (3, 8.5) node{\scriptsize \fbox{\textcolor{black}{$\red{\tilde{X}\tilde{Y}}-\red{Y}-XZ$}}};
	\draw (4.1,4) rectangle ++(1.3,1.5) node[pos=.5]{\footnotesize $W$};
	\draw[->] (3.4,5.25) -- ++(0.7,0);
	\draw[red,->] (3.4,4.25) -- ++(0.7,0);
	\draw[red,-] (3.4,4.25) -- ++ (0,-1);      
	\draw[-] (3.4,5.25) -- ++ (0,1);
	\draw[red,-] (2.4, 3.25)  --  ++(1,0) node[right]{\scriptsize $Y \,\,\textcolor{orange}{\boxed{\vecy}} $};

	\draw[red] (1.4, 3.25-0.6) rectangle ++(1,1.2);
	\draw[red,->] (1.4-0.7, 3.25-0.4) node[anchor=east]{\tiny $\tilde{Y}$} -- ++ (0.7,0);	

	\draw[orange] (-0.4, 3.25+0.7) node[anchor=east]{\tiny $\textcolor{orange}{\boxed{\tiny f_{\one}(\tilde{m}_{\one})}}$};
	\draw[orange] (-0.4, 3.25-0.7) node[anchor=east]{\tiny $\textcolor{orange}{\boxed{\tiny f_{\two}(\tilde{m}_{\two})}}$};

	\draw[red,->] (1.4-0.7, 3.25+0.4) node[anchor=east]{\tiny {$\tilde{X}$}} -- ++ (0.7,0);
	\draw[-] (2.4,6.25) node[anchor=east]{\scriptsize $\textcolor{gray}{\boxed{{f_{\one}(m_{\one})}}}\,\, X$} -- ++ (1,0);
	\draw[->] (5.4,4.75) --  ++ (0.5,0)node[anchor= west] {\scriptsize $Z\,\,\textcolor{gray}{\boxed{\vecz}}$};
\end{tikzpicture}
\label{fig:dec2b} 
}
\caption{The channel output $\vecz$ is such that the tuples $(f_{\one}(m_{\one}), \vecy, \vecz)$, $(\vecx, f_{\two}(\tilde{m}_{\two}), \vecz)$ and $(f_{\one}(\tilde{m}_{\one}), \tilde{\vecy}, \vecz)$ are consistent according to channel law for some $\vecy$, $\vecx$ and $\tilde{\vecy}$ as shown in (a), (b) and (c) respectively. Suppose 
	$(f_{\one}(m_{\one}), {\vecy}, \vecz, {f_{\one}(\tilde{m}_{\one})}, {f_{\two}(\tilde{m}_{\two})})\in T^n_{X{Y}Z{\tilde{X}}{\tilde{Y}}}$ is such that $I({\tilde{X}}{\tilde{Y}};XZ|{Y})$ is small ({\em i.e.} the Markov chain $\tilde{X}\tilde{Y}-{Y}-XZ$ approximately holds). Then, subfigure (d) is a plausible explanation where user \one is honest with input message $\mo$, and  $\tilde{m}_{\one}$ and $\tilde{m}_{\two}$ can be explained by an attack strategy of user \two (compare subfigures (d) and (a)).} \label{fig:dec2}
\end{figure*}

\blue{The proof of Theorem~\ref{thm:main_result} is given in Appendix \ref{sec:proof_thm1}. Here, we provide an informal description of the decoder.} For input distributions $P_{\one}$ and $P_{\two}$ on \cX\ and \cY\ respectively, the decoder works by first collecting potential candidates for the messages sent by each user. 
A message $\mo$ is deemed a {\em candidate} for user \one if it is typical with some (attack) vector $\vecy$ and the output vector $\vecz$ according to the channel law (i.e., for some small $\eta>0$, $\inp{f_{\one}(m_{\one}), \vecy, \vecz} \in T^{n}_{XYZ}$ such that $D\inp{P_{XYZ}||P_{\one}P_YW}\leq\eta$). 
We further prune the list of candidates by only keeping the ones which can account for all other candidates that can lead to ambiguity at the decoder. 
Suppose there are two candidates $\mo\neq\tilde{m}_{\one}$ for user \one and one candidate $\tilde{m}_{\two}$ for user \two. The decoder is confused between \mo and $\tilde{m}_{\one}$, so it cannot reliably choose an output message for user \one. Neither can it adjudge one of the users to be adversarial as both users have valid message candidates. 
In order to get around this, we consider a message candidate $\mo$ {\em viable} only if  for  every pair of candidates ($\tilde{m}_{\one},\tilde{m}_{\two}$), $\tilde{m}_{\one}\neq \mo$, such that $\inp{f_{\one}(\mo), \vecy,  f_{\one}(\tilde{m}_{\one}), f_{\two}(\tilde{m}_{\two}), \vecz}$$\in $$T^n_{XY\tilde{X}\tilde{Y}Z}$, the condition  $I(\tilde{X}\tilde{Y};XZ|Y)<\eta$ holds.   
Under this condition, there is a plausible explanation where an honest user \one sends $\mo$ and  the pair $(\tilde{m}_{\one}$, $\tilde{m}_{\two}$) is part of user \two's attack strategy. See Fig.~\ref{fig:dec2} for details.  
Similarly, if there  is a pair of candidates $(\tilde{m}_{\two 1},\tilde{m}_{\two 2})$ of user \two,  the decoder can neither reliably decode user \two's message, nor can it declare either user as adversarial. 
Then, we require that for every pair of candidates ($\tilde{m}_{\two 1}$,$\tilde{m}_{\two 2}$) of user \two such that $\inp{f_{\one}(\mo), \vecy, f_{\two}(\tilde{m}_{\two 1}), f_{\two}(\tilde{m}_{\two 2}),\vecz}$$\in $$T^n_{XY\tilde{Y}_1\tilde{Y}_2Z}$, the condition $I(\tilde{Y}_1\tilde{Y}_2;XZ|Y)< \eta$ holds. Let $D_{\one}(\eta, \vecz)$ be the set of all candidates of user \one which pass these checks. 
{We define $D_{\two}(\eta, \vecz)$ analogously by interchanging the roles of users \one and \two.}
Finally, the decoder outputs as follows:
\begin{align*}
\phi(\vecz) \defineqq\begin{cases}(\mo,\mt) &\text{if }D_{\one}(\eta, \vecz)\times D_{\two}(\eta, \vecz) = \{(\mo,\mt)\},\\ \oneb \text{ (\small blame \one)} &\text{if }|D_{\one}(\eta, \vecz)| = 0, \, |D_{\two}(\eta, \vecz)| \neq 0,\\ \twob\text{ (\small blame \two)} &\text{if }|D_{\two}(\eta, \vecz)| = 0, \, |D_{\one}(\eta, \vecz)| \neq 0,\\(1,1) &\text{if }|D_{\one}(\eta, \vecz)| = |D_{\two}(\eta, \vecz)| = 0.
\end{cases}
\end{align*}In the spirit of \cite[Lemma 4]{CsiszarN88} and \cite[Lemma 1]{AhlswedeC99}, we show 
that for a non-\spoofable \ \bmac, there exists a  small enough $\eta>0$ such that if $|D_{\one}(\eta, \vecz)|, |D_{\two}(\eta, \vecz)| >0$ then $|D_{\one}(\eta, \vecz)|$ = $|D_{\two}(\eta, \vecz)|$ = 1 (Lemma~\ref{lemma:disambiguity}). Thus, the decoder definition covers all the cases. \blue{When the input is not $(\mo, \mt) = (1, 1)$ (which happens with atleast $\min\inb{(1-1/N_{\one}), (1-1/N_{\two})}$ probability), the output $(1, 1)$ is an error event. Thus, we show }that $|D_{\one}(\eta, \vecz)|$ = $|D_{\two}(\eta, \vecz)| = 0$ is a low probability event.  
By analyzing the error probability of the decoder we show that for non-\spoofable channels it can support positive rates for both users.

\section{Comparison with related models}\label{sec:comparison}
\olive{In this section we contrast the present model with reliable communication and authenticated communication models \blue{and provide examples to show separation.}}

\olive{\paragraph{Reliable communication in a \bmac}\label{para:reliable} 
We consider a \bmac with a stronger decoding guarantee: the decoder, w.h.p, outputs a message pair of which the message(s) of honest user(s) is correct. 
In the presence of a malicious user, the channel from the honest user to the receiver can be treated as an Arbitrarily Varying Channel (AVC) \cite{BBT60} with the input of other user as state. Thus, the capacity region is outer bounded by the rectangular region defined by the AVC capacities of the two users’ channels. Further, it is easy to see that this outer bound is achievable when both users use the corresponding AVC codes.
Csisz\'ar and Narayan show in \cite{CsiszarN88} that the capacity of an AVC is zero if and only if it is {\em symmetrizable}.
\shortonly{We continue the discussion from the introduction (picking up from footnote~\ref{ftn:reliable} in page~\pageref{ftn:reliable}).}  
Translating this to the two-user \bmac, we define a \bmac to be {\em \two-symmetrizable} if there exists a distribution $P_{X|Y}$ such that
\begin{align}\label{eq:symmetrizable}
\sum_{x'\in \cX}P_{X|Y}(x|y')W(z|x,y) = \sum_{x'\in \cX}P_{X|Y}(x|y)W(z|x,y')
\end{align}
for all $(x, y, z)\in \cX\times\cY\times\cZ$. We define an \one-symmetrizable \bmac analogously. A \textbf{symmetrizable} \bmac is one which is either \one- or \two-symmetrizable. 
Thus, reliable communication by both users is feasible in a \bmac if and only if it is not symmetrizable. We denote the reliable communication capacity of a \bmac by $\cC_{\reliable}$.
\paragraph{Authenticated communication in a \bmac \cite{NehaBDPISIT19}}\label{para:weak} 
\longonly{This model considers a \bmac with a weaker decoding guarantee: the decoder should reliably decode the messages when both users are honest. When one user is adversarial, the decoder either outputs a pair of messages of which the message of honest user is correct or it declares the presence of an adversary (without identifying it). 
In this case, the notion of an overwritable \bmac characterizes the class of channels with non-empty capacity region $\cC_{\auth}$ of authenticated communication.} We say that a \bmac is {\em \two-overwritable} \cite[(1)]{NehaBDPISIT19} if there exists a distribution $P_{X'|XY}$ such that
\begin{align}\label{eq:overwritable}
\sum_{x'\in\cX}P_{X'|XY}(x'|x,y)W(z|x',y') = W(z|x,y) 
\end{align}
for all $y,y'\in\cY, \, x\in \cX$ and $z\in \cZ$. Similarly, we can define an \one-overwritable \bmac.  If a \bmac is either \one-~or \two-overwritable, we say that the \bmac is \textbf{overwritable}. Authenticated communication by both users is not feasible in an overwritable \bmac. Theorem~1 in \cite{NehaBDPISIT19} states that if the \bmac is not overwritable, then authenticated communication capacity is the same as the non adversarial capacity of the MAC, {\em i.e.}, $\cC_{\auth} = \cC_{\MAC}$.}

\begin{prop}\label{prop:inclusions}
\olive{All overwritable \bmacs are \spoofable and all \spoofable\ \bmacs are symmetrizable. Furthermore, both these inclusions are strict.}
\end{prop}
\olive{\noindent While the inclusions in Proposition~\ref{prop:inclusions} are obvious from the problem definitions and the feasibility results, we nonetheless provide a direct argument. Suppose a \bmac is \two-overwritable with $P_{X'|XY}$ as the overwriting attack in \eqref{eq:overwritable}. For any distribution $Q_{Y}$ on \cY, let $Q_{X|\tilde{X} \tilde{Y}}(x|\tilde{x}, \tilde{y})\defineqq \sum_{y}Q_{Y}(y)P_{X'|XY}(x|\tilde{x},y)$ for all $x, \tilde{x}, \tilde{y}$ and $Q_{Y|\tilde{Y}Y'}(y|\tilde{y},y')\defineqq Q_{Y}(y)$ for all $y,\tilde{y},y'$. Distributions $Q_{X|\tilde{X} \tilde{Y}}$ and $Q_{Y|\tilde{Y}Y'}$ as defined satisfy \eqref{eq:spoof2}. Now, suppose a \bmac \mch is \two-\spoofable with attacks $Q_{X|\tilde{X} \tilde{Y}}$ and $Q_{Y|\tilde{Y}Y'}$ satisfying \eqref{eq:spoof2}. For all, $x, y$, let $P_{X|Y}(x|y)\defineqq Q_{X|\tilde{X} \tilde{Y}}(x|\tilde{x},y)$ for any $\tilde{x}\in \cX$. It can be easily seen that the attack $P_{X|Y}$ as defined satisfies \eqref{eq:symmetrizable}.
Examples~\ref{ex:1} and~\ref{ex:2} below show strict inclusion (see Fig. \ref{fig:conditions}).}
\olive{\begin{Example}[symmetrizable, but not \spoofable] Binary erasure MAC:  
{\em It has binary inputs $X, Y$ and outputs $Z=X+Y$ where $+$ is real addition, i.e., $\cZ = \{0,1,2\}$.}  \end{Example}}
\olive{\noindent To show symmetrizability, we note that the distribution $P_{X|Y}(x|y) = 1$ for all $x=y$ is a symmetrizing attack in \eqref{eq:symmetrizable}. 
Now, suppose the channel is  \one-{\em \spoofable}, that is, there exist distributions ${Q_{Y|\tilde{X}\tilde{Y}}}$ and ${Q_{X|\tilde{X}X'}}$ such that $\forall\,x', \,\tilde{x},\, \tilde{y},\, z,$
\begin{align*}
&\sum_{y}Q_{Y|\tilde{X}\tilde{Y}}(y|\tilde{x},\tilde{y})\mch(z|x',y) \nonumber\\
&= \sum_{y}Q_{Y|\tilde{X}\tilde{Y}}(y|x',\tilde{y})\mch(z|\tilde{x},y) \nonumber\\
& = \sum_{x}Q_{X|\tilde{X}X'}(x|\tilde{x},x')\mch(z|x,\tilde{y}).
\end{align*}
For $(x', \tilde{x}, \tilde{y}, z) = (1, 0, 1, 2)$, this gives $Q_{Y|\tilde{X}\tilde{Y}}(1|0,1) = 0 = Q_{X|\tilde{X}X'}(1|0,1)$ and for $(x', \tilde{x}, \tilde{y}, z) = (1, 0, 0, 0)$, we get $0 = Q_{Y|\tilde{X}\tilde{Y}}(0|1,0) = Q_{X|\tilde{X}X'}(0|0,1)$. However, $Q_{X|\tilde{X}X'}(1|0,1) = Q_{X|\tilde{X}X'}(0|0,1) = 0$ is not possible. Thus, the channel is not \one-\spoofable. Similarly, we can show that the channel is not \two-\spoofable. \blue{We obtain the capacity region of this MAC in Section~\ref{sec:example_tightness}. }}

\olive{\begin{example}[\spoofable, but not overwritable] \label{ex:2} Binary additive MAC:  {\em It has binary inputs $X, Y$ and outputs binary $Z = X\oplus Y$ where $\oplus$ is the XOR operation. }\end{example}}

\olive{\noindent To show \spoofability, note that the attacks ${Q_{X|\tilde{X}X'}}(x|\tilde{x}, x') = 1/2$ for all $x, \tilde{x}$ and $x'$, and ${Q_{Y|\tilde{X}\tilde{Y}}}(y|\tilde{x}, \tilde{y}) = 1/2$ for all $y, \tilde{x}$ and $\tilde{y}$, satisfy \eqref{eq:spoof1} because they result in the same uniform output distribution over $\cZ$ in all the three cases in \eqref{eq:spoof1}. 
Suppose binary additive MAC $Z = X\oplus Y$ is  \two-overwritable. Let $P_{X'|XY}$ be the overwriting attack by user \one which satisfies \eqref{eq:overwritable}. Then, for $(x, y, z) = (1, 1, 0)$ and all $y'$, \eqref{eq:overwritable} implies 
\begin{align*}{}
P_{X'|XY}(0|1,1)W(0|0,y')+P_{X'|XY}(1|1,1)W(0|1,y') = W(0|1,1) = 1.
\end{align*}
For $y' = 0$ and $1$, this implies that $P_{X'|XY}(0|1,1)=1$ and $P_{X'|XY}(1|1,1)=1$ respectively, which is not possible simultaneously. Thus, the channel cannot be \two-overwritable. Similarly, we can argue that the channel is not \one-overwritable. }
\vspace{-0.2 cm}
\begin{figure}[h]
\centering
\begin{tikzpicture}[scale=0.45]
\draw (0,0) ellipse (6cm and 4cm) node[yshift = -1.05 cm, xshift = -0.7 cm]{\footnotesize symmetrizable} node[yshift = -1.35 cm, xshift = -0.6 cm]{\footnotesize \bmacs} ;
\draw (0.2,0.6) ellipse (4cm and 2.5cm) node[yshift = -0.5 cm, xshift = 0 cm]{\footnotesize spoofable };
\draw (0.2,0.6) ellipse (4cm and 2.5cm) node[yshift = -0.8 cm, xshift = 0 cm]{\footnotesize  \bmacs};
\draw (0.3, 1) ellipse (2.5 cm and 1 cm) node[xshift = -0.1 cm, yshift = 0.2 cm]{\footnotesize overwritable} node[yshift = -0.1 cm]{\footnotesize \bmacs};

\node[label={[label distance=0.01mm]80:\footnotesize \ref{ex:1}}] (A) at (1,-3.1){};
\node[label={[label distance=0.01mm]80:\footnotesize \ref{ex:2}}] (B) at (2.9,-0.4){};
\node[label=below:] (B2) at (6.8,0)[right]{\footnotesize Binary additive MAC};
\node[label=below:] (A1) at (3.5,-2.8){};
\node[label=below:] (A2) at (5,-2.1){};
\node[label=below:] (A3) at (7,-2)[right]{\footnotesize Binary erasure MAC};
\node[label=below:] (C3) at (7.5,2.5){};
\node[label=below:] (C1) at (3,2){};
\node[label=below:] (C2) at (5,2.3){};

\draw[->] (B)--(B2);
\draw[->] (A) .. controls (A1) and (A2) .. (A3);
\draw [fill=black] (A) circle (1.5pt);
\draw [fill=black] (B) circle (1.5pt);
\end{tikzpicture}\caption{The set of overwritable \bmacs is a strict subset of the set of \spoofable\ \bmacs which, in turn, is a strict subset of the set of symmetrizable \bmacs.}\label{fig:conditions} \end{figure}
\vspace{0.3 cm}

\olive{We also note from the problem definitions that $\cC_{\reliable}\subseteq\cC\subseteq\cC_{\auth}\subseteq\cC_{\MAC}$.
Next, we give an example of a channel for which $\cC_{\reliable}$, $\cC$ and $\cC_{\auth}$ are distinct.
The example is constructed by using the \bmacs in Examples~\ref{ex:1} and \ref{ex:2} in parallel. 
\begin{example}[$(Z_1,Z_2)=(X_1+Y_1, X_2\oplus Y_2)$]\label{ex:3}
{\em For binary inputs $X_1, X_2, Y_1, Y_2$, the output is $(Z_1,Z_2) = (X_1+Y_1,X_2\oplus Y_2)$. See Fig.~\ref{fig:fig6}.}
 \end{example}}
 \begin{figure}
 \centering
 \begin{tikzpicture}[circuit ee IEC]
 \draw[->] (0,0) node[left]{$Y_2$} -- ++ (3,0);
 \draw[-] (0,0.9) node[left]{$Y_1$} -- ++ (2.5,0);
\node (a) at (0,0.9){}; 
\node (b) at (2.5,0.9){};
  \draw[-] (2.5,0.9) -- ++ (0,1.1);
  \draw[->] (2.5,2) -- ++ (0.5,0);
  \draw[-] (0,2) node[left]{$X_2$} -- ++ (2,0);
  \draw[-] (2,2) to [kinky cross=(a)--(b), kinky crosses=left] (2,0.5);
 \draw[->] (0,2.5) node[left]{$X_1$} -- ++ (3,0);
  \draw[->] (2, 0.5)  -- ++ (1, 0);
  \draw (3, -0.1) rectangle ++(0.7,0.7) node[pos=.5]{\large $\oplus$};
   \draw (3, 1.9) rectangle ++(0.7,0.7) node[pos=.5]{\large $+$};
   \draw [->] (3.7, 2.25) -- ++ (0.5,0) node[right]{$Z_1$}; 
   \draw [->] (3.7, 0.25) -- ++ (0.5,0) node[right]{$Z_2$}; 

	\end{tikzpicture}
	\caption{This figure depicts the channel in Example~\ref{ex:3}. User \one has input $(X_1, X_2)$ and user \two has input $(Y_1,Y_2)$. The output of the channel is $(Z_1, Z_2)$.}\label{fig:fig6}
 \end{figure}
\olive{\noindent The channels  $Z_1=X_1+Y_1$ and $Z_2=X_2\oplus Y_2$ are both non-overwritable and symmetrizable. Using the fact that the \bmacs do not interact when used in parallel, we will now show that the resultant \bmac $(Z_1,Z_2)=(X_1+Y_1, X_2\oplus Y_2)$ is also non-overwritable and symmetrizable. 
We will first show that this channel is \two-symmetrizable, that is, there exists distribution $P_{X|Y}$ such that
\begin{align*}
\sum_{x\in \cX}P_{X|Y}(x|y')W(z|x,y) = \sum_{x\in \cX}P_{X|Y}(x|y)W(z|x,y')
\end{align*}
for all $y',y,z$.
Consider $P_{X|Y}((x_1,x_2)|(y_1, y_2)) = 1$ when $(x_1,x_2)= (y_1, y_2)$. Then for $y' = (y_1', y_2')$, $y = (y_1, y_2)$ and $z = (y_1' + y_1, y_2'\oplus y_2)$, both the LHS and the RHS of the above equation evaluate to $1$, and for every other $z$, they evaluate to $0$. So, the channel is \two-symmetrizable. Similarly, we can show that the channel is \one-symmetrizable. }

\olive{Next, we show that this channel is non-overwritable. Suppose the channel is \two-overwritable. Let $P_{X'|XY}$ be the overwriting attack by user \one which satisfies \eqref{eq:overwritable}. Then for $(x, y, z) = ((1,1), (1, 1), (2, 0))$ and all $y' = (y'_1, y_2')$, \eqref{eq:overwritable} implies
\begin{align*}
\sum_{(x_1',x_2')}P_{X'|XY}((x_1',x_2')|(1,1),(1,1))W((2,0)|(x_1',x_2'),(y'_1, y_2'))= W((2,0)|(1,1),(1,1)).
\end{align*}
However, for $(y_1',y_2') = (0,0)$, the LHS evaluates to $0$ whereas the RHS evaluates to $1$. Hence, the channel is not \two-overwritable. Similarly, we can show that the channel is not \one-overwritable. }

\olive{Thus, since this \bmac is symmetrizable and non-overwritable, $\cC_{\reliable} = \{(0,0)\}$ and  $\cC_{\auth} = \cC_{\MAC}$ respectively. 
\blue{Using the results in Section~\ref{sec:capacity}, we can show that  the capacity region of communication with adversary identification $\cC$ is the same as the (non-adversarial) capacity region of the binary erasure MAC. This is formally argued in Section~\ref{example:outer_bound}.}
The capacity regions under these three models are plotted in Fig.~\ref{fig:ex3}. }

\begin{figure}[h]
\centering
\begin{tikzpicture}[scale=0.65]
\begin{axis}[
    legend style={font=\large},
    axis lines = left,
    xlabel = {$R_{\one}$},
    ylabel = {$R_{\two}$},
]

\addplot [thick,
    domain=0:0.5, 
    samples=50, smooth, 
    color=olive,
    ]
    {1};
\addlegendentry{\textcolor{olive}{\large $\cC$\phantom{abcde}}}
\addplot[thick,
    domain=0:0.5, 
    samples=50, smooth,
    color=blue,
    ]
    coordinates {
    (0,0)};
\addlegendentry{\textcolor{blue}{\large $\cC_{\reliable}$}}

\addplot [thick,
    domain=0:0.5, 
    samples=50, smooth, 
    color=magenta,
]
{2};
\addlegendentry{\textcolor{magenta}{\large $\cC_{\auth}\phantom{ab}$}}
\addplot [thick,
    domain=0.5:1, 
    samples=50, smooth, 
    color=olive,
    ]
    {1.5-x};
\addplot [thick,
    domain=0.5:2, 
    samples=50, smooth, 
    color=magenta,
]
{2.5-x};
\addplot[thick, 
	samples=50, smooth,
	domain=0:6,
	olive] 
	coordinates {(1,0)(1,0.5)};
\addplot[thick, 
	samples=50, smooth,
	domain=0:6,
	magenta] 
	coordinates {(2,0)(2,0.5)};
\addplot[thick, 
	samples=50, smooth,
	domain=0:6,
	white] 
	coordinates {(2.1,0.1)(2.2,0.1)};
\addplot[thick, 
	samples=50, smooth,
	domain=0:6,
	white] 
	coordinates {(0.1,2.1)(0.2,2.2)};
\addplot[
    color=blue,
    mark=square*,
    mark size=1pt
    ]
    coordinates {
    (0,0)};
\end{axis}
\end{tikzpicture}\caption{Capacity regions for the \bmac in Example~\ref{ex:3}: $\cC_{\reliable} = \{0,0\}$; $\cC = \cC_{\MAC}$ of $Z_1 = X_1+Y_1$; and $\cC_{\auth} = \cC_{\MAC}$ of $(Z_1,Z_2)=(X_1+Y_1, X_2\oplus Y_2)$. }\label{fig:ex3}\end{figure}

\section{Capacity region}\label{sec:capacity}

\subsection{Inner bound}\label{sec:inner_bound}
For distributions $P_{\one}$ and $P_{\two}$  over $\cX$ and $\cY$ respectively, we define  $\cP(P_{\one}, P_{\two}) \defineqq \{P_{XY\tilde{X}\tilde{Y}Z}: P_{X\tilde{Y}Z}=P_{\one}\times P_{\tilde{Y}}\times W \text{ for some }P_{\tilde{Y}} \text{ and }P_{\tilde{X}YZ} = P_{\tilde{X}}\times P_{\two}\times W \text{ for some }P_{\tilde{X}}\}$.
Let $\cR_{1}(P_{\one}, P_{\two})$ be the set of rate pairs $(R_{\one}, R_{\two})$ such that
\begin{align}\label{eq:inner_bd_1}
R_{\one}&\leq \min_{P_{XY\tilde{X}\tilde{Y}Z} \in \cP(P_{\one}, P_{\two})} I(X;Z), \text{ and }\nonumber\\
R_{\two}&\leq \min_{P_{XY\tilde{X}\tilde{Y}Z} \in \cP(P_{\one}, P_{\two}):X\independent Y} I(Y;Z|X).
\end{align} 
Similarly, let $\cR_{2}(P_{\one}, P_{\two})$ be the set of rate pairs given by
\begin{align}\label{eq:inner_bd_2}
R_{\one}&\leq \min_{P_{XY\tilde{X}\tilde{Y}Z} \in \cP(P_{\one}, P_{\two}):X\independent Y} I(X;Z|Y), \text{ and }\nonumber\\
R_{\two}&\leq \min_{P_{XY\tilde{X}\tilde{Y}Z} \in \cP(P_{\one}, P_{\two})} I(Y;Z).
\end{align}

\begin{thm}[Inner bound]\label{thm:inner_bd}
When $\mathsf{int}(\cC)\neq \emptyset$, 
\begin{align*}
\mathsf{conv}(\cup_{P_{\one}, P_{\two}}\inp{\cR_1(P_{\one}, P_{\two})\cup\cR_{2}(P_{\one}, P_{\two})})\subseteq \cC.
\end{align*}

\end{thm}
The proof uses the same codebook lemma (Lemma~\ref{lemma:codebook}) as used for the achievability of Theorem~\ref{thm:main_result} but with different rates. The decoder is a slightly modified version of the decoder used in Theorem~\ref{thm:main_result}. The modified decoder uses similar conditions as used in the decoder for the achievability of Theorem~\ref{thm:main_result}, but the steps are performed in a specific order. At each step only those codewords are considered which have not been eliminated in the previous steps.  
Please see Appendix~\ref{sec:inner_bd_proof} for details. 

\begin{remark}\label{remark:bec_inner_bd}
The inner bound to the capacity region given by Theorem~\ref{thm:inner_bd} is tight for binary erasure MAC~\cite[pg.~83]{YHKEG} \blue{as argued below in Section~\ref{sec:example_tightness}.}
\end{remark}
\begin{remark}\label{remark:bec_inner_bd2}
Our \blue{objective} in deriving Theorem~\ref{thm:inner_bd} is to obtain a simple expression for the inner bound. However, it is not clear if Theorem~\ref{thm:inner_bd} implies the achievability direction of Theorem~\ref{thm:main_result} (i.e., whether the inner bound in Theorem~\ref{thm:inner_bd} has a non-empty interior for all non-spoofable channels). Still the code (codebook and decoder) used in proving Theorem~\ref{thm:inner_bd}  can also be used to achieve positive rates for both users of a non-spoofable channel, which is the forward direction of Theorem~\ref{thm:main_result}. See Remark~\ref{remark:new_code_pos1} for details.
\end{remark}
\noindent\subsubsection{Tightness of the inner bound for the Binary Erasure MAC}\label{sec:example_tightness}
\begin{Example}
\olive{The binary erasure MAC~\cite[pg.~83]{YHKEG} is given by $Z=X+Y$ where $\cX=\cY=\{0,1\}$ and $\cZ = \{0,1,2\}$. We will show that for the binary erasure MAC, the inner bound on \cC\ given by Theorem~\ref{thm:inner_bd} is the same as its (non-adversarial) capacity region $\cC_{\MAC}$. Hence, it is tight. }

\olive{Recall that for distributions $P_{\one}$ and $P_{\two}$  over $\cX$ and $\cY$, $$\cP(P_{\one}, P_{\two}) = \{P_{XY\tilde{X}\tilde{Y}Z}: P_{X\tilde{Y}Z}=P_{\one}\times P_{\tilde{Y}}\times W \text{ for some }P_{\tilde{Y}} \text{ and }P_{\tilde{X}YZ} = P_{\tilde{X}}\times P_{\two}\times W \text{ for some }P_{\tilde{X}}\}$$. We choose $P_{\one}$ and $ P_{\two}$ arbitrarily close to the uniform distribution $U$ on $\{0, 1\}$ while ensuring that  $P_{\one}\neq P_{\two}$. Using these distributions, we will show that  for $P_{XY\tilde{X}\tilde{Y}Z}\in \cP(P_{\one}, P_{\two})$ satisfying $X\indep Y$, $\tilde{X} = X$ and $\tilde{Y} = Y$. }

\olive{To this end, consider $P_{XY\tilde{X}\tilde{Y}Z}\in \cP(P_{\one}, P_{\two})$.
\begin{align}\label{eq:ex_1}
\bbP(Z = 0) = P_{\one}(0)P_{\tilde{Y}}(0) = P_{\tilde{X}}(0)P_{\two}(0).
\end{align}
\begin{align*}
\bbP(Z=2) &= (1-P_{\one}(0))(1-P_{\tilde{Y}}(0)) =  (1-P_{\tilde{X}}(0))(1-P_{\two}(0)).
\end{align*}
This implies that 
\begin{align}
1+ P_{\one}(0)P_{\tilde{Y}}(0) -P_{\one}(0)-P_{\tilde{Y}}(0) = 1+ P_{\tilde{X}}(0)P_{\two}(0) -P_{\tilde{X}}(0)-P_{\two}(0).\label{eq:ex_11}
\end{align}
Using \eqref{eq:ex_1} and \eqref{eq:ex_11}, we get $P_{\one}(0)+P_{\tilde{Y}}(0) = P_{\tilde{X}}(0)+P_{\two}(0)$. Thus, 
\begin{align}\label{eq:ex_2}
P_{\tilde{X}}(0) = P_{\one}(0)+P_{\tilde{Y}}(0) - P_{\two}(0).
\end{align}
Substituting the value of $P_{\tilde{X}}(0)$ from \eqref{eq:ex_2} into \eqref{eq:ex_1}, we get $P_{\one}(0)P_{\tilde{Y}}(0) = P_{\one}(0)P_{\two}(0)+ P_{\tilde{Y}}(0)P_{\two}(0)- P_{\two}(0)P_{\two}(0).$ This implies that 
\begin{align*}
\inp{P_{\one}(0)-P_{\two}(0)}\inp{P_{\tilde{Y}}(0) -P_{\two}(0)} = 0.
\end{align*}
Thus, either $P_{\one}(0)=P_{\two}(0)$ or $P_{\tilde{Y}}(0) = P_{\two}(0)$.
Substituting this in \eqref{eq:ex_2}, we get either $P_{\one}(0)=P_{\two}(0)$ and $P_{\tilde{X}}(0)=P_{\tilde{Y}}(0)$, or $P_{\tilde{Y}}(0) = P_{\two}(0)$ and $P_{\tilde{X}}(0) = P_{\one}(0)$. If we choose $P_{\one}$ and $P_{\two}$ such that $P_{\one}\neq P_{\two}$, then for every $P_{XY\tilde{X}\tilde{Y}Z}\in \cP(P_{\one}, P_{\two})$, $P_{\tilde{Y}} = P_{Y} = P_{\two}$ and $P_{\tilde{X}} = P_{X} = P_{\one}$. }

\olive{We know from the definition of $\cP(P_{\one}, P_{\two})$, that $X\indep \tilde{Y}$ and $\tilde{X}\indep Y$. We now analyse the case when there is a further restriction of $X\indep Y$ on the distributions. From the definition of $\cP(P_{\one}, P_{\two})$, we note that $P_{\tilde{X}Y|X\tilde{Y}}(0,0|0,0) = 1$ and $P_{\tilde{X}Y|X\tilde{Y}}(1,1|1,1) = 1$.  Let $P_{\tilde{X}Y|X\tilde{Y}}(0,1|0,1) = \alpha$ and $P_{\tilde{X}Y|X\tilde{Y}}(1,0|0,1) = 1-\alpha$ (Note that $P_{\tilde{X}Y|X\tilde{Y}}((0,0) |0,1)) = P_{\tilde{X}Y|X\tilde{Y}}((1,1) |0,1)) = 0$ by definition of $\cP(P_{\one}, P_{\two})$). Similarly, let $P_{\tilde{X}Y|X\tilde{Y}}(1, 0|1,0) = \beta$ and $P_{\tilde{X}Y|X\tilde{Y}}(0, 1|1,0) = 1-\beta$. 
Thus, $P_{XY}(0,0) = P_{X\tilde{Y}}(0, 0)P_{\tilde{X}Y|X\tilde{Y}}(0,0|0,0) + P_{X\tilde{Y}}(0, 1)P_{\tilde{X}Y|X\tilde{Y}}(1,0|0,1)= P_{X}(0)P_{\tilde{Y}}(0)\cdot 1 + P_{X}(0)P_{\tilde{Y}}(1)\cdot (1-\alpha)$. Also, $P_{XY}(0,0) = P_{X}(0)P_Y(0) = P_{X}(0)P_{\tilde{Y}}(0)$ (The last equality follows by choosing $P_X \neq P_Y$ (which is the same as $P_{\one}\neq P_{\two}$)). This implies that $\alpha = 1$. By evaluating  $P_{XY}(1,1)$, we can show that $\beta = 1$. This implies that $\tilde{X} = X$ and $\tilde{Y} = Y$.}

\olive{As mentioned earlier, we choose $P_{\one}$ and $P_{\two}$ arbitrarily close to uniform distributions such that $P_{\one}\neq P_{\two}$. Thus, in the limit, \eqref{eq:inner_bd_1} evaluates to $R_{\one}\leq 0.5$ and $R_{\two} \leq 1$, and \eqref{eq:inner_bd_2} evaluates to $R_{\one}\leq 1$ and $R_{\two} \leq 0.5$. Using time sharing between these two rate pairs, we obtain the entire (\blue{non-adversarial}) MAC capacity region (This is the rate region $\cC$ in Fig.~\ref{fig:ex3}). \blue{This is tight because the non-adversarial MAC capacity region is an outer bound on the capacity region of communication with adversary identification.}}
\end{Example}

\subsection{Outer bound}\label{sec:outer_bd}
\blue{The outer bound is provided in terms of the capacity region of an Arbitrarily Varying Multiple Access Channel (AV-MAC). See Section~\ref{sec:AVMAC}.}

\begin{restatable}{defn}{outerboundAVMAC}
\label{defn:outerboundAVMAC}
For a MAC $W$, let  $\tilde{\cW}_W$ be the set of MACs $\tilde{W}$ such that there is a  pair of conditional distributions $Q_{X'|X}$ and $Q_{Y'|Y}$ satisfying 
\begin{align}\label{eq:outer_bound}
\tilde{W}(z|x,y) &= \sum_{x'}Q_{X'|X}(x'|x)W(z|x',y) \nonumber\\
&= \sum_{y'}Q_{Y'|Y}(y'|y)W(z|x,y'),
\end{align}for all $x, y, z\in \cX\times\cY\times \cZ$. \blue{Let the set $\tilde{\cW}_W$ be indexed by a set $\cS$. Then, $\tilde{\cW}_W$ can be written as an AV-MAC as below: $$\tilde{\cW}_W = \inb{\tilde{W}(\cdot|\cdot, \cdot, s): s\in \cS,\, \tilde{W}(z|x,y, s) = \sum_{x'}Q_{X'|X S}(x'|x, s)W(z|x',y) = \sum_{y'}Q_{Y'|Y S}(y'|y, s)W(z|x,y'), \forall x, x', y, z}.$$}
\end{restatable}

\begin{figure}[h]
\centering
\subfigure{
\begin{tikzpicture}[scale=0.4]
	\draw (4.1,4) rectangle ++(1.3,1.5) node[pos=.5]{\footnotesize $W$};
	\draw[red,->] (3.4,5.25) -- ++(0.7,0);
	\draw[->] (3.4,4.25) -- ++(0.7,0);
	\draw[-] (3.4,4.25) -- ++ (0,-1);
	\draw[red,-] (3.4,5.25) -- ++ (0,1);
	\draw[-] (2.4-2.2-0.6, 3.25) node[left]{\scriptsize {$y'$}} --  ++(1+2.8,0) ;
	\draw[red] (2.4-2.2,6.25-0.6) rectangle ++(2.5,1.2) node[pos=0.5]{\scriptsize $Q_{X|X'}$};
	\draw[red,->] (2.4-2.2-0.6,6.25) node[left]{\scriptsize $x'$}-- ++ (0.6,0);
	\draw[red,-] (2.4+0.3,6.25) -- ++ (0.7,0);
	\draw[->] (5.4,4.75) --  ++ (1,0)node[anchor= west] {\scriptsize $Z\qquad \Longleftrightarrow$ };
	\draw[gray] (-0.1,2.2) rectangle ++(6,5);
	\draw[gray] (3.1,2-0.5) node{\scriptsize $\tilde{W}$};
\end{tikzpicture}}
\subfigure{
\begin{tikzpicture}[scale=0.4]
	\draw (4.1+11,4) rectangle ++(1.3,1.5) node[pos=.5]{\footnotesize $W$};
	\draw[->] (3.4+11,5.25) -- ++(0.7,0);
	\draw[red,->] (3.4+11,4.25) -- ++(0.7,0);
	\draw[red,-] (3.4+11,4.25) -- ++ (0,-1);
	\draw[-] (3.4+11,5.25) -- ++ (0,1);
	\draw[red,->] (2.4-2.2-0.6+11, 3.25) node[left]{\scriptsize {$y'$}} --  ++(0.6,0) ;
	\draw[red] (2.4-2.2+11,3.25-0.6) rectangle ++(2.5,1.2) node[pos=0.5]{\scriptsize $Q_{Y|Y'}$};
	\draw[-] (2.4-2.2-0.6+11,6.25) node[left]{\scriptsize $x'$}-- ++ (1+2.8,0);
	\draw[red,-] (2.4+0.3+11,3.25) -- ++ (0.7,0);
	\draw[->] (5.4+11,4.75) --  ++ (1,0)node[anchor= west] {\scriptsize $Z\qquad \Longleftrightarrow$ };
	\draw[gray] (-0.1+11,2.2) rectangle ++(6,5);
	\draw[gray] (3.1+11,2-0.5) node{\scriptsize $\tilde{W}$};
\end{tikzpicture}}
	\subfigure{
	\begin{tikzpicture}[scale=0.4]
	\draw (4.1+20-2,4) rectangle ++(1.3,1.5) node[pos=.5]{\footnotesize $\tilde{W}$};
	\draw[->] (3.4+20-2,5.25) -- ++(0.7,0);
	\draw[->] (3.4+20-2,4.25) -- ++(0.7,0);
	\draw[-] (3.4+20-2,4.25) -- ++ (0,-1);
	\draw[-] (3.4+20-2,5.25) -- ++ (0,1);
	\draw[-] (2.4-4.2-0.6+20+2.8+0.5, 3.25) node[left]{\scriptsize {$y'$}} --  ++(0.5,0) ;
	\draw[-] (2.4-4.2-0.6+20+	2.8+0.5,6.25) node[left]{\scriptsize $x'$}-- ++ (0.5,0);
	\draw[-] (2.4+0.3+20-2,6.25) -- ++ (0.7,0);
	\draw[->] (5.4+20-2,4.75) --  ++ (0.5,0)node[anchor= west] {\scriptsize $Z$ };	
	\draw[white] (-0.1+20-2,1.2) rectangle ++(7,6);
\end{tikzpicture}
}
  \caption{\blue{The figure shows $\tilde{W}$ as defined by \eqref{eq:outer_bound}. 
 $\tilde{W}$ results from a pair of attacks $Q_{X'|X}$ and $Q_{Y'|Y}$ by users \one\ and \two\ respectively which produce the same output.}}
\label{fig:outer_bound} 
\end{figure}

\blue{Fig.~\ref{fig:outer_bound} shows the channel $\tilde{W}$ produced by a  pair of conditional distributions $Q_{X'|X}$ and $Q_{Y'|Y}$ as given in \eqref{eq:outer_bound}.} The outer bound on the capacity region for communication with adversary identification of a \bmac $W$ is in terms of the capacity region of the AV-MAC $\tilde{\cW}_W$ defined above.
We first notice that $W\in \tilde{\cW}_W$ by choosing trivial distributions $Q_{X'|X}(x|x) = 1$ for all $x$ and $Q_{Y'|Y}(y|y) = 1$ for all $y$.
Additionally, the set $\tilde{\cW}_W$ is convex because for every $(Q_{X'|X}, Q_{Y'|Y})$ and $(Q'_{X'|X}, Q'_{Y'|Y})$  satisfying \eqref{eq:outer_bound}, the pair $(\alpha Q_{X'|X}+ (1-\alpha) Q'_{X'|X}, \alpha Q_{Y'|Y} + (1-\alpha) Q'_{Y'|Y})$, $\alpha \in [0,1]$ also satisfies \eqref{eq:outer_bound}.

\begin{figure}[h]
\subfigure{
\begin{tikzpicture}[scale=0.4]
	\draw (4.1,4) rectangle ++(1.5,1.5) node[pos=.5]{\footnotesize $W^n$};
	\draw[red,->] (3.4,5.25) -- ++(0.7,0);
	\draw[->] (3.4,4.25) -- ++(0.7,0);
	\draw[-] (3.4,4.25) -- ++ (0,-1);
	\draw[red,-] (3.4,5.25) -- ++ (0,1);
	\draw[-] (2.4-2.2-0.6, 3.25) node[left]{\scriptsize {$f_{\two}(M_{\two})$}} --  ++(1+2.8,0) ;
	\draw[red] (2.4-2.2,6.25-0.7) rectangle ++(2.5,1.4) node[pos=0.5]{\scriptsize $Q^n_{X|X'}$};
	\draw[red,->] (2.4-2.2-0.6,6.25) node[left]{\scriptsize $f_{\one}(M_{\one})$}-- ++ (0.6,0);
	\draw[red,-] (2.4+0.3,6.25) -- ++ (0.7,0);
	\draw[->] (5.4+0.2,4.75) --  ++ (1,0)node[anchor= west] {\scriptsize $Z^n\qquad \Longleftrightarrow$ };
\end{tikzpicture}
}
\subfigure{
\begin{tikzpicture}[scale=0.4]

	\draw (4.1+11,4) rectangle ++(1.5,1.5) node[pos=.5]{\footnotesize $W^n$};
	\draw[->] (3.4+11,5.25) -- ++(0.7,0);
	\draw[red,->] (3.4+11,4.25) -- ++(0.7,0);
	\draw[red,-] (3.4+11,4.25) -- ++ (0,-1);
	\draw[-] (3.4+11,5.25) -- ++ (0,1);
	\draw[red,->] (2.4-2.2-0.6+11, 3.25) node[left]{\scriptsize {$f_{\two}(M_{\two})$}} --  ++(0.6,0) ;
	\draw[red] (2.4-2.2+11,3.25-0.7) rectangle ++(2.5,1.4) node[pos=0.5]{\scriptsize $Q^n_{Y|Y'}$};
	\draw[-] (2.4-2.2-0.6+11,6.25) node[left]{\scriptsize $f_{\one}(M_{\one})$}-- ++ (3.8,0);
	\draw[red,-] (2.4+0.3+11,3.25) -- ++ (0.7,0);
	\draw[->] (5.4+11+0.2,4.75) --  ++ (1,0)node[anchor= west] {\scriptsize $Z^n\qquad \Longleftrightarrow$ };
\end{tikzpicture}
}
\subfigure{
\begin{tikzpicture}[scale=0.4]

	\draw (4.1+22,4) rectangle ++(1.5,1.5) node[pos=.5]{\footnotesize $\tilde{W}^n$};
	\draw[->] (3.4+22,5.25) -- ++(0.7,0);
	\draw[->] (3.4+22,4.25) -- ++(0.7,0);
	\draw[-] (3.4+22,4.25) -- ++ (0,-1);
	\draw[-] (3.4+22,5.25) -- ++ (0,1);
	\draw[-] (2.4-2.2-0.6+2.8+22, 3.25) node[left]{\scriptsize {$f_{\two}(M_{\two})$}} --  ++(1,0) ;
	\draw[-] (2.4-2.2-0.6+2.8+22,6.25) node[left]{\scriptsize $f_{\one}(M_{\one})$}-- ++ (1,0);
	\draw[->] (5.4+0.2+22,4.75) --  ++ (1,0)node[anchor= west] {\scriptsize $Z^n$ };
	

\end{tikzpicture}}

\caption{\blue{The figure shows an attack strategy based on attacks of user \one\ and user \two\ satisfying \eqref{eq:outer_bound}.   For a malicious user $\one$, an attack $Q^n_{X|X'} = \prod_{i=1}^{n}Q_{i,X'|X}$ by user $\one$ where each $Q_{i,X'|X}$ satisfies \eqref{eq:outer_bound} for some $Q_{i,Y'|Y}$ cannot be distinguished from the attack $\prod_{i=1}^{n}Q_{i,Y'|Y}(=Q^n_{Y|Y'})$ by user \two. Hence, the decoder must output a pair of messages of which the message of the honest user is correct with high probability. This means that the decoded output should be correct for the channel $\tilde{W}^n$.}}\label{fig:outer_bound_nletter}
\end{figure}

To get an outer bound, consider an adversary identifying code with a small probability of error for a MAC $W$. Suppose user \one is malicious and attacks in the following manner: it runs its encoder on a uniformly distributed message from its message set, then passes the output of the encoder through $\prod_{i=1}^{n}Q_{i,X'|X} \blue{(=Q^n_{X|X'})}$ where for all $i\in [1:n]$, $(Q_{i,X'|X}, Q_{i,Y'|Y})$ satisfy \eqref{eq:outer_bound} for some $Q_{i,Y'|Y}$. The output of $\prod_{i=1}^{n}Q_{i,X'|X}$ is finally sent to the MAC ${W}$ as input by user \one. 
This can also be interpreted as an attack by user \two using $\prod_{i=1}^{n}Q_{i,Y'|Y}\blue{(=Q^n_{Y|Y'})}$. 
Thus, at the receiver, it is not clear if user \one attacked 
or user \two attacked. \blue{This is shown in Fig.~\ref{fig:outer_bound_nletter}.}
Hence, the malicious user cannot be identified reliably. So, the decoder must output a pair of messages of which the message for the honest user is correct with high probability. This means that the decoding should be correct with high probability for $\tilde{W}^{(n)}(\cdot|\cdot,\cdot) = \prod_{i=1}^{n}\inp{\sum_{x_i\in\cX}Q_{i,X'|X}(x_i|\cdot)W(\cdot|x_i, \cdot)} = \prod_{i=1}^{n}\sum_{y_i\in \cY}Q_{i,Y'|Y}(y_i|\cdot)W(\cdot|\cdot, y_i)$. Thus, any good adversary identifying code for the MAC W must also be a good communication code for $\tilde{W}^{(n)}$. 
In fact, we can show the following:
\begin{lemma}\label{lemma:AV_MAC}
Any $(N_{\one}, N_{\two}, n)$ adversary identifying code $(f_{\one},f_{\two}, \phi)$ for a MAC $W$ with $P_{e}(f_{\one},f_{\two}, \phi) \leq \epsilon$ is also an $(N_{\one}, N_{\two}, n)$ communication code for the AV-MAC $\tilde{\cW}_W$ with average probability of error at most $2\epsilon$.
\end{lemma}
\blue{The proof of Lemma~\ref{lemma:AV_MAC} is given in Appendix~\ref{proof_outer_bound}. The lemma} implies that the capacity region $\cC$ (for communication with adversary identification) of MAC $W$ must be a subset of the capacity region of the AV-MAC $\tilde{\cW}_W$ (Definition~\ref{defn:outerboundAVMAC}) parametrized by a pair of distributions $(Q_{X'|X}, Q_{Y'|Y})$ satisfying \eqref{eq:outer_bound}. 
This argument is formalized below:
\begin{thm}[Outer bound]\label{thm:outer_bd}
$\cC \subseteq \cC_{\AVMAC}(\tilde{\cW}_W)$. Moreover, there exists 
an AV-MAC ${\cW}_W$ such that $\cC_{\AVMAC}({\cW}_W)=\cC_{\AVMAC}(\tilde{\cW}_W)$ 
and\footnote{The capacity region of an AV-MAC only depends on its convex hull which is defined by taking convex combinations of channels under different states \cite{Jahn81,AhlswedeC99}. The AV-MAC ${\cW}_W$ has the same convex hull as $\tilde{\cW}_W$ but with only finitely many states.} $|{\cW}_W| \leq 2^{|\cX|^2 + |\cY|^2}$.
\end{thm}
\begin{proof}
Lemma~\ref{lemma:AV_MAC} implies that $\cC$, the capacity region of communication with adversary identification of a MAC $W$, is outer bounded by the capacity region $\cC_{\AVMAC}(\tilde{\cW}_W)$ of the AV-MAC $\tilde{\cW}_{W}$.
Further, note that the capacity of an AV-MAC $\cW$ only depends on the convex hull of $\cW$ (see \cite[Theorem 1]{Jahn81} \blue{and Section~\ref{sec:AVMAC}}). So, the capacity of $\tilde{\cW}_{W}$ is the same as the capacity of another AV-MAC $\cW_{W}$ which consists of vertices of the convex polytope $\tilde{\cW}_{W}\subseteq \bbR^{|\cX|\times|\cY|\times|\cZ|}$.
The elements in the set $\tilde{\cW}_{W}$ are parameterized by $(Q_{X'|X}, Q_{Y'|Y})$ pairs. It consists of the vertices of the polytope formed using constraints in \eqref{eq:outer_bound} and constraints of the form: (1) $\sum_{x'}P_{X'|X}(x'|x) = 1$ for all $x$, and (2) $P_{X'|X}(x'|x)\geq 0$. There are similar constraints for $P_{Y'|Y}$. Note that there are $|X|^2 + |Y|^2$ inequality constraints. Every point in the resulting polytope satisfies all the equality constraints. We will get faces, edges, vertices etc. depending on the number of additional inequality constraints satisfied at that point. Thus, \blue{the} number of vertices $\leq 2^{|X|^2 + |Y|^2}$.
\end{proof}

We will now give an alternative proof of the first part of Theorem~\ref{thm:main_result}. We will use the connection of  the present model with the AV-MAC  given by Lemma~\ref{lemma:AV_MAC} and then establish a connection between an  $\one$-\spoofable \bmac and a symmetrizable-$\cX$ AV-MAC \blue{(see \eqref{eq:sym_x}). As we saw in Theorem~\ref{thm:gubner_sym}, a} symmetrizable-$\cX$ AV-MAC does not allow reliable communication by user \one. 
\begin{thm}[First part of Theorem~\ref{thm:main_result}]\label{thm:main_result_first_part}
If a \bmac is \one-\spoofable (resp. \two-\spoofable), communication with adversary identification from user-\one (resp. user-\two) is impossible. Specifically, for any $(N_{\one}, N_{\two}, n)$ adversary identifying code with $N_{\one}\geq 2$ (resp. $N_{\two}\geq 2$), the probability of error is at least $1/12$.  
\end{thm}
\begin{proof}[Alternate proof]
We will first show that if a \bmac is $\one$-\spoofable then the corresponding AV-MAC given by Definition~\ref{defn:outerboundAVMAC} is symmetrizable-$\cX$ \blue{\eqref{eq:sym_x}}. 

To this end, suppose a given \bmac $W$ is $\one$-\spoofable. This implies that \eqref{eq:spoof1} holds. Then, for $\cS = \cX$ and by replacing $\tilde{X}$ and with $S$ in \eqref{eq:spoof1}, we obtain 
\begin{align}
&\sum_{y}Q_{Y|\tilde{Y}S}(y|\tilde{y},\tilde{x})\mch(z|x',y) \label{eq:spoof1Sa}\\
&= \sum_{y}Q_{Y|\tilde{Y}S}(y|\tilde{y},x')\mch(z|\tilde{x},y) \label{eq:spoof1Sb}\\
& = \sum_{x}Q_{X|X'S}(x|x',\tilde{x})\mch(z|x,{\tilde{y}})\qquad\text{ for all }x', \,\tilde{x},\, \tilde{y},\, z.\label{eq:spoof1Sc}
\end{align}
For every $s\in \cS$, \eqref{eq:spoof1Sa} and \eqref{eq:spoof1Sc} imply that 
\begin{align*}
\tilde{W}_{Z|XYS}(z|x',\tilde{y},s)\defineqq &\sum_{y}Q_{Y|\tilde{Y}S}(y|\tilde{y},s)\mch(z|x',y) \\
 = &\sum_{x}Q_{X|X'S}(x|x',s)\mch(z|x,{\tilde{y}})\qquad \text{ for all }x', \tilde{y}, z.
\end{align*} Thus, the set $\tilde{\cW}_{W}$ in Definition~\ref{defn:outerboundAVMAC} is such that 
\begin{align*}
\inb{\tilde{W}_{Z|XYS}(\cdot|\cdot,\cdot,s):s\in \cS}\subseteq  \tilde{\cW}_{W}.
\end{align*}
By \eqref{eq:spoof1Sa} and \eqref{eq:spoof1Sb}, we note  that for all $z, \tilde{x}, \tilde{y}, x'$,
\begin{align*}
\tilde{W}_{Z|XYS}(z|x',\tilde{y},\tilde{x})= \tilde{W}_{Z|XYS}(z|\tilde{x},\tilde{y},x')
\end{align*}Thus, for $P_{S|X}(s|x) \defineqq \mathbbm{1}_{\inb{s=x}},\, s,x\in \cX$, where $\mathbbm{1}$ is the indicator function, 
\begin{align*}
\sum_{s}P_{S|X}(s|\tilde{x})\tilde{W}_{Z|XYS}(z|x',\tilde{y},s)= \sum_{s'}P_{S|X}(s'|x')\tilde{W}_{Z|XYS}(z|\tilde{x},\tilde{y},s)
\end{align*} for all $x', \tilde{x}, \tilde{y}$ and $z$. That is, the AV-MAC $\tilde{\cW}_{W}$ is symmetrizable-$\cX$.
Hence, if the \bmac $W$ is \one-\spoofable, then the AV-MAC $\tilde{\cW}_W$ is symmetrizable-$\cX$. 

Now, suppose that there is an $(N_{\one}, N_{\two}, n)$ adversary identifying code $(f_{\one},f_{\two}, \phi)$ such that $P_{e}(f_{\one},f_{\two}, \phi) \leq 1/12$. From Lemma~\ref{lemma:AV_MAC}, this implies that there is an $(N_{\one}, N_{\two}, n)$ communication code for the AV-MAC $\tilde{\cW}_{W}$ with average probability of error at most $1/6$. However, for any $(N_{\one}, N_{\two}, n)$ code (with $N_{\one}\geq 2$) for a symmetrizable-$\cX$ AV-MAC,  \blue{Gubner~\cite[Lemma~3.10]{Gubnerthesis}} shows that the average probability of error is at least $1/4$ \blue{(see Remark~\ref{remark:gubner})}. Thus, there does not exist any $(N_{\one}, N_{\two}, n)$ adversary identifying code $(f_{\one},f_{\two}, \phi)$ such that $P_{e}(f_{\one},f_{\two}, \phi) \leq 1/12$ and $N_{\one}\geq 2$.


\end{proof}
\subsubsection{\blue{Computing $\cC$ for Example~\ref{ex:3}}}\label{example:outer_bound}
\begin{Eexample}\olive{To compute $\cC$, we first give an outer bound. \blue{From Definition~\ref{defn:outerboundAVMAC}, note that }$\tilde{\cW}_W$ contains a channel $\tilde{W}$ satisfying
\begin{align*}
\tilde{W}(x_1+y_1, v|(x_1, x_2), (y_1, y_2)) = 0.5,
\end{align*} for all $x_1, x_2, y_1, y_2, v\in \inb{0,1}$. This is obtained by using the pair {$(Q_{X'|X}, Q_{Y'|Y})$} defined by $Q_{X'|X}((x_1, u)|(x_1,x_2)) = 0.5$ for all $u,x_1, x_2\in \{0,1\}$ and $Q_{Y'|Y}((y_1, v)|(y_1,y_2)) = 0.5$ for all $v,y_1, y_2\in \{0,1\}$ in~\eqref{eq:outer_bound}. 
Note that the channel $\tilde{W}$ has the same first component as $W$ (i.e., a binary erasure MAC) and a second component whose output $Z_2$ is independent of the inputs. We will argue that $\cC$ is outer bounded by the (non-byzantine) MAC capacity region of $\tilde{W}$.
 Instead of transmitting the output of its encoder, a malicious user \one may attack by passing it through the DMC $Q_{X'|X}$ defined above and transmitting the resulting output over the \bmac. This results in the effective MAC $\tilde{W}$. The receiver cannot distinguish between this malicious user \one attacking with $Q_{X'|X}$ and a malicious user \two attacking with $Q_{Y'|Y}$ as both result in the same resultant MAC $\tilde{W}$. This implies that the receiver must decode the input of both users correctly under $\tilde{W}$. Thus, $\cC$ is outer bounded by the (non-adversarial) capacity region of $\tilde{W}$ which is the capacity region of the binary erasure MAC. We can show that this outer bound is tight by using an adversary identifying code for the binary erasure MAC component $Z_1=X_1+Y_1$ (see Section~\ref{sec:example_tightness}) and any arbitrary inputs for the other component.}
\end{Eexample}



\section{Connections to Randomized Coding Capacity Region}\label{sec:rand_capacity}
While not the focus of this paper, in this section we comment on the relationship of the capacity region $\cC$ of deterministic codes with that of randomized codes. Note that we do not provide any  direct achievability scheme for randomized codes in this section.
Before presenting the formal definitions, we draw the reader's attention to two points about our setup for randomized codes:
\begin{itemize}
        \item The encoders share {\em independent} randomness with the decoder. This is similar to the randomized code of Jahn~\cite{Jahn81} for AV-MACs.
        \item When a user is adversarial, we allow that user to adversarially select the realization of its randomness.\footnote{This is analogous to the model studied in \cite{NehaBDP23} for reliable communication.}
\end{itemize}
Both these choices will prove important in making the connection between the capacity regions under deterministic and randomized codes.  

\begin{defn}[Randomized adversary identifying code]\label{defn:rand-code}
	An $(\nummsg_{\one},\nummsg_{\two},\numrand_{\one},\numrand_{\two},n)$  {\em randomized adversary identifying code} $\inp{F_{\one}, F_{\two}, \phi_{F_{\one},F_{\two}}}$ for a MAC with byzantine users consists of the following: 
\begin{enumerate}[label=(\roman*)]
\item Two message sets, $\mathcal{M}_i = \{1,\ldots,\nummsg_i\}$, $i=\one,\two$,
\item Two collections of deterministic encoders, 
	$\codeset_{\one}=\{f_{\one,1},\ldots,f_{\one,\numrand_{\one}}\}$, where $f_{\one,i}:\mathcal{M}_{\one}\rightarrow \mathcal{X}^n$, $i=1,\ldots,\numrand_{\one}$, and $\codeset_{\two}=\{f_{\two,1},\ldots,f_{\two,\numrand_{\two}}\}$, where $f_{\two,i}:\mathcal{M}_{\two}\rightarrow \mathcal{Y}^n$, $i=1,\ldots,\numrand_{\two}$,
\item Two independent randomized encoders $F_{\one}$ and $F_{\two}$ with distributions $p_{\Fo}$ and $p_{\Ft}$ over the sets $\codeset_{\one}$ and $\codeset_{\two}$, respectively, and
\item A collection of decoding maps, $\phi_{\fo,\ft}:\mathcal{Z}^n\rightarrow(\mathcal{M}_{\one}\times\mathcal{M}_{\two})\cup\{\oneb,\, \twob\},$ where $\fo\in \codeset_{\one}$ and $\ft\in\codeset_{\two}$. 
\end{enumerate}
\end{defn}
As before, the decoder outputs the symbol \oneb (resp., \twob) to declare that user \one (resp., \two) is adversarial. 
The {\em average probability of error} $P_{e}$ is the maximum of the average probabilities of error in the following three cases: (1) both users are honest, (2) user \one is adversarial, and (3) user \two is adversarial. When both users are honest, the error occurs if the decoder outputs anything other than the pair of correct messages. Let 
\[\cE^{\fo,\ft}_{\mo,\mt} = \inb{\vecz:\phi_{\fo,\ft}(\vecz)\neq(\mo, \mt)}\]
denote the error event when both users are honest. 
The average error {probability} when both users are honest is
\begin{align}
&P^{\rand}_{e,\na}\hspace{-0.25em} \defineqq
    \frac{1}{N_{\one}\cdot N_{\two}} 
        \sum_{\substack{\mo\in\mathcal{M}_{\one}\\\mt\in\mathcal{M}_{\two}}}
            \sum_{\substack{\fo\in\codeset_{\one}\\\ft\in\codeset_{\two}}}
                p_{\Fo}(\fo)p_{\Ft}(\ft)W^n\inp{\cE^{\fo,\ft}_{\mo,\mt}|\fo(\mo), \ft(\mt)}.
\end{align}
Recall that when user \one is adversarial we would like the decoder's output to either be the symbol $\oneb$ or a pair of messages of which the message of user \two is correct. The error event is $\cE^{\fo,\ft}_{\mt} \defineqq \inb{\vecz:\phi_{\fo,\ft}(\vecz)\notin\inp{\cM_{\one}\times\{\mt\}}\cup{\{\oneb\}}}$. The average probability of error when user \one is adversarial is 
\begin{align}
&P^{\rand}_{e,\malone} \defineqq 
        \max_{\substack{\vecx\in\cX^n\\\fo\in\codeset_{\one}}} 
            \left(\frac{1}{N_{\two}}\sum_{m_{\two}\in \mathcal{M}_{\two}}
                \sum_{\ft\in\codeset_{\two}}
                    p_{\Ft}(\ft)W^n\inp{\cE^{\fo,\ft}_{\mt}|\vecx, \ft(\mt)}\right).\label{eq:rand-mal1}
\end{align}
Notice that the adversarial user \one selects both its transmission $\vecx$ and the code $\fo$ (i.e., the realization of the randomness it shares with the decoder) adversarially in order to maximize the probability that the decoder (working with this code $\fo$) makes an error. 
As in the deterministic coding case, the probability of error is maximized by deterministic attacks of the adversary. This is because for any attack distribution $q_{\vecX, F_{\one}}$,
\begin{align}
\sum_{\vecx, f_{\one}}q_{\vecX, F_{\one}}(\vecx, f_{\one})\left(\frac{1}{N_{\two}}\sum_{m_{\two}\in \mathcal{M}_{\two}}
                \sum_{\ft\in\codeset_{\two}}
                    p_{\Ft}(\ft)W^n\inp{\cE^{\fo,\ft}_{\mt}|\vecx, \ft(\mt)}\right)\leq P^{\rand}_{e,\malone}.
\end{align}
Similarly, for $\cE^{\fo,\ft}_{\mo} \defineqq \inb{\vecz:\phi_{\fo,\ft}(\vecz)\notin\inp{\{\mo\}\times\cM_{\two}}\cup{\{\twob\}}}$, the average probability of error when user \two is adversarial is 
\begin{align}
&P^{\rand}_{e,\maltwo} \defineqq 
    \max_{\substack{\vecy\in\cY^n\\\ft\in\codeset_{\two}}}
        \left(\frac{1}{N_{\one}}\sum_{m_{\one}\in \mathcal{M}_{\one}}
                \sum_{\fo\in\codeset_{\one}}
                    p_{\Fo}(\fo)W^n\inp{\cE^{\fo,\ft}_{\mo}|\fo(\mo), \vecy}\right).\label{eq:rand-mal2}
\end{align}
We define the \emph{average probability of error} as
\begin{align*}
        P^{\rand}_{e}\defineqq \max{\inb{P^{\rand}_{e,\na},P^{\rand}_{e,\malone},P^{\rand}_{e,\maltwo}}}.
\end{align*}
The capacity region $\cC^{\rand}$ of communication with adversary identification under randomized coding may be defined along the lines of Definition~\ref{defn:capacity}: we say that a rate $(R_{\one},R_{\two})$ is achievable if there exists a sequence of $(\lfloor2^{nR_{\one}}\rfloor,\lfloor2^{nR_{\two}}\rfloor,L_{\one,n},L_{\two,n},n)$  randomized adversary identifying codes (for some $L_{\one, n}, L_{\two, n}$) such that $P^{\rand}_{e}\rightarrow0$ as $n\rightarrow\infty$. The capacity region $\cC^{\rand}$ is the closure of the set of all such achievable rate pairs.

\subsection{Dichotomy theorem for deterministic coding capacity region}\label{sec:radn_det_capacity_connection}
The following theorem states that when the deterministic coding capacity region $\cC$ has a non-empty interior, it is the same as the randomized coding capacity region $\cC^{\rand}$. This is analogous to similar results in the AVC literature~\cite[Theorem~1]{Ahlswede78},\cite[Theorem~1, Section IV]{Jahn81}.
\begin{thm}\label{thm:capacity_equivalence}
 $ \cC=\cC^{\rand}$ whenever $(R_{\one}, R_{\two})\in \cC$ for some $R_{\one}, R_{\two}>0$.
\end{thm}
Since deterministic codes are a subset of randomized codes, $\cC\subseteq \cC^{\rand}$. We formally show that $ \cC\supseteq \cC^{\rand}$ in Appendix~\ref{app:rand_reduc}. For this, we first prove  that given any randomized adversary identifying code with a small probability of error, there exists another randomized adversary identifying code, also with a small probability of error, which requires only $O(\log\inp{n})$ shared random bits between each (honest) sender and the receiver (Lemma~\ref{thm:rand_reduc} in Appendix~\ref{app:rand_reduc}). This {\em randomness reduction} argument is along the lines of Ahlswede~\cite{Ahlswede78} (and its extension for AV-MAC by Jahn~\cite{Jahn81}).

With this, one may construct a two-phased code for showing the achievability of Theorem~\ref{thm:capacity_equivalence}. In the first phase of the code, a positive-rate, deterministic adversary identifying code (guaranteed to exist by the hypothesis of the theorem) is used to send $O(\log\inp{n})$ random bits as messages by each (honest) sender to the receiver. This establishes the $O(\log\inp{n})$ amount of shared randomness required for the randomized adversary identifying code. In the second phase, the shared randomness based code obtained from the randomness reduction argument is used. Note that an adversarial user can maliciously select their input in the first phase. This is why in \eqref{eq:rand-mal1} and \eqref{eq:rand-mal2} we defined the probability of errors of the randomized adversary identifying code such that the adversarial user may select their own randomness. 

\subsection{An outer bound to the randomized coding capacity region}\label{sec:outer_bd_proof}
Recall that $\cC^{\rand}_{\AVMAC}(\tilde{\cW}_W)$ denote the randomized coding capacity for the AV-MAC $\tilde{\cW}_W$ (see Definition~\ref{defn:AVMAC}). Here, each user shares independent randomness with the receiver which is unknown to the adversary and the other user. 
\begin{thm}[Outer bound]\label{thm:rand-outer_bd}
$\cC^{\rand} \subseteq \cC^{\rand}_{\AVMAC}(\tilde{\cW}_W)$. Moreover, there exists 
an AV-MAC ${\cW}_W$ such that $\cC^{\rand}_{\AVMAC}({\cW}_W)=\cC^{\rand}_{\AVMAC}(\tilde{\cW}_W)$ 
and\footnote{The capacity region of an AV-MAC only depends on its convex hull which is defined by taking convex combinations of channels under different states \cite{Jahn81,AhlswedeC99}. The AV-MAC ${\cW}_W$ has the same convex hull as $\tilde{\cW}_W$ but with only finitely many states.} $|{\cW}_W| \leq 2^{|\cX|^2 + |\cY|^2}$.
\end{thm}
\begin{remark}\label{remark:outerbound_code} We will prove the outer bound for a weaker adversary  who cannot choose the realization of the randomness it shares with the decoder. Thus, the outer bound holds even if the randomization is provided by the decoder (or an external agent). Note that the adversary knows the realization of its randomness and may choose its input to the channel based on this.\footnote{This is the intermediate model described in  \cite[Footnote 10]{NehaBDP23} for the reliable communication problem.}
\end{remark}
\begin{proof}
Let $\inp{F_{\one}, F_{\two}, \phi_{F_{\one},F_{\two}}}$ be an $(\nummsg_{\one}, \nummsg_{\two}, \numrand_{\one}, \numrand_{\two}, n)$ randomized adversary identifying code
with $P^{\rand}_{e} \leq \epsilon$. 
For $i\in [1:n]$, let $(Q_{i, X'|X}, Q_{i, Y'|Y})$ be an arbitrary sequence of pairs of conditional distributions satisfying \eqref{eq:outer_bound} and define $\tilde{W}_{i}$ as 
\begin{align}
\tilde{W}_{i}(z|x,y) \defineqq \sum_{x'}Q_{i,X'|X}(x'|x)W(z|x',y) = \sum_{y'}Q_{i,Y'|Y}(y'|y)W(z|x,y'), x\in\cX, y\in\cY, z\in\cZ.\label{eq:thm:rand-outer_bd:eq1}
\end{align}
Let $Q_{\vecX'|\vecX}\defineqq \prod_{i=1}^{n}Q_{i,X'|X}$, $Q_{\vecY'|\vecY}\defineqq \prod_{i=1}^{n}Q_{i,Y'|Y}$ and $\tilde{W}^{(n)} \defineqq \prod_{i=1}^{n}\tilde{W}_{i}$.

\noindent Let $\cE^{\two,\mt}_{f_{\one}, f_{\two}} := \inb{\vecz:\phi_{f_{\one}, f_{\two}}(\vecz) \in \{\twob\}\cup\inp{\cM_{\one}\times(\cM_{\two}\setminus\{\mt\})}}$ and $\cE^{\one,\mo}_{f_{\one}, f_{\two}} := \inb{\vecz:\phi_{f_{\one}, f_{\two}}(\vecz) \in \{\oneb\}\cup\inp{(\cM_{\one}\setminus\{\mo\})\times\cM_{\two}}}$. Consider a malicious user-\one who chooses $M_{\one}$ uniformly from $\cM_{\one} = [1:N_{\one}]$, passes $F_{\one}(M_{\one})$ over $Q_{\vecX'|\vecX}$, and transmits the resulting vector, that is, the user sends a vector distributed as $Q_{\vecX'|\vecX}(\cdot|F_{\one}(M_{\one}))$ as input to the channel. Note that, as pointed out in Remark~\ref{remark:outerbound_code}, the randomness in encoder $F_{\one}$ is not chosen maliciously by user-\one. We may conclude from \eqref{eq:rand-mal1} and \eqref{eq:rand-mal2} that 
\begin{align*}
P^{\rand}_{e, \malone} \geq \frac{1}{N_{\one}\cdot N_{\two}}\sum_{\mo, \mt}\sum_{\vecx}\sum_{f_{\one}, f_{\two}}p_{F_{\one}}(f_{\one})p_{F_{\two}}(f_{\two})Q_{\vecX'|\vecX}(\vecx|f_{\one}(\mo))W^n\inp{\cE^{\two,\mt}_{f_{\one},f_{\two}}\Big|\vecx,f_{\two}(\mt)}. 
\end{align*} Similarly,     
\begin{align*}
P^{\rand}_{e, \maltwo} \geq \frac{1}{N_{\one}\cdot N_{\two}}\sum_{\mo, \mt}\sum_{\vecy}\sum_{f_{\one}, f_{\two}}p_{F_{\one}}(f_{\one})p_{F_{\two}}(f_{\two})Q_{\vecY'|\vecY}(\vecy|f_{\two}(\mt))W^n\inp{\cE^{\one,\mo}_{f_{\one}, f_{\two}}\Big|f_{\one}(\mo), \vecy}.
\end{align*}
Using these inequalities,
\begin{align}
2\epsilon\geq P^{\rand}_{e, \malone} + P^{\rand}_{e, \maltwo} &\geq \frac{1}{N_{\one}\cdot N_{\two}}\sum_{\mo, \mt}\sum_{f_{\one}, f_{\two}}p_{F_{\one}}(f_{\one})p_{F_{\two}}(f_{\two})\Bigg(\sum_{\vecx}Q_{\vecX'|\vecX}(\vecx|f_{\one}(\mo))W^n\inp{\cE^{\two,\mt}_{f_{\one}f_{\two}}\Big|\vecx,f_{\two}(\mt)}\nonumber\\
&\qquad\qquad\qquad\qquad\qquad\qquad\qquad+ \sum_{\vecy}Q_{\vecY'|\vecY}(\vecy|f_{\two}(\mt))W^n\inp{\cE^{\one,\mo}_{f_{\one}, f_{\two}}\Big|f_{\one}(\mo), \vecy}\Bigg)\nonumber\\
&= \frac{1}{N_{\one}\cdot N_{\two}}\sum_{\mo, \mt}\sum_{f_{\one}, f_{\two}}p_{F_{\one}}(f_{\one})p_{F_{\two}}(f_{\two})\Bigg(\tilde{W}^n\inp{\cE^{\two,\mt}_{f_{\one}f_{\two}}\Big|f_{\one}(\mo),f_{\two}(\mt)}\nonumber\\
&\qquad\qquad\qquad\qquad\qquad\qquad\qquad\qquad\qquad\qquad+ \tilde{W}^n\inp{\cE^{\one,\mo}_{f_{\one}, f_{\two}}\Big|f_{\one}(\mo), f_{\two}(\mt)}\Bigg)\nonumber\\
& \geq \frac{1}{N_{\one}\cdot N_{\two}}\sum_{\mo, \mt}\sum_{f_{\one}, f_{\two}}p_{F_{\one}}(f_{\one})p_{F_{\two}}(f_{\two})\tilde{W}^{(n)}\inp{\cE^{\one,\mo}_{f_{\one}, f_{\two}}\cup\cE^{\two,\mt}_{f_{\one},f_{\two}}\Big|f_{\one}(\mo),f(\mt)},\label{eq:thm:rand-outer_bd:eq2}
\end{align} where we use \eqref{eq:thm:rand-outer_bd:eq1} and $\tilde{W}^{(n)} = \prod_{i=1}^{n}\tilde{W}_{i}$ in the penultimate step.
Notice that for $(\mo,\mt)\in \cM_{\one}\times\cM_{\two}$, $\inp{f_{\one}, f_{\two}}\in \Gamma_{\one}\times\Gamma_{\two}$ and $\vecz\in\cZ^n$,
\begin{align*}
&\inb{\vecz:\phi_{f_{\one},f_{\two}}(\vecz)\neq(\mo, \mt)}\\
&\quad= \inb{\vecz:\phi_{f_{\one},f_{\two}}(\vecz)\in\{\oneb\}}\bigcup\inb{\vecz:\phi_{f_{\one},f_{\two}}(\vecz)\in\{\twob\}}\bigcup \inb{\vecz:\phi_{f_{\one},f_{\two}}(\vecz) \in \inp{\cM_{\one}\times\cM_{\two}}\setminus\{(\mo, \mt)\}}\\
&\quad= \inb{\vecz:\phi_{f_{\one},f_{\two}}(\vecz) \in \{\twob\}\cup\inp{\cM_{\one}\times(\cM_{\two}\setminus\{\mt\})}} \bigcup
\inb{\vecz:\phi_{f_{\one},f_{\two}}(\vecz) \in \{\oneb\}\cup\inp{(\cM_{\one}\setminus\{\mo\})\times\cM_{\two}}}\\
&\quad= \cE^{\two,\mt}_{f_{\one},f_{\two}}\cup\cE^{\one,\mo}_{f_{\one}, f_{\two}}.
\end{align*}
Thus, from \eqref{eq:thm:rand-outer_bd:eq2},
\begin{align*}
&\frac{1}{N_{\one}\cdot N_{\two}}\sum_{\mo, \mt}\sum_{f_{\one}, f_{\two}}p_{F_{\one}}(f_{\one})p_{F_{\two}}(f_{\two})\tilde{W}^{(n)}\inp{\inb{\vecz:\phi_{f_{\one},f_{\two}}(\vecz) \neq (\mo, \mt)}|f_{\one}(\mo),f_{\two}(\mt)} \\
&=\frac{1}{N_{\one}\cdot N_{\two}}\sum_{\mo, \mt}\sum_{f_{\one}, f_{\two}}p_{F_{\one}}(f_{\one})p_{F_{\two}}(f_{\two})\tilde{W}^{(n)}\inp{\cE^{\one,\mo}_{f_{\one}, f_{\two}}\cup\cE^{\two,\mt}_{f_{\one},f_{\two}}\Big|f_{\one}(\mo),f_{\two}(\mt)}\\
&\leq  2\epsilon.
\end{align*}
Recall that every pair $(Q_{X'|X}, Q_{Y'|Y})$ satisfying \eqref{eq:outer_bound} corresponds to an element in $\tilde{\cW}_W$ which is a convex set (see the discussion in Section~\ref{sec:outer_bd}). Thus, for all $\epsilon>0$, any adversary identifying code for the \bmac $W$ with probability of error $\epsilon$ is also a communication code for the AV-MAC $\tilde{\cW}_{W}$ with probability of error at most $2\epsilon$. So, the randomized coding capacity region $\cC^{\rand}$ of $W$ is contained by the randomized coding capacity region $\cC^{\rand}_{\AVMAC}(\tilde{\cW}_W)$ of the AV-MAC $\tilde{\cW}_{W}$. 

Since the capacity region of an AV-MAC, both under randomized codes and deterministic codes,  only depends on its convex hull \cite{Jahn81,AhlswedeC99} \blue{(also see Section~\ref{sec:AVMAC})}, we may use a similar argument as the one used in the proof of Theorem~\ref{thm:outer_bd} to conclude that there exists 
an AV-MAC ${\cW}_W$ such that $\cC^{\rand}_{\AVMAC}({\cW}_W)=\cC^{\rand}_{\AVMAC}(\tilde{\cW}_W)$ 
and $|{\cW}_W| \leq 2^{|\cX|^2 + |\cY|^2}$.
\end{proof}

\section{Summary and discussions}\label{sec:summary}
\blue{In this work, we introduce the problem of communication with adversary identification in a \bmac. This is the second part of our study of different decoding guarantees in a \bmac, which models a Multiple Access Channel where users may maliciously deviate from the protocol. We study reliable communication in the first part~\cite{NehaBDP23}, where we require the decoder to always output correct messages for honest (non-malicious) senders. 
The present model relaxes the decoding guarantee in the presence of a malicious sender by allowing the decoder to identify the adversary, without decoding the message of the honest sender. A previous work \cite{NehaBDPISIT19} further relaxes the decoding guarantee in the presence of malicious senders, by allowing the decoding to declare adversarial interference, without decoding the messages of honest senders or identifying the malicious user. }

\blue{We study the present model for a two user \bmac. We characterize the feasibility conditions and provide inner and outer bound to the deterministic coding capacity region. We also show connections between the deterministic coding capacity region and randomized coding capacity region. In particular, Theorem~\ref{thm:capacity_equivalence} shows that whenever the deterministic coding capacity region has a non-empty interior, it is the same as the randomized coding capacity region. This is analogous to similar results in the AVC literature~\cite[Theorem~1]{Ahlswede78},\cite[Theorem~1, Section IV]{Jahn81}. We also contrast the present model with models of reliable communication and authentication communication through examples in Section~\ref{sec:comparison}.} 

\blue{We believe that our work leads to several intriguing directions for future research. Characterizing the exact capacity region for our problem remains an open question. Another open question is the generalization to more than two users. This was done for the problem of reliable communication in \cite{NehaBDP23}. 
All models of \bmac are studied without input constraints. It would be interesting to see how the results change with input constraints. We believe that results on AVCs and AV-MACs with state constraints can provide guidance in this. }

\appendices
\section{Proof of Theorem~\ref{thm:main_result}}\label{sec:proof_thm1}
\olive{We start with a proof of the converse.
We prove it  for the stronger version when both encoders can privately randomize.
\begin{lemma}\label{thm:converse}
 If a channel is \one-spoofable (resp. \two-\spoofable), then  for any $(N_{\one}, N_{\two}, n)$ adversary identifying code with $N_{\one}\geq 2$ (resp. $N_{\two}\geq 2$), the probability of error is at least $1/12$.  
 \end{lemma}}
\begin{proof}
\olive{The proof uses ideas from the proof of \cite[Lemma 1, page 187]{CsiszarN88}. Suppose the channel is \one-spoofable such that $Q_{Y|\tilde{X} \tilde{Y}}$ and $Q_{X|\tilde{X}X'}$ are conditional distributions satisfying \eqref{eq:spoof1}. Let $(F_{\one},F_{\two}, \phi)$ be a given  $(N_{\one}, N_{\two},n) $ code  where $F_{\one}: [1:N_{\one}]\rightarrow\cX^n$ and $F_{\two}:[1:N_{\two}]\rightarrow \cY^n$ are (privately) randomized maps and $\phi:\cZ^n\rightarrow \cM_{\one}\times\cM_{\two}\cup\inb{\oneb, \twob}$ is a deterministic map\footnote{\blue{In the AVC literature, a code where encoders have private randomness and decoders are deterministic, is called a stochastic code \cite[Chapter 12]{CsiszarKorner}. Here, we show the converse for stochastic codes.}}. We can define the probability of error for encoders with private randomness in a similar fashion as defined for the deterministic code in \eqref{eq:na}-\eqref{eq:mal2}. Recall that $\cE_{\mo,\mt} = \inb{\vecz:\phi(\vecz)\neq(\mo, \mt)}$, 
$\cE^{\two}_{\mt} = \inb{\vecz:\phi_{\two}(\vecz)\notin\{\mt, \oneb\}}$ and $\cE^{\one}_{\mo} = \inb{\vecz:\phi_{\one}(\vecz)\notin\{\mo,\twob\}}$.
When both users are honest, we define
\begin{align}
&P^{\textup{pvt}}_{e,\na}\hspace{-0.25em} \defineqq\frac{1}{N_{\one}\cdot N_{\two}} \sum_{\substack{(\mo, \mt)\\ \in\mathcal{M}_{\one}\times\mathcal{M}_{\two}}}\sum_{\vecx,\vecy}\bbP\inp{F_{\one}(\mo)=\vecx}\bbP\inp{F_{\two}(\mt) = \vecy}W^n\inp{\cE_{\mo,\mt}|\vecx, \vecy}.\label{eq:napvt}
\end{align} The probability in the terms $\bbP\inp{F_{\one}(\mo)=\vecx}$ and $\bbP\inp{F_{\two}(\mt) = \vecy}$ is over the randomness of the encoders $F_{\one}$ and $F_{\two}$ respectively.}

\olive{When a user is malicious, we define the probabilities of error under randomized attacks for convenience. Along the lines of Remark~\ref{remark:det_rand_attacks} (which considered a deterministic code), the probability of error is the same under both randomized and deterministic attacks even for privately randomized encoders. When user \one is malicious, we define 
\begin{align}
&P^{\textup{pvt}}_{e,\malone} \defineqq \max_{P_{\vecX}} \sum_{\vecx}P_{\vecX}(\vecx)\inp{\frac{1}{N_{\two}}\sum_{m_{\two}\in \mathcal{M}_{\two}}\sum_{{\vecy}}\bbP\inp{F_{\two}(\mt) = {\vecy}}W^n\inp{\cE^{\two}_{\mt}|\vecx, {\vecy}}}.\label{eq:mal1pvt}
\end{align}  Here, the maximization is over all randomized attacks $P_{\vecX}$ distributed on $\cX^n$. The probability in the term $\bbP\inp{F_{\two}(\mt) = {\vecy}}$ is over the randomness of the encoder $F_{\two}$.
Similarly,  the average probability of error when user \two is adversarial is 
\begin{align}
&P^{\textup{pvt}}_{e,\maltwo} \defineqq \max_{P_{\vecY}}\sum_{\vecy}P_{\vecY}(\vecy)\left(\frac{1}{N_{\one}}\sum_{m_{\one}\in \mathcal{M}_{\one}}\sum_{\vecx}\bbP\inp{F_{\one}(\mo) = \vecx}W^n\inp{\cE^{\one}_{\mo}|\vecx, \vecy}\right).\label{eq:mal2pvt}
\end{align}
We define the \emph{average probability of error} as
\begin{align*}
P^{\textup{pvt}}_{e}(F_{\one}, F_{\two}, \phi)\defineqq \max{\inb{P^{\textup{pvt}}_{e,\na},P^{\textup{pvt}}_{e,\malone},P^{\textup{pvt}}_{e,\maltwo}}}.
\end{align*}  
For the rest of the proof, we will use the notation $\vecz = (z_1, z_2, \ldots, z_n)$, $\vecx' = (x'_1, x'_2, \ldots, x'_n)$, $\tilde{\vecx} = (\tilde{x}_1, \tilde{x}_2, \ldots, \tilde{x}_n)$ and $\tilde{\vecy} = (\tilde{y}_1, \tilde{y}_2, \ldots, \tilde{y}_n)$ for $n$-length vectors $\vecz$, $\vecx'$, $\tilde{\vecx}$ and $\tilde{\vecy}$.}

\olive{Consider the following scenarios for $i,j\in\cM_{\one}$, $k\in\cM_{\two}$ and independent encoders $F_{\one}, F'_{\one}$ and $F_{\two}$, where $F_{\one}$ and $ F'_{\one}$ are two independent copies  of user \one's encoder:
\begin{enumerate}
	\item[(i)] User \one~sends input to the channel according to $F_{\one}(i)$. User $B$ uses an independent copy $F'_{\one}$ of user \one's encoder. 
	The input of user~\two to the channel is produced by passing $(F'_{\one}\inp{j}, F_{\two}(k))$ through $Q^n_{Y|\tilde{X} \tilde{Y}}$. For $\vecz\in \cZ^n$, the output distribution of the channel (denoted by $P_{i,j,k}(\vecz)$) is given by 
	\begin{align}
	P_{i,j,k}(\vecz) \defineqq \sum_{\vecx', \tilde{\vecx}, \tilde{\vecy}}\bbP\inp{F_{\one}(i) = \vecx'}\bbP\inp{F'_{\one}\inp{j} = \tilde{\vecx}}\bbP\inp{ F_{\two}(k) = \tilde{\vecy}}\inp{\prod_{t=1}^n\sum_{y\in\cY}Q_{Y|\tilde{X} \tilde{Y}}(y|\tilde{x}_t,\tilde{y}_{t})W_{Z|XY}(z_t|x'_t,y)}.\label{eq:dist1}
	\end{align}
	\item[(ii)] User~\two sends input according to $F_{\two}(k)$. The input of user~\one to the channel is produced by passing $(F_{\one}(i), F'_{\one}\inp{j})$ through $Q^n_{X|\tilde{X}, X'}$.  For $\vecz\in \cZ^n$, the output distribution of the channel (denoted by $Q_{i,j,k}(\vecz)$) is given by 
\begin{align}
	Q_{i,j,k}(\vecz) \defineqq \sum_{\vecx', \tilde{\vecx}, \tilde{\vecy}}\bbP\inp{F_{\one}(i) = \vecx'} \bbP\inp{F'_{\one}\inp{j} = \tilde{\vecx}}\bbP\inp{ F_{\two}(k) = \tilde{\vecy}}\inp{\prod_{t=1}^n\sum_{x\in\cX}Q_{X|\tilde{X} X'}(x|x'_t,\tilde{x}_{t})W_{Z|XY}(z_t|x,\tilde{y}_{t})}.\label{eq:dist2}
	\end{align}
\end{enumerate}
By \eqref{eq:spoof1} (also see Fig.~\ref{fig:spoof1}), we see that for all $i,j\in \cM_{\one}, k\in \cM_{\two}$ and $\vecz\in \cZ^n$, 
\begin{equation} P_{i,j,k}(\vecz)=P_{j,i,k}(\vecz)=Q_{i,j,k}(\vecz). \label{eq:feas_conv_eq}\end{equation}}

\olive{In scenario~(i), suppose user~\one chooses $i\in\cM_{\one}$ uniformly at random (and is hence honest) and, independently, the adversarial user~\two chooses $(j,k)\in\cM_{\one}\times\cM_{\two}$ uniformly at random. Then, from \eqref{eq:mal2pvt} and \eqref{eq:dist1}, we see that
\begin{align*}
P^{\textup{pvt}}_{e,\maltwo}&\geq \frac{1}{N^2_{\one}\times N_{\two}}\sum_{i,j\in\cM_{\one}}\sum_{k\in \cM_{\two}}\sum_{\vecz:\phi_{\one}(\vecz)\notin\{i,\twob\}}P_{i,j,k}(\vecz).
\end{align*} 
Interchanging the roles of $i$ and $j$, we have  
\begin{align*}
P^{\textup{pvt}}_{e,\maltwo}&\geq \frac{1}{N^2_{\one}\times N_{\two}}\sum_{i,j\in\cM_{\one}}\sum_{k\in \cM_{\two}}\sum_{\vecz:\phi_{\one}(\vecz)\notin\{j,\twob\}}P_{j,i,k}(\vecz).
\end{align*} 
In scenario~(ii), suppose user~\two chooses $k\in\cM_{\two}$ uniformly at random (and hence is honest), while, independently, the adversarial user~\one chooses $(i,j)\in \cM_{\one}^2$ uniformly at random. Using \eqref{eq:mal1pvt} and \eqref{eq:dist2}, we obtain
\begin{align*}
P^{\textup{pvt}}_{e,\malone}&\geq \frac{1}{N^2_{\one}\times N_{\two}}\sum_{i,j\in\cM_{\one}}\sum_{k\in \cM_{\two}}\sum_{\vecz:\phi_{\two}(\vecz)\notin\{k,\oneb\}}Q_{i,j,k}(\vecz).
\end{align*}
\noindent Thus,
\begin{align*}
3&P^{\textup{pvt}}_{e}(F_{\one},F_{\two},\phi)\geq P^{\textup{pvt}}_{e,\maltwo}+P^{\textup{pvt}}_{e,\maltwo}+P^{\textup{pvt}}_{e,\malone}\\
&\geq \frac{1}{N^2_{\one}\times N_{\two}}\sum_{i,j\in\cM_{\one}}\sum_{k\in \cM_{\two}}\inp{\sum_{\vecz:\phi_{\one}(\vecz)\notin\{i,\twob\}}P_{i,j,k}(\vecz) + \sum_{\vecz:\phi_{\one}(\vecz)\notin\{j,\twob\}}P_{j,i,k}(\vecz) +\sum_{\vecz:\phi_{\two}(\vecz)\notin\{k,\oneb\}}Q_{i,j,k}(\vecz) }\\
&\stackrel{\text{(a)}}{=}\frac{1}{N^2_{\one}\times N_{\two}}\sum_{i,j\in\cM_{\one}}\sum_{k\in \cM_{\two}}\inp{\sum_{\vecz:\phi_{\one}(\vecz)\notin\{i,\twob\}}P_{i,j,k}(\vecz) + \sum_{\vecz:\phi_{\one}(\vecz)\notin\{j,\twob\}}P_{i,j,k}(\vecz) +\sum_{\vecz:\phi_{\two}(\vecz)\notin\{k,\oneb\}}P_{i,j,k}(\vecz) }\\
&\stackrel{\text{(b)}}{\geq}\frac{1}{N^2_{\one}\times N_{\two}}\sum_{i,j\in\cM_{\one}, i\neq j}\sum_{k\in \cM_{\two}}\inp{\sum_{\vecz\in\cZ^n}P_{i,j,k}(\vecz)}\\
&=\frac{N_{\one}(N_{\one}-1)N_{\two}}{2N^2_{\one}\times N_{\two}}\\
& = \frac{N_{\one}-1}{2N_{\one}}\\
&\geq \frac{1}{4}\qquad \text{ for $N_{\one}\geq 2$},
\end{align*}
where (a) follows by noting from \eqref{eq:feas_conv_eq} that $P_{i,j,k}(\vecz)=P_{j,i,k}(\vecz)=Q_{i,j,k}(\vecz)$ and (b) follows by noting that $\inb{\vecz:\phi_{\one}(\vecz)\notin\{i,\twob\}}\cup \inb{\vecz:\phi_{\one}(\vecz)\notin\{j,\twob\}}\cup\inb{\vecz:\phi_{\two}(\vecz)\notin\{k,\oneb\}} = \inb{\vecz:\phi_{\one}(\vecz)\neq\twob}\cup\inb{\vecz:\phi_{\two}(\vecz)\notin\{k,\oneb\}}\supseteq \inb{\vecz:\phi_{\one}(\vecz)\neq\twob}\cup\inb{\vecz:\phi_{\two}(\vecz)=\twob} = \inb{\vecz:\phi_{\one}(\vecz)\neq\twob}\cup\inb{\vecz:\phi_{\one}(\vecz)=\twob}= \cZ^n.$ Here, the second last equality follows by recalling from \eqref{eq:feas_conv_dec1} and \eqref{eq:feas_conv_dec2} that if $\phi_{\two}(\vecz) = \twob$, then $\phi_{\one}(\vecz) = \twob$. Thus, for any given code $(F_{\one}, F_{\two}, \phi)$, for an \one-spoofable channel, $P^{\textup{pvt}}_e(F_{\one},F_{\two}, \phi)\geq \frac{1}{12}$.
A similar analysis follows when the channel is \two-\spoofable.}
\end{proof}

\olive{Next, for the proof of achievability of Theorem~\ref{thm:main_result}, we first state a {\em codebook lemma} which will be used to show all our achievability results. This gives a randomly generated codebook which satisfies certain properties. The technical proof of the lemma, which is along the lines of that of \cite[Lemma 3]{CsiszarN88}, is in Appendix~\ref{sec:codebooklemma}. The lemma can be thought of as a generalization of \cite[Lemma 3]{CsiszarN88} for two users. In particular, \eqref{codebook:1} is similar to \cite[Lemma 3, (3.2)]{CsiszarN88}, \eqref{codebook:2b} and \eqref{codebook:4} are generalizations of \cite[Lemma 3, (3.3)]{CsiszarN88}  and \eqref{codebook:3b} and \eqref{codebook:5} are generalizations of \cite[Lemma 3, (3.1)]{CsiszarN88} for a pair of messages.
\begin{lemma}[Codebook lemma]
\label{lemma:codebook}
	Suppose $\mathcal{X,Y,Z}$ are finite. Let $P_{\one}\in \cP^n_{\cX}$ and $P_{\two}\in \cP^n_{\cY}$. For any $\epsilon>0$, there exists $n_0(\epsilon)$ such that for all $n\geq n_0(\epsilon),\, N_{\one}, N_{\two}\geq \exp(n\epsilon)$, there are codewords $\vecx_1, \vecx_2, \ldots, \vecx_{N_{\one}}$ of type $P_{\one}$ and $\vecy_1, \vecy_2, \ldots, \vecy_{N_{\two}}$ of type $P_{\two}$ such that for all $\vecx, \vecx'\in \cX^n$  and $\vecy, \vecy'\in \cY^n$, and joint types $P_{X\tilde{X}\tilde{Y}Y}\in \cP^n_{\cX\times \cX\times \cY\times \cY}$ and $P_{X'\tilde{Y}_1\tilde{Y}_2Y'}\in \cP^n_{\cX\times \cY\times \cY\times \cY}$, 
	 and for $R_{\one}\defineqq (1/n)\log N_{\one}$ and $R_{\two}\defineqq (1/n)\log N_{\two}$, 
	 the following holds:\footnote{Note that $\exp$ and $\log$ are with respect to base 2.}
\begin{align}
&\blue{\frac{1}{N_{\one}}}\left|\inb{m_{\one}: (\vecx_{\mo}, \vecy)\in T_{XY}^n}\right|\leq \exp\inb{-n\epsilon/2},\text{ if }I(X;Y)> \epsilon;\label{codebook:1}\\
&\blue{\frac{1}{N_{\one}}}{\left|\inb{\mo:(\vecx_{\mo}, \vecx_{\tilde{m}_{\one}}, \vecy_{\mt}, \vecy)\in T^n_{X\tilde{X}\tilde{Y}Y} \text{ for some }\tilde{m}_{\one}\neq \mo\text{ and some }\mt} \right|} \leq \exp\inb{-n\epsilon/2}, \nonumber\\
& \qquad \text{ if }I(X;\tilde{X}\tilde{Y}Y)-|R_{\one}- I(\tilde{X};\tilde{Y}Y)|^{+}-|R_{\two}-I(\tilde{Y};Y)|^{+} >\epsilon;\label{codebook:2b}\\
&\left|\inb{(\tilde{m}_{\one}, \tilde{m}_{\two}):(\vecx, \vecx_{\tilde{m}_{\one}}, \vecy_{\tilde{m}_{\two}}, \vecy)\in T^n_{X\tilde{X}\tilde{Y}Y}}\right|\nonumber\\
&\qquad\leq \exp\inb{n\inp{|R_{\one}- I(\tilde{X};\tilde{Y}XY)|^{+}+|R_{\two}-I(\tilde{Y};XY)|^{+}+\epsilon}};\label{codebook:3b}\\
&\blue{\frac{1}{N_{\one}}}{\left|\inb{\mo:(\vecx_{\mo}, \vecy_{\tilde{m}_{\two 1}}, \vecy_{\tilde{m}_{\two 2}}, \vecy')\in T^n_{X'\tilde{Y}_1\tilde{Y}_2Y'} \text{ for some }\tilde{m}_{\two 1},\tilde{m}_{\two 2}, } \right|} \leq \exp\inb{-n\epsilon/2}, \nonumber\\
& \qquad \text{ if }I(X';\tilde{Y}_1\tilde{Y}_2Y')-|R_{\two}-I(\tilde{Y}_1;Y')|^{+}-|R_{\two}-I(\tilde{Y}_2;\tilde{Y}_1 Y')|^{+} >\epsilon;\label{codebook:4}\\
&\text{and }\left|\inb{(\tilde{m}_{\two 1}, \tilde{m}_{\two 2}):(\vecx', \vecy_{\tilde{m}_{\two 1}}, \vecy_{\tilde{m}_{\two 2}}, \vecy')\in T^n_{X'\tilde{Y}_1\tilde{Y}_2Y'} }\right|\nonumber\\
&\qquad\leq \exp\inb{n\inp{|R_{\two}-I(\tilde{Y}_1;X'Y')|^{+}+|R_{\two}-I(\tilde{Y}_2;\tilde{Y}_1 X'Y')|^{+}+\epsilon}},\label{codebook:5}
\end{align}
and statements analogous to \eqref{codebook:1}-\eqref{codebook:5} with the roles of users \one and \two are interchanged.\footnote{See Appendix~\ref{sec:codebooklemma} for the full statement.}
\end{lemma}
\noindent With this, we are ready to prove the achievability of Theorem~\ref{thm:main_result}.
\begin{lemma}\label{thm:detCodesPositivity}
The rate region for deterministic codes is non-empty if the channel is non-\spoofable.
\end{lemma}}
\begin{proof}
\blue{Fix some $P_{\one}$ and $P_{\two}$ satisfying $\min_{x\in \cX}P_{\one}(x)\geq\alpha$ and $\min_{y\in \cY}P_{\two}(y)\geq\alpha$ respectively for some $\alpha>0$ and $\eta>0$ which is a function of $\alpha$ and the channel \footnote{We need  $\alpha>0$ to show Lemma~\ref{lemma:disambiguity}. This is similar to the requirement $\beta>0$ in \cite[Lemma~4]{CsiszarN88}.}. We also fix  $\epsilon>0$   $n>n_0(\epsilon)$, $\delta>\epsilon$  and $R_{\one}=R_{\two}=\delta$ such that $\eta> 3\epsilon + 4\delta$ holds. In particular, we may choose $\epsilon = \delta = \eta/8$. }\olive{The codebook is given by Lemma~\ref{lemma:codebook}. }

\olive{
\noindent {\em Encoding.} 
	Let $N_{\one} = 2^{nR_{\one}}$, $N_{\two} = 2^{nR_{\two}}$,  $\cM_{\one} = \inb{1,\ldots, N_{\one}}$ and $\cM_{\two} = \inb{1, \ldots, N_{\two}}$. For $\mo \in \cM_{\one}$, $f_{\one}(\mo) = \vecx_{\mo}$ and for $\mt \in \cM_{\two}$, $f_{\two}(\mt) = \vecy_{\mt}$.\\
{\em Decoding.} 
Let $\cD_{\eta}$ be the set of joint distributions defined as
\begin{align}
\blue{\cD_\eta} \defineqq  \big\{P_{XYZ}\in \cP^n_{\cX\times \cY \times \cZ}: \, D\inp{P_{XYZ}||P_XP_YW}\leq\eta \big\}.\label{eq:D_eta}
\end{align}
\begin{defn}[$D_{\one}(\eta, \vecz)$]\label{D_eta_z}
For the given codebook, the parameter $\eta>0$ and the received channel output sequence $\vecz$, let $D_{\one}(\eta, \vecz)$ be  the set of messages $m_{\one}\in\mathcal{M}_{\one}$ such that  there exists $\vecy\in \cY^n$ satisfying the following conditions:\footnote{In {\em\ref{check:1}}, when we write $\inp{f_{\one}(m_{\one}), \vecy, \vecz} \in T^{n}_{XYZ}$, we implicitly define the joint distribution $P_{XYZ}$ associated with the type $T^n_{XYZ}$ to be the joint empirical type of $\inp{f_{\one}(m_{\one}), \vecy, \vecz}$. This convention is followed throughout the paper.}
\begin{enumerate}[label=(\roman*)]
\item $\inp{f_{\one}(m_{\one}), \vecy, \vecz} \in T^{n}_{XYZ}$ such that $P_{XYZ}\in \cD_{\eta}.$ \label{check:1}
\item {For every} $(\tilde{m}_{\one},\, \tilde{m}_{\two})\in \cM_{\one}\times \cM_{\two}, \, \tilde{m}_{\one}\neq m_{\one}$ and  $(\vecy', \, \vecx')\in \cY^n\times \cX^n$ such that $\inp{f_{\one}(\mo), \vecy, f_{\one}(\tilde{m}_{\one}), \vecy',\vecx', f_{\two}(\tilde{m}_{\two}), \vecz}\in T^n_{XY\tilde{X}Y'X'\tilde{Y}Z}$ with $P_{\tilde{X}Y'Z}\in \cD_{\eta}$ and  $P_{X'\tilde{Y}Z}\in \cD_{\eta}$, {we have} $I(\tilde{X}\tilde{Y};XZ|Y)<\eta$ for  $P_{XY\tilde{X}\tilde{Y}Z}$.\label{check:2}
\item {For every} $\tilde{m}_{\two 1},\, \tilde{m}_{\two 2}\in \cM_{\two}$ where $\tilde{m}_{\two 1}\neq\tilde{m}_{\two 2}$ and $\vecx'_1, \, \vecx'_2\in \cX^n$ such that $\inp{f_{\one}(\mo), \vecy,\vecx'_1, f_{\two}(\tilde{m}_{\two 1}),\vecx'_2, f_{\two}(\tilde{m}_{\two 2}),\vecz}\in T^n_{XYX'_1\tilde{Y}_1X'_2\tilde{Y}_2Z}$ with $P_{X'_1\tilde{Y}_1Z}\in \cD_{\eta}$ and  $P_{X'_2\tilde{Y}_2Z}\in \cD_{\eta}$, {we have} $I(\tilde{Y}_1\tilde{Y}_2;XZ|Y)< \eta$ for $P_{XY\tilde{Y}_1\tilde{Y}_2Z}$.\label{check:3}
\end{enumerate}
\end{defn}
We define $D_{\two}(\eta, \vecz)$ analogously (by interchanging the roles of user \one and \two). The output of the decoder  is as follows: 
\begin{align*}
\phi(\vecz) \defineqq\begin{cases}(\mo,\mt), &\text{ if }D_{\one}(\eta, \vecz)\times D_{\two}(\eta, \vecz) = \{(\mo,\mt)\}\\ \oneb, &\text{ if }|D_{\one}(\eta, \vecz)| = 0, \, |D_{\two}(\eta, \vecz)| \neq 0\\\twob, &\text{ if }|D_{\two}(\eta, \vecz)| = 0, \, |D_{\one}(\eta, \vecz)| \neq 0\\(1,1)&\text{ otherwise.}
\end{cases}
\end{align*}
The last of the above cases (``otherwise'') occurs when either of the following two events occur: {\em (i)} $|D_{\two}(\eta, \vecz)| =|D_{\one}(\eta, \vecz)| = 0$, {\em (ii)} $|D_{\two}(\eta, \vecz)| \geq 1, \, |D_{\one}(\eta, \vecz)| \geq  1$ and $|D_{\two}(\eta, \vecz)|+|D_{\one}(\eta, \vecz)| \geq 3$, that is, both $D_{\two}(\eta, \vecz)$ and  $D_{\one}(\eta, \vecz)$ are non-empty and at least one of these two sets has two or more elements. As we will see, the first event is an atypical event and hence will occur with a vanishing probability.
The following lemma (proved in Appendix~\ref{app:disambiguity}) implies that the second event cannot occur  for non-spoofable channels if $\eta>0$ is chosen sufficiently small. 
\begin{lemma}\label{lemma:disambiguity}
{Suppose $\alpha>0$. For a channel which is not \one-spoofable, for sufficiently small $\eta>0$,} there does not exist a distribution $P_{XY\tilde{X}Y'X'\tilde{Y}Z}\in \cP^n_{XY\tilde{X}Y'X'\tilde{Y}Z}$ with $\min_{x}P_X(x), \min_{\tilde{x}}P_{\tilde{X}}(\tilde{x}), \min_{\tilde{y}}P_{\tilde{Y}}(\tilde{y})\geq \alpha$  which  satisfies the following:
\begin{enumerate}[label=(\Alph*)]
	\item $P_{XYZ}\in \blue{\cD_{\eta}}$,\label{disamb:1}
	\item $P_{\tilde{X}Y'Z}\in \blue{\cD_{\eta}}$,\label{disamb:2}
	\item $P_{X'\tilde{Y}Z}\in \blue{\cD_{\eta}}$,\label{disamb:3}
	\item $I(\tilde{X}\tilde{Y};XZ|Y)<\eta$,\label{disamb:4}
	\item $I(X\tilde{Y};\tilde{X}Z|Y')<\eta$ and\label{disamb:5}
	\item $I(X\tilde{X};\tilde{Y}Z|X')<\eta$.\label{disamb:6}
\end{enumerate}
\blue{The analogous} condition holds for a channel which is not \two-\spoofable.
\end{lemma}
For a non-spoofable channel, given $\alpha>0$ and small enough $\eta>0$, Lemma~\ref{lemma:disambiguity} implies that if $|D_{\one}(\eta, \vecz)|, \, |D_{\two}(\eta, \vecz)| \geq 1$, then $|D_{\two}(\eta, \vecz)|=|D_{\one}(\eta, \vecz)| = 1$. To see this, suppose $|D_{\one}(\eta, \vecz)|\geq 2$ and $|D_{\two}(\eta, \vecz)| \geq 1$. Let $m_{\one}, \tilde{m}_{\one}\in D_{\one}(\eta, \vecz)$ and $\tilde{m}_{\two}\in D_{\two}(\eta, \vecz)$. Then, Definition~\ref{D_eta_z} implies that there exist $\vecx$, $\vecy$ and $\vecy'$ such that for $\inp{f_{\one}(\mo), \vecy, f_{\one}(\tilde{m}_{\one}), \vecy',\vecx', f_{\two}(\tilde{m}_{\two}), \vecz}\in T^n_{XY\tilde{X}Y'X'\tilde{Y}Z}$, $P_{XYZ}\in \cD_{\eta}$, $P_{\tilde{X}Y'Z}\in \cD_{\eta}$, $P_{X'\tilde{Y}Z}\in \cD_{\eta}$, $I(\tilde{X}\tilde{Y};XZ|Y)<\eta$, $I(X\tilde{Y};\tilde{X}Z|Y')<\eta$ and $I(X\tilde{X};\tilde{Y}Z|X')<\eta$.  However, such joint distribution $P_{XY\tilde{X}Y'X'\tilde{Y}Z}$ cannot exist as per Lemma~\ref{lemma:disambiguity}. Analogously, for a channel which is not \two-\spoofable, $|D_{\one}(\eta, \vecz)|\geq 1$ and $|D_{\two}(\eta, \vecz)| \geq 2$ is not possible. \blue{As mentioned earlier,} we choose $\eta$, $\epsilon$ and $\delta$$(>\epsilon)$ small enough so that Lemma~\ref{lemma:disambiguity} holds and 
\begin{align}
\eta> 3\epsilon + 4\delta. \label{eq:eta_cond}
\end{align} 
To analyze the probability of error, we first recall from \eqref{honest_error_ub} that $P_{e,\na}\leq P_{e,\malone} + P_{e,\maltwo}$. So, we only need to analyse the case when one of the users is malicious. }

\blue{We will use the `method of types' for the analysis \cite{720546,CsiszarKorner}. Before we begin, we restate some basic properties from \cite[Chapter~2]{CsiszarKorner} that we need for the proof.}
\olive{Let $X$ and $Y$ be two jointly distributed random variables such that $P_{XY}\in \cP^n_{\cX\times\cY}$. For $\vecx\in T^n_X$, a  distribution $Q$ on $\cX$ and a discrete memoryless channel $U_{Y|X}$ from $\cX$ to $\cY$, 
\begin{align}
|\cP^n_{\cX}|&\leq (n+1)^{|\cX|},\label{eq:type_property1}\\
(n+1)^{-|\cX|}\exp\inp{nH(X)}&\leq |T^n_{X}|\leq \exp\inp{nH(X)},\label{eq:type_property2}\\
(n+1)^{-|\cX||\cY|}\exp\inp{nH(Y|X)}&\leq |T^n_{Y|X}(\vecx)|\leq \exp\inp{nH(Y|X)},\label{eq:type_property3}\\
(n+1)^{-|\cX|}\exp\inb{-nD(P_X||Q)}&\leq \sum_{\vecx'\in T^n_{X}}Q^n(\vecx')\leq \exp\inb{-nD(P_X||Q)}\text{ and}\label{eq:type_property4}\\
\sum_{\vecy\in T^n_{Y|X}(\vecx)}U_{Y|X}^n(\vecy|\vecx)&\leq \exp\inb{-nD(P_{XY}||P_XU_{Y|X})}.\label{eq:type_property5}
\end{align}
We consider the case when user \two is malicious.
We will analyse $P_{e,\maltwo}$. Suppose a malicious user \two sends $\vecy$. Let $P_{e,\maltwo}(\vecy)$ denote the probability of error when user \two is malicious and sends $\vecy$. That is, for $\cE^{\one}_{\mo} \defineqq \inb{\vecz:\phi_{\one}(\vecz)\notin\{\mo,\twob\}}$,  
\begin{align*}
P_{e,\maltwo}(\vecy) = \frac{1}{N_{\one}}\sum_{m_{\one}\in \mathcal{M}_{\one}}W^n\inp{\cE^{\one}_{\mo}|f_{\one}(\mo), \vecy}
\end{align*}  and $$P_{e,\maltwo} = \max_{\vecy}P_{e,\maltwo}(\vecy).$$ We will show that $P_{e,\maltwo}(\vecy)$ is small for each $\vecy\in \cY^n$. The analysis follows the flowchart given in Figure~\ref{fig:flowchart1_authcom}.}

\olive{\noindent We define the following sets.
\begin{align*}
\cH_1&\defineqq \inb{m_{\one}: (\vecx_{m_{\one}}, \vecy)\in T_{XY}^n \text{ such that } I(X;Y)> \epsilon}, \text{ and}\\
\cH_2&\defineqq \inb{m_{\one}: (\vecx_{m_{\one}}, \vecy)\in T_{XY}^n \text{ such that }  I(X;Y)\leq \epsilon}
\end{align*}
\begin{align}
P_{e,\maltwo}(\vecy) &\leq \frac{1}{N_{\one}}|\cH_1| + \frac{1}{N_{\one}}\sum_{m_{\one}\in \cH_2}\inp{\sum_{P_{XYZ}\in \cD_{\eta}^c}\sum_{\vecz\in T^n_{Z|XY}(\vecx_{m_{\one}},\vecy)}W^n(\vecz|\vecx_{m_{\one}}, \vecy)} \nonumber\\
&\qquad\qquad + \frac{1}{N_{\one}}\sum_{m_{\one}\in \cH_2}\inp{\sum_{P_{XYZ}\in \cD_{\eta}}\sum_{\substack{\vecz\in T^n_{Z|XY}(\vecx_{m_{\one}},\vecy)\\: \phi_{\one}(\vecz)\notin\inb{m_{\one}, \twob}}}W^n(\vecz|\vecx_{m_{\one}}, \vecy)}\nonumber\\
&=: P_1(\vecy) +P_2(\vecy)+P_3(\vecy).\label{eq:error}
\end{align}}

\begin{figure}[!h]
\begin{centering}
\begin{tikzpicture}[scale=1.2]
 
  \node (A1) at (-1.5, 0) [rectangle, draw] {$P_{e,\maltwo}(\vecy)$};
  \node (C) at (-1.5, -2) [rounded corners=3pt, draw] {Union bound};
  \node[xshift = 0.8cm, right of = C] {eq.~\eqref{eq:error}};
  \node (D) at (1.5, -2) {};

  \node (E) at (-3, -4)  {\small $\substack{\text{\small  small}\\\text{ (atypical event)}}$};
\node (F) at (-1.5, -4.5)  {$\substack{ \text{\small small}\\\text{ (atypical event)}}$};
 \node (G) at (0, -4) [rounded corners=3pt, draw] {Union bound};
  \node[xshift = 0.8cm, right of = G] {eq.~\eqref{err:upperbound_p}};
  \node (G1) at (-1.5, -6) [rounded corners=3pt, draw] {Union bound}; 
  \node (G3) at (1.5, -6) [rounded corners=3pt, draw] {Union bound};
 
  \node (G11) at (-2.7, -8)  {$\substack{\text{\small  small } \\ \text{(by codebook}\\\text{ property\eqref{codebook:2b})}}$};  
  \node (G12) at (-0.6, -7.5) {$\substack{\text{\small small }\\ \text{(as $\eta>3\epsilon + 4\delta$}\\\text{{\em i.e.,}$R_{\one}=R_{\two}=\delta$ is small }\\\text{ enough such that }\\\eta>3\epsilon+4\delta \text{  and } \delta>\epsilon.)}$}; 

 \node (G31) at (2.3, -8)  {$\substack{\text{\small small }\\ \text{(as $\eta>3\epsilon + 4\delta$}\\\text{{\em i.e.,}$R_{\one}=R_{\two}=\delta$ is small }\\\text{ enough such that }\\\eta>3\epsilon+4\delta \text{  and } \delta>\epsilon.)}$}; 
  \node (G32) at (0.3, -8.7) {$\substack{\text{\small  small } \\ \text{(by codebook}\\\text{ property\eqref{codebook:4})}}$};  

 	\draw[->] (G) -- (G1) node[midway, right]{$\scriptstyle P_{\cE_{\mo,1}}(\vecy)$};
 	\draw[->] (G) -- (G3) node[midway, right]{$\scriptstyle P_{\cE_{\mo,2}}(\vecy)$};
 	\draw[->] (G1) -- (G11) node[midway, left]{$\small \stackrel{\text{\eqref{eq:1_p} does }}{\text{\scriptsize not hold}}$};
 	\draw[->] (G1) -- (G12) node[midway, xshift = 0cm, yshift = 0.2cm,  right]{$\scriptstyle \eqref{eq:1_p} \text{ holds}$};

 	 	\draw[->] (G3) -- (G32) node[xshift = 0.3cm, yshift = 0.6cm,midway, left]{$\small \stackrel{\text{\eqref{eq:2_p} does }}{\text{\scriptsize not hold}}$};
 	\draw[->] (G3) -- (G31) node[midway, xshift = 0.2cm, yshift = -0.2cm,  right]{$\scriptstyle \eqref{eq:2_p} \text{ holds}$};

 	\draw[->] (C) -- (E) node[midway, right]{$\scriptstyle P_1(\vecy)$};
 	\draw[->] (C) -- (F) node[yshift = -0.2cm, midway, right]{$\scriptstyle P_2(\vecy)$};
 	\draw[->] (C) -- (G) node[midway, right]{$\scriptstyle \scriptstyle P_3(\vecy)$};
  \draw[->] (A1) -- (C) node[midway, right]{};
\node at (7, -3.5)  {
\begin{tabular}{p{1.5cm}|p{6.5cm}} 
$P_{e,\maltwo}(\vecy)$& the average probability of error when malicious user \two sends $\vecy$\\
\hline
$P_1(\vecy)$ & the average probability that channel inputs are atypical\\
\hline	
$P_2(\vecy)$ & the average probability that the channel output is atypical \\
\hline	
$P_3(\vecy)$ & the average probability of error when channel inputs and output are typical\\
\hline	
$P_{\cE_{\mo,1}}(\vecy)$& condition \ref{check:2} in Definition~\ref{D_eta_z} does not hold\\
\hline	
$P_{\cE_{\mo,2}}(\vecy)$& condition ~\ref{check:3}  in Definition~\ref{D_eta_z} does not hold\\
\end{tabular}
};
\end{tikzpicture}
\caption{Flowchart depicting the flow of analysis of $P_{e,\maltwo}(\vecy)$, the average probability of error when user \two is malicious and sends $\vecy$.} \label{fig:flowchart1_authcom}
\end{centering}
\end{figure}

\olive{The first term $
P_1(\vecy) = \frac{1}{N_{\one}}|\cH_1|
$ is upper bounded by
\begin{align*}
\blue{\frac{|\cP^n_{\cX\times \cY}|}{N_{\one}}}\times|\inb{m_{\one}: (\vecx_{m_{\one}}, \vecy)\in T_{XY}^n, \, I(X;Y)> \epsilon}|
\end{align*}
which goes to zero as $n\rightarrow \infty$ by~\eqref{codebook:1} and noting that there are only polynomially many types. 
We now analyse the second term
\begin{align*}
P_2(\vecy)=\frac{1}{N_{\one}}\sum_{m_{\one}\in \cH_2}\inp{\sum_{P_{XYZ}\in \cD_{\eta}^c}\sum_{\vecz\in T^n_{Z|XY}(\vecx_{m_{\one}},\vecy)}W^n(\vecz|\vecx_{m_{\one}}, \vecy)}.
\end{align*}
For any $m_{\one}\in \cH_2$,
\begin{align*}
\sum_{P_{XYZ}\in \cD_{\eta}^c}\sum_{\vecz\in T^n_{Z|XY}(\vecx_{m_{\one}},\vecy)}W^n(\vecz|\vecx_{m_{\one}}, \vecy) &\stackrel{\text{(a)}}{\leq} |\cD_{\eta}^c|\exp{\inp{-nD(P_{XYZ}||P_{XY}W)}}\\
& = |\cD_{\eta}^c|\exp{\inp{-n\inp{D(P_{XYZ}||P_{X}P_{Y}W) -I(X;Y)}}}\\
& \leq |\cD_{\eta}^c|\exp{\inp{-n\inp{\eta-\epsilon}}} \rightarrow 0 \text{ as }\eta>\epsilon\text{ and }|\cD_{\eta}^c|\leq (n+1)^{|\cX||\cY||\cZ|},
\end{align*}
where (a) follows from \eqref{eq:type_property5}. We are left to analyse the last term
\begin{align}
P_3(\vecy)= \frac{1}{N_{\one}}\sum_{m_{\one}\in \cH_2}\inp{\sum_{P_{XYZ}\in \cD_{\eta}}\sum_{\substack{\vecz\in T^n_{Z|XY}(\vecx_{m_{\one}},\vecy)\\: \phi_{\one}(\vecz)\notin\inb{m_{\one}, \twob}}}W^n(\vecz|\vecx_{m_{\one}}, \vecy)}\label{eq:P_3_vecy}.
\end{align}}

\olive{\noindent Recall that because of Lemma~\ref{lemma:disambiguity}, whenever $|D_{\one}(\eta, \vecz)|, |D_{\two}(\eta, \vecz)|>0$, we have $|D_{\one}(\eta, \vecz)| = |D_{\two}(\eta, \vecz)| = 1$. This implies that for $(\vecx_{\mo}, \vecy, \vecz)\in P_{XYZ}$ such that $P_{XYZ}\in \cD_{\eta}$ and $\mo\in \cH_2$, the output of $\phi_{\one}(\vecz)$ is not in the set $\inb{m_{\one}, \twob}$ only if one of the following happens:
\begin{itemize}
	\item $|D_{\one}(\eta, \vecz)| = |D_{\two}(\eta, \vecz)| = 1$, but $\mo\notin D_{\one}(\eta, \vecz)$.
	\item $|D_{\one}(\eta, \vecz)| = 0$.
\end{itemize} 
To formalize this, we define the following sets for $\mo\in \cM_{\one}$.
\begin{align*}
\cG_{\mo} &= \inb{\vecz: (\vecx_{\mo}, \vecy, \vecz)\in P_{XYZ}, P_{XYZ}\in \cD_{\eta}, I(X; Y)\leq \epsilon},\\
\cG_{\mo,0} &= \cG_{\mo} \cap \inb{\vecz: \phi_{\one}(\vecz)\notin\inb{m_{\one}, \twob} },\\
\cG_{\mo,1} &= \cG_{\mo} \cap\inb{\vecz:  |D_{\one}(\eta, \vecz)| = |D_{\two}(\eta, \vecz)| = 1 , \mo\notin D_{\one}(\eta, \vecz) },\\
\cG_{\mo,2} &= \cG_{\mo} \cap\inb{\vecz:  |D_{\one}(\eta, \vecz)| = 0},\text{ and }\\
\cG_{\mo,3} &= \cG_{\mo} \cap\inb{\vecz: \mo\notin D_{\one}(\eta, \vecz) }.
\end{align*}
We are interested in $\cG_{\mo,0}$. Note that $\cG_{\mo,0} \subseteq \cG_{\mo,1}\cup \cG_{\mo,2}\subseteq \cG_{\mo,3}$. So, it suffices to upper bound the probability of $\cG_{\mo,3}$ when $\vecx_{\mo}$ is sent by user \one and $\vecy$ by user \two. From the definition of $D_{\one}(\eta, \vecz)$, we see that $\cG_{\mo,3}$ is the set of $\vecz \in \cZ^n$ which satisfy decoding condition~\ref{check:1} (this is because $\vecz\in\cG_{\mo,3}$ implies $\vecz\in \cG_{\mo}$ since $\cG_{\mo,3}\subseteq \cG_{\mo}$) but do not satisfy either decoding condition~\ref{check:2} or decoding condition~\ref{check:3}. We capture this by defining the following sets of distributions:
\begin{align}
\cP_1& =  \{P_{X\tilde{X}Y\tilde{Y}Z}\in \cP^n_{\cX\times\cX\times\cY\times\cY\times\cZ}: P_{XYZ}\in \cD_{\eta},I(X;Y)\leq \epsilon,\, P_{\tilde{X}Y'Z}\in \cD_{\eta} \text{ for some }P_{Y'|\tilde{X}Z},\nonumber\\
&\qquad  \, P_{X'\tilde{Y}Z}\in \cD_{\eta} \text{ for some }P_{X'|\tilde{Y}Z}, P_{X}=P_{\tilde{X}}=P_{\one}, P_{\tilde{Y}} = P_{\two}\text{ and }I(\tilde{X}\tilde{Y};XZ|Y)\geq\eta\}\label{eq:achiev_P1}\\
\cP_2& = \{P_{X\tilde{Y}_1\tilde{Y}_2YZ}\in \cP^n_{\cX\times\cY\times\cY\times\cY\times\cZ}: P_{XYZ}\in \cD_{\eta},I(X;Y)\leq \epsilon,\, P_{X'_1\tilde{Y}_1Z}\in \cD_{\eta} \text{ for some }P_{X'_1|\tilde{Y}_1Z},\nonumber\\
&\qquad \, P_{X'_2\tilde{Y}_2Z}\in \cD_{\eta} \text{ for some }P_{X'_2|\tilde{Y}_2Z}, P_{X}=P_{\one}, P_{\tilde{Y}_1}=P_{\tilde{Y}_2} = P_{\two}\text{ and }I(\tilde{Y}_1\tilde{Y}_2;XZ|Y)\geq\eta\}\label{eq:achiev_P2}.
\end{align}
For $P_{X\tilde{X}Y\tilde{Y}Z}\in \cP_1$ and $P_{X\tilde{Y}_1\tilde{Y}_2YZ}\in \cP_2 $, let
\begin{align*}
\cE_{\mo,1}(P_{X\tilde{X}Y\tilde{Y}Z}) & = \big\{\vecz: \exists(\tilde{m}_{\one},\, \tilde{m}_{\two})\in \cM_{\one}\times \cM_{\two}, \, \tilde{m}_{\one}\neq m_{\one}, \,  \inp{\vecx_{\mo},\vecx_{\tilde{m}_{\one}}, \vecy,   \vecy_{\tilde{m}_{\two}}, \vecz}\in T^n_{X\tilde{X}Y\tilde{Y}Z} \big\} \text{ and }\\
\cE_{\mo,2}(P_{X\tilde{Y}_1\tilde{Y}_2YZ}) & = \big\{\vecz: \exists \tilde{m}_{\two 1},\, \tilde{m}_{\two 2}\in \cM_{\two}, \tilde{m}_{\two 1}\neq \tilde{m}_{\two 2},\,\inp{\vecx_{\mo},  \vecy_{\tilde{m}_{\two 1}}, \vecy_{\tilde{m}_{\two 2}},\vecy,\vecz}\in T^n_{X\tilde{Y}_1\tilde{Y}_2YZ}\big\}.
\end{align*}
Note that $\cG_{\mo,3}= \inp{\cup_{P_{X\tilde{X}Y\tilde{Y}Z}\in \cP_1}\cE_{\mo,1}(P_{X\tilde{X}Y\tilde{Y}Z})}\cup \inp{\cup_{P_{X\tilde{Y}_1\tilde{Y}_2YZ}\in \cP_2}\cE_{\mo,2}(P_{X\tilde{Y}_1\tilde{Y}_2YZ})}$.}

\olive{Thus, \eqref{eq:P_3_vecy} can be analyzed as below.
\begin{align}
P_3(\vecy)&=\frac{1}{N_{\one}}\sum_{m_{\one}\in \cH_2}\inp{\sum_{P_{XYZ}\in \cD_{\eta}}\sum_{\substack{\vecz\in T^n_{Z|XY}(\vecx_{m_{\one}},\vecy):\\ \phi_{\one}(\vecz)\notin\inb{m_{\one}, \twob}}}W^n(\vecz|\vecx_{m_{\one}}, \vecy)}\nonumber\\
&\leq\frac{1}{N_{\one}}\sum_{\mo\in \cH_2}\sum_{P_{X\tilde{X}Y\tilde{Y}Z}\in \cP_1}  W^n\inp{\cE_{\mo,1}(P_{X\tilde{X}Y\tilde{Y}Z})|\vecx_{\mo}, \vecy} \nonumber\\
&\qquad \qquad+ \frac{1}{N_{\one}}\sum_{\mo\in \cH_2}\sum_{P_{X\tilde{Y}_1\tilde{Y}_2YZ}\in \cP_2}  W^n\inp{\cE_{\mo,2}(P_{X\tilde{Y}_1\tilde{Y}_2YZ})|\vecx_{\mo}, \vecy}\\
&=:P_{\cE_{\mo,1}}(\vecy) + P_{\cE_{\mo,2}}(\vecy). \label{err:upperbound_p}
\end{align}
We see that $|\cP_1|$ and $|\cP_2|$ increase at most polynomially in $n$ (see \eqref{eq:type_property1}) and clearly $|\cH_2|\leq N_{\one}$. So, it will  suffice to uniformly upper bound $W^n\inp{\cE_{\mo,1}(P_{X\tilde{X}Y\tilde{Y}Z})|\vecx_{\mo}, \vecy}$ and $W^n\inp{\cE_{\mo,2}(P_{X\tilde{Y}_1\tilde{Y}_2YZ})|\vecx_{\mo}, \vecy}$ by a term exponentially decreasing in $n$ for all $P_{X\tilde{X}Y\tilde{Y}Z}\in \cP_1$ and $P_{X\tilde{Y}_1\tilde{Y}_2YZ}\in \cP_2 $ respectively. 
We start the analysis of $P_{\cE_{\mo,1}}(\vecy)$ by upper bounding $W^n\inp{\cE_{\mo,1}(P_{X\tilde{X}Y\tilde{Y}Z})|\vecx_{\mo}, \vecy}$. By using~\eqref{codebook:2b}, we see that for $P_{X\tilde{X}Y\tilde{Y}Z}\in \cP_1$ such that
\begin{align}
I\inp{X;\tilde{X}\tilde{Y}Y}> |R_{\one}- I(\tilde{X};\tilde{Y}Y)|^{+}+|R_{\two}-I(\tilde{Y};Y)|^{+}+\epsilon,\,\label{eq:1_p_complement}
\end{align}we have,
\begin{align*}
\blue{\frac{1}{N_{\one}}}{\left|\inb{\mo:(\vecx_{\mo}, \vecx_{\tilde{m}_{\one}}, \vecy_{\mt}, \vecy)\in T^n_{X\tilde{X}\tilde{Y}Y} \text{ for some }\tilde{m}_{\one}\neq \mo\text{ and some }\mt} \right|} \leq \exp\inb{-n\epsilon/2}.
\end{align*}
So, for all such $P_{X\tilde{X}Y\tilde{Y}Z}$,
\begin{align*}
&\frac{1}{N_{\one}}\sum_{\mo\in \cH_2} W^n\inp{\cE_{\mo,1}(P_{X\tilde{X}Y\tilde{Y}Z})|\vecx_{\mo}, \vecy} \\
& = \frac{1}{N_{\one}}\sum_{\substack{\mo:(\vecx_{\mo}, \vecx_{\tilde{m}_{\one}}, \vecy_{\mt}, \vecy)\in T^n_{X\tilde{X}\tilde{Y}Y},\\ \tilde{m}_{\one}\in \cM_{\one},\tilde{m}_{\one}\neq \mo,\mt\in\cM_{\two}}} \,\,\sum_{\vecz\in T^{n}_{Z|X\tilde{X}Y\tilde{Y}}(\vecx_{\mo},\vecx_{\tilde{m}_{\one}},\vecy,\vecy_{\tilde{m}_{\two}})}W^n\inp{\vecz|\vecx_{\mo}, \vecy}\\
&\leq \exp\inb{-n\epsilon/2}.
\end{align*}
Thus, it is sufficient to consider distributions $P_{X\tilde{X}Y\tilde{Y}Z}\in \cP_1$ for which 
\begin{align}
I\inp{X;\tilde{X}\tilde{Y}Y}\leq |R_{\one}- I(\tilde{X};\tilde{Y}Y)|^{+}+|R_{\two}-I(\tilde{Y};Y)|^{+}+\epsilon.\label{eq:1_p}
\end{align}
For $P_{X\tilde{X}Y\tilde{Y}Z}\in \cP_1$ satisfying~\eqref{eq:1_p},
\begin{align}
&\sum_{\vecz\in \cE_{\mo,1}(P_{X\tilde{X}Y\tilde{Y}Z})}W^n(\vecz|\vecx_{\mo}, \vecy)\nonumber\\
&\qquad=\sum_{\substack{\tilde{m}_{\one}, \tilde{m}_{\two}: \tilde{m}_{\one}\neq \mo, \\(\vecx_{\mo}, \vecx_{\tilde{m}_{\one}}, \vecy_{\tilde{m}_{\two}},\vecy)\in T^{n}_{X\tilde{X}\tilde{Y}Y}}}\sum_{\vecz\in T^n_{Z|X\tilde{X}\tilde{Y}Y}(\vecx_{\mo}, \vecx_{\tilde{m}_{\one}}, \vecy_{\tilde{m}_{\two}},\vecy)}W^n(\vecz|\vecx_{\mo}, \vecy)\nonumber\\
&\qquad\stackrel{\text{(a)}}{\leq} \sum_{\substack{\tilde{m}_{\one}, \tilde{m}_{\two}: \tilde{m}_{\one}\neq \mo,\\(\vecx_{\mo}, \vecx_{\tilde{m}_{\one}}, \vecy_{\tilde{m}_{\two}},\vecy)\in T^{n}_{X\tilde{X}\tilde{Y}Y}}}\frac{|T^{n}_{Z|X\tilde{X}\tilde{Y}Y}(\vecx_{\mo},\vecx_{\tilde{m}_{\one}}, \vecy_{\tilde{m}_{\two}}, \vecy)|}{|T^n_{Z|XY}(\vecx_{\mo},\vecy)|}\nonumber\\
&\qquad \stackrel{\text{(b)}}{\leq} \sum_{\substack{\tilde{m}_{\one}, \tilde{m}_{\two}: \tilde{m}_{\one}\neq \mo,\\(\vecx_{\mo}, \vecx_{\tilde{m}_{\one}}, \vecy_{\tilde{m}_{\two}},\vecy)\in T^{n}_{X\tilde{X}\tilde{Y}Y}}}\frac{\exp\inp{nH(Z|X\tilde{X}\tilde{Y}Y)}}{(n+1)^{-|\cX||\cY||\cZ|}\exp\inp{nH(Z|XY)}}\nonumber\\
&\qquad \stackrel{\text{(c)}}\leq \sum_{\substack{\tilde{m}_{\one}, \tilde{m}_{\two}:\\(\vecx_{\mo}, \vecx_{\tilde{m}_{\one}}, \vecy_{\tilde{m}_{\two}},\vecy)\in T^{n}_{X\tilde{X}\tilde{Y}Y}}}\exp\inp{-n\inp{I(Z;\tilde{X}\tilde{Y}|XY)-\epsilon}}\nonumber\\
&\qquad\stackrel{\text{(d)}}{\leq}\exp\inp{n\inp{|R_{\one}- I(\tilde{X};\tilde{Y}XY)|^{+}+|R_{\two}-I(\tilde{Y};XY)|^{+}-I(Z;\tilde{X}\tilde{Y}|XY)+2\epsilon}},\label{eq:upperbound2_p}
\end{align}
where (a) follows by noting that whenever $\vecz$ belongs to $T^n_{Z|X\tilde{X}\tilde{Y}Y}(\vecx_{\mo}, \vecx_{\tilde{m}_{\one}}, \vecy_{\tilde{m}_{\two}},\vecy)$, $\vecz$ also belongs to $T^n_{Z|XY}(\vecx_{\mo},\vecy)$;	 and for each $\vecz$ in $T^n_{Z|XY}(\vecx_{\mo},\vecy)$,  the value of $W^n(\vecz|\vecx_{\mo}, \vecy)$ is the same and hence is upper bounded by $1/|T^n_{Z|XY}(\vecx_{\mo}, \vecy)|$.  (b) follows from \eqref{eq:type_property3}, (c) holds for sufficiently large $n$ and (d) follows from~\eqref{codebook:3b}. 
We see that
\begin{align*}
I(Z;\tilde{X}\tilde{Y}|XY) &= I(XZ;\tilde{X}\tilde{Y}|Y)-I(X;\tilde{X}\tilde{Y}|Y)\\
&\stackrel{\text{(a)}}{\geq} \eta-I(X;\tilde{X}\tilde{Y}|Y)\\
&\stackrel{\text{(b)}}{\geq} \eta-I(X;\tilde{X}\tilde{Y}Y)\\
&\stackrel{\text{(c)}}{\geq} \eta- |R_{\one}- I(\tilde{X};\tilde{Y}Y)|^{+}-|R_{\two}-I(\tilde{Y};Y)|^{+}-\epsilon
\end{align*} 
where (a) follows from the fact that $I(XZ;\tilde{X}\tilde{Y}|Y)\geq \eta$ for $P_{X\tilde{X}Y\tilde{Y}Z}\in\cP_1$ (see \eqref{eq:achiev_P1}), (b) from $I(X;\tilde{X}\tilde{Y}|Y)\leq I(X;\tilde{X}\tilde{Y}Y)$ and (c) follows from \eqref{eq:1_p}.
Applying this to \eqref{eq:upperbound2_p},
\begin{align}
&\sum_{\vecz\in \cE_{\mo,1}(P_{X\tilde{X}Y\tilde{Y}Z})}W^n(\vecz|\vecx_{\mo}, \vecy)\nonumber\\
&\leq \exp\inp{n\inp{|R_{\one}- I(\tilde{X};\tilde{Y}XY)|^{+}+|R_{\two}-I(\tilde{Y};XY)|^{+}+|R_{\one}- I(\tilde{X};\tilde{Y}Y)|^{+}+|R_{\two}-I(\tilde{Y};Y)|^{+}-\eta+3\epsilon}}\nonumber\\
&\stackrel{\text{(a)}}{\leq} \exp\inp{n\inp{4\delta-\eta+3\epsilon}}\label{eq:upperbound123}\\
&{\rightarrow} 0 \text{ when }{\eta>3\epsilon+4\delta \quad(\text{see}~\eqref{eq:eta_cond}).}\nonumber
\end{align}
Here, (a) follows by recalling that $R_{\one} = R_{\two} = \delta$ and noting that $|R_{\one}- I(\tilde{X};\tilde{Y}XY)|^{+}+|R_{\two}-I(\tilde{Y};XY)|^{+}+|R_{\one}- I(\tilde{X};\tilde{Y}Y)|^{+}+|R_{\two}-I(\tilde{Y};Y)|^{+}\leq 2R_{\one}+2R_{\two} = 4\delta$.
Now, we move on to the second term $P_{\cE_{\mo,2}}(\vecy)$ in~\eqref{err:upperbound_p}. Proceeding in a similar fashion, we see that by using~\eqref{codebook:4}, it is sufficient to consider distributions $P_{X\tilde{Y}_1\tilde{Y}_2YZ}\in \cP_2$ for which  
\begin{align}
I\inp{X;\tilde{Y}_1\tilde{Y}_2Y}\leq|R_{\two}-I(\tilde{Y}_1;Y)|^{+}+|R_{\two}-I(\tilde{Y}_2;\tilde{Y}_1 Y)|^{+} +\epsilon.\label{eq:2_p}
\end{align}
For $P_{X\tilde{Y}_1\tilde{Y}_2YZ}\in \cP_2$ satisfying~\eqref{eq:2_p}, along similar lines as the steps which led to \eqref{eq:upperbound2_p} and \eqref{eq:upperbound123},
\begin{align}
&\sum_{\vecz\in \cE_{\mo, 2}(P_{X\tilde{Y}_1\tilde{Y}_2YZ})}W^n(\vecz|\vecx_{\mo}, \vecy)\nonumber\\
&\qquad=\sum_{\substack{\tilde{m}_{\two 1},  \tilde{m}_{\two 2}:\tilde{m}_{\two 1}\neq \tilde{m}_{\two 2},\\(\vecx_{\mo}, \vecy_{\tilde{m}_{\two 1}}, \vecy_{\tilde{m}_{\two 2}},\vecy)\in T^{n}_{X\tilde{Y}_1\tilde{Y}_2Y}}}\sum_{\vecz:(\vecx_{\mo}, \vecy_{\tilde{m}_{\two 1}}, \vecy_{\tilde{m}_{\two 2}},\vecy, \vecz)\in T^{n}_{X\tilde{Y}_1\tilde{Y}_2YZ}}W^n(\vecz|\vecx_{\mo}, \vecy)\nonumber\\
&\qquad\leq \sum_{\substack{\tilde{m}_{\two 1}, \tilde{m}_{\two 2}:\tilde{m}_{\two 1}\neq \tilde{m}_{\two 2},\\(\vecx_{\mo}, \vecy_{\tilde{m}_{\two 1}}, \vecy_{\tilde{m}_{\two 2}},\vecy)\in T^{n}_{X\tilde{Y}_1\tilde{Y}_2Y}}}\frac{|T^{n}_{Z|X\tilde{Y}_1\tilde{Y}_2Y}(\vecx_{\mo},\vecy_{\tilde{m}_{\two 1}}, \vecy_{\tilde{m}_{\two 2}}, \vecy)|}{|T^n_{Z|XY}(\vecx_{\mo},\vecy)|}\nonumber\\
&\qquad \leq \sum_{\substack{\tilde{m}_{\two 1}, \tilde{m}_{\two 2}:\tilde{m}_{\two 1}\neq \tilde{m}_{\two 2},\\(\vecx_{\mo}, \vecy_{\tilde{m}_{\two 1}}, \vecy_{\tilde{m}_{\two 2}},\vecy)\in T^{n}_{X\tilde{Y}_1\tilde{Y}_2Y}}}\frac{\exp\inp{nH(Z|X\tilde{Y}_1\tilde{Y}_2Y)}}{(n+1)^{-|\cX||\cY||\cZ|}\exp\inp{nH(Z|XY)}}\nonumber\\
&\qquad \stackrel{\text{(a)}}{\leq} \sum_{\substack{\tilde{m}_{\two 1}, \tilde{m}_{\two 2}:\tilde{m}_{\two 1}\neq \tilde{m}_{\two 2},\\(\vecx_{\mo}, \vecy_{\tilde{m}_{\two 1}}, \vecy_{\tilde{m}_{\two 2}},\vecy)\in T^{n}_{X\tilde{Y}_1\tilde{Y}_2Y}}}\exp\inp{-n\inp{I(Z;\tilde{Y}_1\tilde{Y}_2|XY)-\epsilon}} \nonumber\\
&\qquad\stackrel{\text{(b)}}{\leq}\exp\inp{n\inp{|R_{\two}-I(\tilde{Y}_1;XY)|^{+}+|R_{\two}-I(\tilde{Y}_2;\tilde{Y}_1 XY)|^{+}-I(Z;\tilde{Y}_1\tilde{Y}_2|XY)+2\epsilon}}\label{eq:upperbound3_p}\\
&\qquad = \exp\inp{n\inp{|R_{\two}-I(\tilde{Y}_1;XY)|^{+}+|R_{\two}-I(\tilde{Y}_2;\tilde{Y}_1 XY)|^{+}+I(X;\tilde{Y}_1\tilde{Y}_2|Y)-I(XZ;\tilde{Y}_1\tilde{Y}_2|Y)+2\epsilon}}\nonumber\\
&\qquad\stackrel{\text{(c)}}{\leq}\exp\Big(n\Big(|R_{\two}-I(\tilde{Y}_1;XY)|^{+}+|R_{\two}-I(\tilde{Y}_2;\tilde{Y}_1 XY)|^{+}+|R_{\two}-I(\tilde{Y}_1;Y)|^{+}\nonumber\\
&\qquad\qquad+|R_{\two}-I(\tilde{Y}_2;\tilde{Y}_1 Y)|^{+}-\eta+3\epsilon\Big)\Big)\nonumber\\
&\qquad\stackrel{\text{(d)}}{\leq} \exp\inp{n\inp{4\delta-\eta+3\epsilon}}\label{eq:upperbound123_p}\\
&\qquad\rightarrow 0\text{ when }{\eta>3\epsilon+4\delta \quad(\text{see}~\eqref{eq:eta_cond}).\nonumber}
\end{align}
where (a) holds for large $n$, (b) follows from~\eqref{codebook:5}, (c) from \eqref{eq:2_p} (note that $I(X;\tilde{Y}_1\tilde{Y}_2|Y)\leq I(X;\tilde{Y}_1\tilde{Y}_2Y)$) and the fact that $I(XZ;\tilde{Y}_1\tilde{Y}_2|Y)\geq \eta$ since $P_{X\tilde{Y}_1\tilde{Y}_2YZ}\in \cP_{2}$ (see \eqref{eq:achiev_P2}) and (d) by recalling that $R_{\two} = \delta$ and hence $|R_{\two}-I(\tilde{Y}_1;XY)|^{+}+|R_{\two}-I(\tilde{Y}_2;\tilde{Y}_1 XY)|^{+}+|R_{\two}-I(\tilde{Y}_1;Y)|^{+}+|R_{\two}-I(\tilde{Y}_2;\tilde{Y}_1 Y)|^{+}\leq 4R_{\two} = 4\delta$.}

\olive{Thus, $P_{e,\maltwo}\rightarrow 0$ as $n\rightarrow \infty$ for the sufficiently small $\delta>\epsilon>0$ so that $\eta>3\epsilon+ 4\delta$ is sufficiently small for Lemma~\ref{lemma:disambiguity} to hold. 
Similarly, because of symmetry, we can show that if $\eta>3\epsilon+4\delta$ the probability of error goes to zero as $n$ goes to infinity when user \one is malicious.}

\end{proof}


\section{Codebook for Theorems~\ref{thm:main_result} and~\ref{thm:inner_bd}} \label{sec:codebooklemma}
We first restate Lemma~\ref{lemma:codebook} by including the analogous statements when roles of users \one and \two are interchanged.
\begin{Lemma}
    Suppose $\mathcal{X,Y,Z}$ are finite. Let $P_{\one}\in \cP^n_{\cX}$ and $P_{\two}\in \cP^n_{\cY}$. For any $\epsilon>0$, there exists $n_0(\epsilon)$ such that for all $n\geq n_0(\epsilon),\, N_{\one}, N_{\two}\geq \exp(n\epsilon)$, there are codewords $\vecx_1, \vecx_2, \ldots, \vecx_{N_{\one}}$ of type $P_{\one}$ and $\vecy_1, \vecy_2, \ldots, \vecy_{N_{\two}}$ of type $P_{\two}$ such that for all $\vecx, \vecx', \tilde{\vecx}, \bar{\vecx}\in \cX^n$  and $\vecy, \vecy', \tilde{\vecy}, \bar{\vecy}\in \cY^n$, and joint types $P_{X\tilde{X}\tilde{Y}Y}\in \cP^n_{\cX\times \cX\times \cY\times \cY}$, $P_{X'\tilde{Y}_1\tilde{Y}_2Y'}\in \cP^n_{\cX\times \cY\times \cY\times \cY}$, $P_{\hat{X}\bar{X}\bar{Y}\hat{Y}}\in \cP^n_{\cX\times \cX\times \cY\times \cY}$ and $P_{X''\hat{X}_1\hat{X}_2Y''}\in \cP^n_{\cX\times \cX\times \cX\times \cY}$ 
     and for $R_{\one}\defineqq (1/n)\log N_{\one}$ and $R_{\two}\defineqq (1/n)\log N_{\two}$, 
     the following holds:\footnote{Note that $\exp$ and $\log$ are with respect to base 2.} 
\begin{align}
&\blue{\frac{1}{N_{\one}}}{\left|\inb{m_{\one}: (\vecx_{\mo}, \vecy)\in T_{XY}^n}\right|}\leq \exp\inb{-n\epsilon/2},\text{ if }I(X;Y)> \epsilon;\hspace{7cm}&\text{\eqref{codebook:1}}\\
&\blue{\frac{1}{N_{\one}}}\left|\inb{\mo:(\vecx_{\mo}, \vecx_{\tilde{m}_{\one}}, \vecy_{\mt}, \vecy)\in T^n_{X\tilde{X}\tilde{Y}Y} \text{ for some }\tilde{m}_{\one}\neq \mo\text{ and some }\mt} \right| \leq \exp\inb{-n\epsilon/2}, \nonumber\\
& \qquad \text{ if }I(X;\tilde{X}\tilde{Y}Y)-|R_{\one}- I(\tilde{X};\tilde{Y}Y)|^{+}-|R_{\two}-I(\tilde{Y};Y)|^{+} >\epsilon;&\text{\eqref{codebook:2b}}\\
&\left|\inb{(\tilde{m}_{\one}, \tilde{m}_{\two}):(\vecx, \vecx_{\tilde{m}_{\one}}, \vecy_{\tilde{m}_{\two}}, \vecy)\in T^n_{X\tilde{X}\tilde{Y}Y}}\right|\nonumber\\
&\qquad\leq \exp\inb{n\inp{|R_{\one}- I(\tilde{X};\tilde{Y}XY)|^{+}+|R_{\two}-I(\tilde{Y};XY)|^{+}+\epsilon}};&\text{\eqref{codebook:3b}}\\
&\blue{\frac{1}{N_{\one}}}\left|\inb{\mo:(\vecx_{\mo}, \vecy_{\tilde{m}_{\two 1}}, \vecy_{\tilde{m}_{\two 2}}, \vecy')\in T^n_{X'\tilde{Y}_1\tilde{Y}_2Y'} \text{ for some }\tilde{m}_{\two 1},\tilde{m}_{\two 2}, } \right| \leq \exp\inb{-n\epsilon/2}, \nonumber\\
& \qquad \text{ if }I(X';\tilde{Y}_1\tilde{Y}_2Y')-|R_{\two}-I(\tilde{Y}_1;Y')|^{+}-|R_{\two}-I(\tilde{Y}_2;\tilde{Y}_1 Y')|^{+} >\epsilon;&\text{\eqref{codebook:4}}\\
&\left|\inb{(\tilde{m}_{\two 1}, \tilde{m}_{\two 2}):(\vecx', \vecy_{\tilde{m}_{\two 1}}, \vecy_{\tilde{m}_{\two 2}}, \vecy')\in T^n_{X'\tilde{Y}_1\tilde{Y}_2Y'} }\right|\nonumber\\
&\qquad\leq \exp\inb{n\inp{|R_{\two}-I(\tilde{Y}_1;X'Y')|^{+}+|R_{\two}-I(\tilde{Y}_2;\tilde{Y}_1 X'Y')|^{+}+\epsilon}},\qquad\qquad\qquad\qquad\qquad\qquad \hspace{2.35cm}&\text{\eqref{codebook:5}}\\
&\blue{\frac{1}{N_{\one}}}\left|\inb{m_{\two}: (\bar{\vecx}, \vecy_{\mt})\in T_{\bar{X}\bar{Y}}^n}\right|\leq \exp\inb{-n\epsilon/2},\text{ if }I(\bar{X};\bar{Y})> \epsilon;\label{codebook:1q}\\
&\blue{\frac{1}{N_{\two}} }\left|\inb{\mt:(\vecx_{\mo}, \vecy_{\tilde{m}_{\two}}, \vecy_{\mt}, \bar{\vecx})\in T^n_{\hat{X}\hat{Y}\bar{Y}\bar{X}} \text{ for some }\tilde{m}_{\two}\neq \mt \text{ and some }\mo} \right|\leq \exp\inb{-n\epsilon/2}, \nonumber\\
& \qquad \text{ if }I(\bar{Y};\hat{Y}\hat{X}\bar{X})-|R_{\two}- I(\hat{Y};\hat{X}\bar{X})|^{+}-|R_{\one}-I(\hat{X};\bar{X})|^{+} >\epsilon;\label{codebook:2bq}\\
&\left|\inb{(\tilde{m}_{\one}, \tilde{m}_{\two}):(\bar{\vecy},  \vecy_{\tilde{m}_{\two}}, \vecx_{\tilde{m}_{\one}},\bar{\vecx})\in T^n_{\bar{Y}\hat{Y}\hat{X}\bar{X}}}\right|\nonumber\\
&\qquad\leq \exp\inb{n\inp{|R_{\two}- I(\hat{Y};\hat{X}\bar{X}\bar{Y})|^{+}+|R_{\one}-I(\hat{X};\bar{X}\bar{Y})|^{+}+\epsilon}};\label{codebook:3bq}\\
&\blue{\frac{1}{N_{\two}}}\left|\inb{\mt:(\vecy_{\mt}, \vecx_{\tilde{m}_{\one 1}}, \vecx_{\tilde{m}_{\one 2}}, \tilde{\vecx})\in T^n_{Y''\hat{X}_1\hat{X}_2X''} \text{ for some }\tilde{m}_{\one 1},\tilde{m}_{\one 2}, } \right| \leq \exp\inb{-n\epsilon/2}, \nonumber\\
& \qquad \text{ if }I(Y'';\hat{X}_1\hat{X}_2X'')-|R_{\one}-I(\hat{X}_1;X'')|^{+}-|R_{\one}-I(\hat{X}_2;\hat{X}_1 X'')|^{+} >\epsilon;\label{codebook:4q}\\
&\text{and }\left|\inb{(\tilde{m}_{\one 1}, \tilde{m}_{\one 2}):(\tilde{\vecy}, \vecx_{\tilde{m}_{\one 1}}, \vecx_{\tilde{m}_{\one 2}}, \tilde{\vecx})\in T^n_{Y''\hat{X}_1\hat{X}_2X''} }\right|\nonumber\\
&\qquad\leq \exp\inb{n\inp{|R_{\one}-I(\hat{X}_1;X''Y'')|^{+}+|R_{\one}-I(\hat{X}_2;\hat{X}_1 X''Y'')|^{+}+\epsilon}}.\label{codebook:5q}
\end{align}
Here, statements \eqref{codebook:1q}-\eqref{codebook:5q} are analogous to \eqref{codebook:1}-\eqref{codebook:5} with the roles of users \one and \two are interchanged. In particular, $(\vecx, \vecx', \vecy, \vecy')$ is replaced with $(\bar{\vecy}, \tilde{\vecy}, \bar{\vecx}, \tilde{\vecx})$, $(X, \tilde{X}, \tilde{Y}, Y)$ with $(\bar{Y}, \hat{Y}, \hat{X}, \bar{X})$, $(X', \tilde{Y}_1, \tilde{Y}_2, Y')$ with $(Y'', \hat{X}_1, \hat{X}_2, X'')$ and $R_{\one}$ with $R_{\two}$.
\end{Lemma}
\begin{proof}
This proof is along the lines of the proof of \cite[Lemma 3]{CsiszarN88}. We will generate the codebook by a random experiment. For fixed vectors $\vecx, \vecx', \tilde{\vecx}, \bar{\vecx}, \vecy,\vecy', \tilde{\vecy}, \bar{\vecy}$ and joint types $P_{X\tilde{X}\tilde{Y}Y}$, $P_{X'\tilde{Y}_1\tilde{Y}_2Y'}$, $P_{\hat{X}\bar{X}\bar{Y}\hat{Y}}$ and $P_{X''\hat{X}_1\hat{X}_2Y''}$ satisfying the conditions of the lemma, we will show that for each of the  statements~\eqref{codebook:1}-\eqref{codebook:5} and \eqref{codebook:1q}-\eqref{codebook:5q}, the probability that the statement does not hold falls doubly exponentially in $n$. Since, $|\cX^n|$, $|\cY^n|$, $|\cP^n_{\cX\times \cX\times \cY\times \cY}|$, $|\cP^n_{\cX\times \cY\times \cY\times \cY}|$ and $|\cP^n_{\cX\times \cX\times \cX\times \cY}|$ grow at most exponentially in $n$, a union bound will imply that the probability that any of the statements~\eqref{codebook:1}-\eqref{codebook:5} and~\eqref{codebook:1q}-\eqref{codebook:5q} fail for some $\vecx, \vecx', \tilde{\vecx}, \bar{\vecx}, \vecy,\vecy', \tilde{\vecy}, \bar{\vecy}$,\,$P_{X\tilde{X}\tilde{Y}Y}$, $P_{X'\tilde{Y}_1\tilde{Y}_2Y'}$, $P_{\hat{X}\bar{X}\bar{Y}\hat{Y}}$ and $P_{X''\hat{X}_1\hat{X}_2Y''}$ also falls doubly exponentially. This will show the existence of a codebook as required by the lemma. The proof will employ \cite[Lemma A1]{CsiszarN88}, which is reproduced below.
\begin{lemma}\cite[Lemma A1]{CsiszarN88}\label{lemma:A1}
Let $Z_1, \ldots, Z_{N}$ be arbitrary random variables, and let $f_i(Z_1, \ldots, Z_i)$ be arbitrary with $0\leq f_i\leq 1$, $i = 1, \ldots, N$. Then the condition
\begin{align}\label{eq:A1}
E\insq{f_i(Z_1,\ldots, Z_i)|Z_1, \ldots, Z_{i-1}}\leq a \textup{ a.s.,} \quad i = 1, \ldots, N, 
\end{align}
implies that 
\begin{align}\label{eq:A2}
\bbP\inb{\frac{1}{N}\sum_{i = 1}^{N}f_i(Z_1, \ldots, Z_i)>t}\leq \exp{\inb{-N(t-a\log{e})}}.
\end{align}
\end{lemma}

We denote the  type classes of $P_{\one}$ and $P_{\two}$ by $T^n_{\one}$ and $T^n_{\two}$ respectively. Let $\vecX_1, \vecX_2, \ldots, \vecX_{N_{\one}}$ be independent random vectors each uniformly distributed on $T^n_{\one}$ and $\vecY_1, \vecY_2, \ldots, \vecY_{N_{\two}}$ be another set of independent random vectors (independent of $\vecX_1, \vecX_2, \ldots, \vecX_{N_{\one}}$) with each element uniformly distributed on $T^n_{\two}$.  $(\vecX_1, \vecX_2, \ldots, \vecX_{N_{\one}})$ and $(\vecY_1, \vecY_2, \ldots, \vecY_{N_{\two}})$ are the random codebooks for user \one and \two respectively.
Fix $P_{X\tilde{X}\tilde{Y}Y}\in \cP^n_{\cX\times \cX\times \cY\times \cY}$, $P_{X'\tilde{Y}_1\tilde{Y}_2Y'}\in \cP^n_{\cX\times \cY\times \cY\times \cY}$, $P_{\hat{X}\bar{X}\bar{Y}\hat{Y}}\in \cP^n_{\cX\times \cX\times \cY\times \cY}$, $P_{X''\hat{X}_1\hat{X}_2Y''}\in \cP^n_{\cX\times \cX\times \cX\times \cY}$, such that $P_X=P_{\tilde{X}} = P_{X'} = P_{\hat{X}} = P_{\hat{X}_1} = P_{\hat{X}_2} = P_{\one}$ and $P_{\tilde{Y}} = P_{\tilde{Y}_1} = P_{\tilde{Y}_2}= P_{\bar{Y}} = P_{\hat{Y}} = P_{Y''} = P_{\two}$, and $(\vecx,\vecy)\in T^n_{XY}$, $(\vecx',\vecy')\in T^n_{X'Y'}$,  $(\tilde{\vecx}, \tilde{\vecy})\in T^n_{X''Y''}$ and $(\bar{\vecx},\bar{\vecy})\in T^n_{\bar{X}\bar{Y}}$\footnote{Note that the lemma statements hold trivially if these conditions do not hold. For example, \eqref{codebook:1} holds trivially if $\vecy\notin T^n_{Y}.$}. We will often use $(\vecx_1, \ldots, \vecx_{N_{\one}})$ and $(\vecy_1, \ldots, \vecy_{N_{\two}})$ as placeholders  to denote the realizations of for  $(\vecX_1, \ldots, \vecX_{N_{\one}})$ and $(\vecY_1, \ldots, \vecY_{N_{\two}})$ respectively (for example, see \eqref{func:g_i}).

\noindent\underline{\em Analysis of \eqref{codebook:3b}}\\
For $i\in [1:n]$, define
\begin{align}\label{func:g_i}
g_i(\vecy_1, \vecy_2, \ldots, \vecy_i) \defineqq \begin{cases}
    1, & \text{if } \vecy_i\in T^n_{\tilde{Y}|{X}Y}(\vecx,\vecy) \\
    0, & \text{otherwise,}
   \end{cases}
\end{align}
and for  $\tilde{\vecy}\in T^n_{\tilde{Y}|{X}Y}(\vecx,\vecy)$,
\begin{align}\label{func:h_i}
h^{\tilde{\vecy}}_i(\vecx_1, \vecx_2, \ldots, \vecx_i) \defineqq \begin{cases}
    1, & \text{if } \vecx_i\in T^n_{\tilde{X}|\tilde{Y}XY}(\tilde{\vecy}, \vecx,\vecy) \\
    0, & \text{otherwise.}
   \end{cases}
\end{align}
Let events $\cE, \cE_1$ and $\cE^{\tilde{\vecy}}_2$ be defined as 
\begin{align*}
\cE &\defineqq \Big\{\left|\inb{(\tilde{m}_{\one}, \tilde{m}_{\two}):(\vecx, \vecX_{\tilde{m}_{\one}}, \vecY_{\tilde{m}_{\two}}, \vecy)\in T^n_{X\tilde{X}\tilde{Y}Y}}\right|\\
&\qquad> \exp\inb{n\inp{|R_{\one}- I(\tilde{X};\tilde{Y}XY)|^{+}+|R_{\two}-I(\tilde{Y};XY)|^{+}+\epsilon}}\Big\},\\
\cE_1 &\defineqq \inb{\sum_{i = 1}^{N_{\two}}g_i(\vecY_1, \vecY_2, \ldots, \vecY_i)>\exp{\inb{n\inp{|R_{\two}-I(\tilde{Y};XY)|^{+}+\frac{\epsilon}{2}}}}}\text{, and }\\
\text{ for }\tilde{\vecy}\in T^n_{\tilde{Y}|{X}Y}(\vecx,\vecy),\hspace{0.2cm}\cE_2^{\tilde{\vecy}} &\defineqq \inb{\sum_{j = 1}^{N_{\one}} h^{\tilde{\vecy}}_j(\vecX_1, \vecX_2, \ldots, \vecX_j)>\exp{\inb{n\inp{|R_{\one}-I(\tilde{X};\tilde{Y}XY)|^{+}+\frac{\epsilon}{2}}}}}.
\end{align*}
We note that
\begin{align*}
&\left|\inb{(\tilde{m}_{\one}, \tilde{m}_{\two}):(\vecx, \vecX_{\tilde{m}_{\one}}, \vecY_{\tilde{m}_{\two}}, \vecy)\in T^n_{X\tilde{X}\tilde{Y}Y}}\right|\\
&\qquad \qquad \qquad= \sum_{i = 1}^{N_{\two}}g_i(\vecY_1, \vecY_2, \ldots, \vecY_i)\inp{\sum_{j = 1}^{N_{\one}} h^{\vecY_i}_j(\vecX_1, \vecX_2, \ldots, \vecX_j)}.
\end{align*}
Thus, 
\begin{align} \cE \subseteq \inp{\cup_{\tilde{\vecy}\in T_{\tilde{Y}|XY}(\vecx, \vecy)}\cE_2^{\tilde{\vecy}}}\cup \cE_1.\label{eq:subsetting}
\end{align} In order to apply Lemma~\ref{lemma:A1} to \eqref{func:g_i} with $(\vecY_1, \ldots, \vecY_{N_{\two}})$ as the random variables $(Z_1, \ldots, Z_N)$, we note that 
\begin{align*}
E\insq{g_i(\vecY_1,\ldots, \vecY_i)|\vecY_1, \ldots, \vecY_{i-1}} = &\bbP\inb{\vecY_i\in T^n_{\tilde{Y}|{X}Y}(\vecx,\vecy)}\\
=&\frac{|T^n_{\tilde{Y}|{X}Y}(\vecx,\vecy)|}{|T^n_{\two}|}\\
\stackrel{\text{(a)}}{\leq}&\frac{\exp\inp{nH(\tilde{Y}|{X}Y)}}{(n+1)^{-|\cY|}\exp\inp{nH(\tilde{Y})}}\\
=&(n+1)^{|\cY|}\exp{\inp{-nI(\tilde{Y};XY)}},
\end{align*}
where (a) follows from \eqref{eq:type_property2},\,\eqref{eq:type_property3} and because $P_{\two}=P_{\tilde{Y}}$ (which implies $|T^n_{\two}| = |T^n_{\tilde{Y}}|$). Taking $t =\frac{1}{N_{\two}}\exp{\inb{n\inp{|R_{\two}-I(\tilde{Y};XY)|^{+}+\frac{\epsilon}{2}}}}$ and $n\geq n_1(\epsilon)$, where $n_1(\epsilon) \defineqq \min{\inb{n:(n+1)^{|\cY|}\log{e}<\frac{1}{2}\exp(\frac{n\epsilon}{2})}},$
we see that $N_{\two}(t-a\log{e})\geq (1/2)\exp(\frac{n\epsilon}{2})$. Using~\eqref{eq:A2}, this gives us
\begin{align} \label{eq:e1}
  \bbP(\cE_1) \leq \exp\inb{-\frac{1}{2}\exp\inb{\frac{n\epsilon}{2}}}.
\end{align}
Similarly, we apply Lemma~\ref{lemma:A1} to \eqref{func:h_i} with $(\vecX_1, \ldots, \vecX_{N_{\one}})$ as the random variables $(Z_1, \ldots, Z_N)$. We can show that $a= (n+1)^{|\cX|}\exp{\inp{-nI(\tilde{X};\tilde{Y}XY)}}$ satisfies \eqref{eq:A1}. We take $t =\frac{1}{N_{\one}}\exp{\inb{n\inp{|R_{\one}-I(\tilde{X};\tilde{Y}XY)|^{+}+\frac{\epsilon}{2}}}}$ and $n\geq n_2(\epsilon)$ where $n_2(\epsilon) \defineqq \min{\inb{n:(n+1)^{|\cX|}\log{e}<\frac{1}{2}\exp(\frac{n\epsilon}{2})}}$. This gives $N_{\one}(t-a\log{e})\geq (1/2)\exp(\frac{n \epsilon}{2})$ which, when plugged into \eqref{eq:A2}, gives 
\begin{align}\label{eq:e2}
\bbP\inp{\cE_2^{\tilde{\vecy}}} \leq \exp\inb{-\frac{1}{2}\exp\inb{\frac{n\epsilon}{2}}}.
\end{align}
Using \eqref{eq:subsetting}, \eqref{eq:e1} and \eqref{eq:e2}, 
\begin{align}\label{eq:prob1}
\bbP\inp{\cE}\leq \inp{|T^n_{\tilde{Y}|XY}(\vecx, \vecy)|+1}\exp\inb{-\frac{1}{2}\exp\inb{\frac{n\epsilon}{2}}}.
\end{align}
This shows that the probability that \eqref{codebook:3b} does not hold falls doubly exponentially.

\noindent\underline{\em Analysis of \eqref{codebook:1}} \\
We will use the same arguments as used in obtaining \eqref{eq:e2}. We replace $\tilde{X}$ with $X$ and $(\tilde{Y}, X, Y)$ with $Y$, to obtain
\begin{align}\label{eq:intermediate:app1:2}
\bbP\inb{\left|\inb{m_{\one}: (\vecX_{\mo}, \vecy)\in T_{XY}^n }\right|> \exp\inb{n\inp{|R_{\one}-I(X;Y)|^{+}+\frac{\epsilon}{2}}}}\leq \exp\inb{-\frac{1}{2}\exp\inb{\frac{n \epsilon}{2}}}.
\end{align}
So, 
\begin{align*}
\bbP\inb{\frac{1}{N_{\one}}\left|\inb{m_{\one}: (\vecX_{\mo}, \vecy)\in T_{XY}^n }\right|> \exp\inb{n\inp{|R_{\one}-I(X;Y)|^{+}-R_{\one}+\frac{\epsilon}{2}}}}\leq \exp\inb{-\frac{1}{2}\exp\inb{\frac{n \epsilon}{2}}}.
\end{align*}
We are given that $I(X;Y)> \epsilon$. When $R_{\one}>I(X;Y)$, we have  $|R_{\one}-I(X;Y)|^{+}$$-R_{\one}+\frac{\epsilon}{2}$ = $\frac{\epsilon}{2} - I(X;Y)\leq  -\frac{\epsilon}{2}$. When $R_{\one}\leq I(X;Y)$, we have  $|R_{\one}-I(X;Y)|^{+}$$-R_{\one}+\frac{\epsilon}{2}$ = $\frac{\epsilon}{2} - R_{\one}\leq  -\frac{\epsilon}{2}$ (because $R_{\one}\geq \epsilon$). Thus, $|R_{\one}-I(X;Y)|^{+}-R_{\one}+\frac{\epsilon}{2} \leq  -\frac{\epsilon}{2}$ and 
\begin{align*}
&\bbP\inb{\frac{1}{N_{\one}}\left|\inb{m_{\one}: (\vecX_{\mo}, \vecy)\in T_{XY}^n }\right|> \exp\inb{\frac{-n\epsilon}{2}}}\\
&\leq \bbP\inb{\frac{1}{N_{\one}}\left|\inb{m_{\one}: (\vecX_{\mo}, \vecy)\in T_{XY}^n }\right|> \exp\inb{n\inp{|R_{\one}-I(X;Y)|^{+}-R_{\one}+\frac{\epsilon}{2}}}}\\
&\leq \exp\inb{-\frac{1}{2}\exp\inb{\frac{n \epsilon}{2}}}.
\end{align*}

\noindent\underline{\em Analysis of \eqref{codebook:2b}}\\
For $i\in [1:N_{\one}]$, let $A_i$ be the set of indices $(j,k)\in [1:N_{\one}]\times[1:N_{\two}],\, j<i$ such that $(\vecx_j, \vecy_k)\in T^{n}_{\tilde{X}\tilde{Y}|Y}(\vecy)$ provided $|A_i|\leq \exp{\inb{n\inp{|R_{\one}- I(\tilde{X};\tilde{Y}Y)|^{+}+|R_{\two}-I(\tilde{Y};Y)|^{+}} +\frac{\epsilon}{4}}}$. Otherwise, $A_i = \emptyset$. That is,
{\footnotesize
\begin{align}\label{A_1}
A_i\defineqq
\begin{cases}
\inb{(j,k)\in [1:N_{\one}]\times[1:N_{\two}]: j<i,  (\vecx_j, \vecy_k)\in T^{n}_{\tilde{X}\tilde{Y}|Y}(\vecy)},&\text{if }|A_i|\leq \exp{\inb{n\inp{|R_{\one}- I(\tilde{X};\tilde{Y}Y)|^{+}+|R_{\two}-I(\tilde{Y};Y)|^{+}} +\frac{\epsilon}{4}}},\\
\emptyset,&\text{ otherwise}.
\end{cases}
\end{align}}

For $\vecy_1,\vecy_2,\ldots,\vecy_{N_{\two}}\in \cY^n$, let
\begin{align}\label{func:f_i}
f_i^{[\vecy_1,\vecy_2,\ldots,\vecy_{N_{\two}}]}\inp{\vecx_1,\vecx_2,\ldots, \vecx_i} \defineqq \begin{cases}
    1, & \text{if } \vecx_i\in \cup_{(j,k)\in A_i} T^n_{X|\tilde{X}\tilde{Y}Y}(\vecx_j,\vecy_k,\vecy) \\
    0, & \text{otherwise.}
   \end{cases}
\end{align}
Then, 
\begin{align}
\bbP&\inb{\sum_{i=1}^{N_\one}f^{[\vecY_1,\vecY_2,\ldots,\vecY_{N_{\two}}]}_{i}\inp{\vecX_1,\vecX_2,\ldots, \vecX_i}\neq \left|\inb{i:\vecX_i\in T^n_{X|\tilde{X}\tilde{Y}Y}(\vecX_j,\vecY_k,\vecy) \text{ for some }j<i \text{ and some }k}\right|}\nonumber\\
=&\bbP\inb{\left|\inb{(\tilde{m}_{\one}, \tilde{m}_{\two}):(\vecX_{\tilde{m}_{\one}}, \vecY_{\tilde{m}_{\two}}, \vecy)\in T^n_{\tilde{X}\tilde{Y}Y}}\right|> \exp{\inb{n\inp{|R_{\one}- I(\tilde{X};\tilde{Y}Y)|^{+}+|R_{\two}-I(\tilde{Y};Y)|^{+}} +\frac{\epsilon}{4}}}}\nonumber\\
\leq& \inp{|T_{\tilde{Y}|Y}(\vecy)|+1}\exp\inb{-\frac{1}{2}\exp\inb{\frac{n\epsilon}{8}}},\label{eq:conditioning}
\end{align}
where the last inequality can be obtained from the definition of event $\cE$ and \eqref{eq:prob1} where we replace $(X, Y)$ with $Y$, $(\vecx, \vecy)$ with $\vecy$, and $\epsilon$ with $\epsilon/4$.

For $\vecy_i\in T^n_{\two}, \, i = 1, \ldots, N_{\two}$, we will now apply Lemma~\ref{lemma:A1} on $f_i^{[\vecy_1,\vecy_2,\ldots,\vecy_{N_{\two}}]}$ with $(\vecX_1,\ldots,\vecX_{N_{\one}})$ as the random variables $(Z_1, \ldots, Z_N)$. We will first compute the value of $a$ in \eqref{eq:A1}. We note that, for $i\in [1:N_{\one}]$, \\
 $E\insq{f^{[\vecy_1,\vecy_2,\ldots,\vecy_{N_{\two}}]}_{i}\inp{\vecX_1,\vecX_2,\ldots, \vecX_i}\Big|\vecX_1,\vecX_2,\ldots, \vecX_{i-1}}$, being a random function of $(\vecX_1,\vecX_2,\ldots, \vecX_{i-1})$, is a random variable. We will compute it for $(\vecX_1,\vecX_2,\ldots, \vecX_{i-1})= (\vecx_1,\vecx_2,\ldots, \vecx_{i-1})$.

\begin{align*}
&E\insq{f^{[\vecy_1,\vecy_2,\ldots,\vecy_{N_{\two}}]}_{i}\inp{\vecX_1,\vecX_2,\ldots, \vecX_i}\Big|(\vecX_1,\vecX_2,\ldots, \vecX_{i-1}) = (\vecx_1,\vecx_2,\ldots, \vecx_{i-1})}\\
&=\bbP\inp{\vecX_i\in \cup_{(j,k)\in A_i} T^n_{X|\tilde{X}\tilde{Y}Y}(\vecx_j,\vecy_k,\vecy)}\\
&\stackrel{\text{(a)}}{\leq} |A_{i}|\frac{\exp\inb{nH(X|\tilde{X}\tilde{Y}Y)}}{(n+1)^{-|\cX|}\exp(nH(X))}\\
& = (n+1)^{|\cX|}\exp\inb{n\inp{|R_{\one}- I(\tilde{X};\tilde{Y}Y)|^{+}+|R_{\two}-I(\tilde{Y};Y)|^{+}} -I(X;\tilde{X}\tilde{Y}Y)+\frac{\epsilon}{4}},
\end{align*}
where (a) follows from \eqref{eq:type_property2},\,\eqref{eq:type_property3}, a union bound over $(j,k)\in A_{i}$ and by noting that $|T^n_{\one}| = |T^n_{X}|$. For all $i\in [1:N_{\one}]$, this upper bound holds for every realization of $(\vecX_1,\vecX_2,\ldots, \vecX_{i-1})$. Thus, in \eqref{eq:A1}, we may take $a=(n+1)^{|\cX|}\exp$$\Big\{n\Big(|R_{\one}- I(\tilde{X};\tilde{Y}Y)|^{+}$$+|R_{\two}-I(\tilde{Y};Y)|^{+}\Big) -I(X;\tilde{X}\tilde{Y}Y)+\frac{\epsilon}{4}\Big\}$.
If $I(X;\tilde{X}\tilde{Y}Y) >|R_{\one}- I(\tilde{X};\tilde{Y}Y)|^{+}+|R_{\two}-I(\tilde{Y};Y)|^{+} +\epsilon$ (as postulated in~\eqref{codebook:2b}), \eqref{eq:A1} holds with $a = (n+1)^{|\cX|}\exp\inb{-\frac{3}{4}n\epsilon}$. For $t = \exp\inb{\frac{-n\epsilon}{2}}$ and $n\geq n_2(\epsilon)$ with $n_2(\epsilon) \defineqq \min{\inb{n:(n+1)^{|\cX|}\log e<\frac{1}{2}\exp{\inb{\frac{n\epsilon}{4}}}}},$ 
we get
\begin{align}
\bbP&\inb{\frac{1}{N_{\one}}\sum_{i=1}^{N_{\one}}f_i^{[\vecy_1,\vecy_2,\ldots,\vecy_{N_{\two}}]}\inp{\vecX_1,\vecX_2,\ldots, \vecX_i}>\exp\inb{\frac{-n\epsilon}{2}}}\nonumber \\
&\leq \exp\inb{-\frac{N_{\one}}{2}\exp\inb{-\frac{n\epsilon}{2}}}\nonumber\\
&\leq \exp\inb{-\frac{1}{2}\exp\inb{\frac{n\epsilon}{2}}}, \nonumber
\end{align}
where the last inequality uses the assumption that $N_\one\geq \exp\inb{n\epsilon}$.  Averaging over $(\vecY_1, \ldots, \vecY_{\two})$, we get
\begin{align}
\bbP&\inb{\frac{1}{N_{\one}}\sum_{i=1}^{N_{\one}}f_i^{[\vecY_1,\vecY_2,\ldots,\vecY_{N_{\two}}]}\inp{\vecX_1,\vecX_2,\ldots, \vecX_i}>\exp\inb{\frac{-n\epsilon}{2}}}\nonumber \\
&\leq \exp\inb{-\frac{1}{2}\exp\inb{\frac{n\epsilon}{2}}}. \label{eq:prob}
\end{align}
Let events $\cF_1$ and $\cF_2$ be defined as 
\begin{align*}
\cF_1 &\defineqq \inb{\frac{1}{N_{\one}}\left|\inb{i:\vecX_i\in T^n_{X|\tilde{X}\tilde{Y}Y}(\vecX_j, \vecY_k, \vecy)\text{ for some } j<i \text{ and some }k}\right|>\exp\inb{\frac{-n\epsilon}{2}}},\\
\cF_2 &\defineqq \inb{\sum_{i=1}^{N_\one}f^{[\vecY_1,\vecY_2,\ldots,\vecY_{N_{\two}}]}_{i}\inp{\vecX_1,\vecX_2,\ldots, \vecX_i}\neq \left|\inb{i:\vecX_i\in T^n_{X|\tilde{X}\tilde{Y}Y}(\vecX_j,\vecY_k,\vecy) \text{ for some }j<i \text{ and some }k}\right|},\\
\cF_3 &\defineqq \inb{\sum_{i=1}^{N_\one}f^{[\vecY_1,\vecY_2,\ldots,\vecY_{N_{\two}}]}_{i}\inp{\vecX_1,\vecX_2,\ldots, \vecX_i}>\exp\inb{\frac{-n\epsilon}{2}}}.
\end{align*}
We are interested in $\bbP\inp{\cF_1}$. We see that 
\begin{align*}
\bbP\inp{\cF_1}& = \bbP\inp{\cF_1\cap \cF_2} +\bbP\inp{\cF_1\cap \cF^{c}_2}\\
& \leq \bbP\inp{\cF_2} +\bbP\inp{\cF_1\cap \cF^{c}_2}\\
& \leq \bbP\inp{\cF_2} +\bbP\inp{\cF_3}\\
&\stackrel{\text{(a)}}{\leq}\inp{|T_{\tilde{Y}|Y}(\vecy)|+1}\exp\inb{-\frac{1}{2}\exp\inb{\frac{n\epsilon}{8}}}+\exp\inb{-\frac{1}{2}\exp\inb{\frac{n\epsilon}{2}}}\\
& \leq \inp{|T_{\tilde{Y}|Y}(\vecy)|+2}\exp\inb{-\frac{1}{2}\exp\inb{\frac{n\epsilon}{8}}},
\end{align*}
where (a) follows from \eqref{eq:conditioning} and \eqref{eq:prob}. Thus, 
\begin{align*}
\bbP&\inp{\frac{1}{N_{\one}}\left|\inb{i:\vecX_i\in T^n_{X|\tilde{X}\tilde{Y}Y}(\vecX_j, \vecY_k, \vecy)\text{ for some}, j<i \text{ and }k}\right|>\exp\inb{\frac{-n\epsilon}{2}}}\\
&\leq \inp{|T_{\tilde{Y}|Y}(\vecy)|+2}\exp\inb{-\frac{1}{2}\exp\inb{\frac{n\epsilon}{8}}}.
\end{align*}
By symmetry, we get the same upper bound when $j>i$. Thus,
\begin{align*}
\bbP&\inb{\frac{\left|\inb{\mo:(\vecX_{\mo}, \vecX_{\tilde{m}_{\one}}, \vecY_{\mt}, \vecy)\in T^n_{X\tilde{X}\tilde{Y}Y} \text{ for some }\tilde{m}_{\one}\neq \mo\text{ and some }\mt} \right|}{N_{\one}} >\exp\inb{-n\epsilon/2} }\\
&< 2\inp{|T_{\tilde{Y}|Y}(\vecy)|+2}\exp\inb{-\frac{1}{2}\exp\inb{\frac{n\epsilon}{8}}}.
\end{align*}
This completes the analysis for \eqref{codebook:2b}. 

\noindent\underline{\em Analysis of \eqref{codebook:5}}\\
We will split the analysis in two parts as suggested by the inequalities below, where the first inequality is a union bound,
\begin{align*}
&\bbP\inb{\left|\inb{(\tilde{m}_{\two 1}, \tilde{m}_{\two 2}):(\vecx', \vecY_{\tilde{m}_{\two 1}}, \vecY_{\tilde{m}_{\two 2}}, \vecy')\in T^n_{X'\tilde{Y}_1\tilde{Y}_2Y'} }\right|
> \exp\inb{n\inp{|R_{\two}-I(\tilde{Y}_1;X'Y')|^{+}+|R_{\two}-I(\tilde{Y}_2;\tilde{Y}_1 X'Y')|^{+}+\epsilon}}}\\
&\leq \bbP\Big\{\left|\inb{(\tilde{m}_{\two 1}, \tilde{m}_{\two 2}):\tilde{m}_{\two 1}\neq  \tilde{m}_{\two 2},(\vecx', \vecY_{\tilde{m}_{\two 1}}, \vecY_{\tilde{m}_{\two 2}}, \vecy')\in T^n_{X'\tilde{Y}_1\tilde{Y}_2Y'} }\right|\\
 &\qquad\qquad>\frac{1}{2}\exp\inb{n\inp{|R_{\two}-I(\tilde{Y}_1;X'Y')|^{+}+|R_{\two}-I(\tilde{Y}_2;\tilde{Y}_1 X'Y')|^{+}+\epsilon}}\Big\}\\
&+\bbP\Big\{\left|\inb{(\tilde{m}_{\two 1}, \tilde{m}_{\two 2}):\tilde{m}_{\two 1}= \tilde{m}_{\two 2},(\vecx', \vecY_{\tilde{m}_{\two 1}}, \vecY_{\tilde{m}_{\two 2}}, \vecy')\in T^n_{X'\tilde{Y}_1\tilde{Y}_2Y'} }\right|\\
&\qquad\qquad> \frac{1}{2}\exp\inb{n\inp{|R_{\two}-I(\tilde{Y}_1;X'Y')|^{+}+|R_{\two}-I(\tilde{Y}_2;\tilde{Y}_1 X'Y')|^{+}+\epsilon}}\Big\}\\
&\leq \bbP\Big\{\left|\inb{(\tilde{m}_{\two 1}, \tilde{m}_{\two 2}):\tilde{m}_{\two 1}\neq  \tilde{m}_{\two 2},(\vecx', \vecY_{\tilde{m}_{\two 1}}, \vecY_{\tilde{m}_{\two 2}}, \vecy')\in T^n_{X'\tilde{Y}_1\tilde{Y}_2Y'} }\right|\\
 &\qquad\qquad>\exp\inb{n\inp{|R_{\two}-I(\tilde{Y}_1;X'Y')|^{+}+|R_{\two}-I(\tilde{Y}_2;\tilde{Y}_1 X'Y')|^{+}+\epsilon'}}\Big\}\\
&+\bbP\Big\{\left|\inb{(\tilde{m}_{\two 1}, \tilde{m}_{\two 2}):\tilde{m}_{\two 1}= \tilde{m}_{\two 2},(\vecx', \vecY_{\tilde{m}_{\two 1}}, \vecY_{\tilde{m}_{\two 2}}, \vecy')\in T^n_{X'\tilde{Y}_1\tilde{Y}_2Y'} }\right|\\
&\qquad\qquad> \exp\inb{n\inp{|R_{\two}-I(\tilde{Y}_1;X'Y')|^{+}+|R_{\two}-I(\tilde{Y}_2;\tilde{Y}_1 X'Y')|^{+}+\epsilon'}}\Big\}
\end{align*}
for $\epsilon' = \epsilon/2$ and sufficiently large $n$.
We first consider the first term  where $\tilde{m}_{\two 1}\neq \tilde{m}_{\two 2}$.
We follow arguments similar to those for \eqref{codebook:3b} and get an upper bound. We define
\begin{align}\label{func:tg}
\tilde{g}_i(\vecy_1, \vecy_2, \ldots, \vecy_i) \defineqq \begin{cases}
    1, & \text{if } \vecy_i\in T^n_{\tilde{Y}_1|{X}'Y'}(\vecx',\vecy') \\
    0, & \text{otherwise.}
   \end{cases}
\end{align}
For $\tilde{\vecy}\in T^n_{\tilde{Y}_1|{X}'Y'}(\vecx',\vecy')$,
\begin{align}\label{func:th}
\tilde{h}^{\tilde{\vecy}}_i(\vecy_1, \vecy_2, \ldots, \vecy_i) \defineqq \begin{cases}
    1, & \text{if } \vecy_i\in T^n_{\tilde{Y}_2|\tilde{Y}_1X'Y'}(\tilde{\vecy}, \vecx',\vecy') \\
    0, & \text{otherwise.}
   \end{cases}
\end{align}
Define events $\tilde{\cE}$ and $ \tilde{\cE}_1$ as 
\begin{align*}
\tilde{\cE} = &\Big\{\left|\inb{(\tilde{m}_{\two 1}, \tilde{m}_{\two 2}):\tilde{m}_{\two 1}\neq \tilde{m}_{\two 2},\,(\vecx', \vecY_{\tilde{m}_{\two 1}}, \vecY_{\tilde{m}_{\two 2}}, \vecy')\in T^n_{X'\tilde{Y}_1\tilde{Y}_2Y'}}\right|\\
&\qquad> \exp\inb{n\inp{|R_{\two}- I(\tilde{Y}_2;X'Y'\tilde{Y}_1)|^{+}+|R_{\two}-I(\tilde{Y}_1;X'Y')|^{+}+\epsilon'}}\Big\},\\
\tilde{\cE}_1 =& \inb{\sum_{i = 1}^{N_{\two}}\tilde{g}_i(\vecY_1, \vecY_2, \ldots, \vecY_i)>\exp{\inb{n\inp{|R_{\two}-I(\tilde{Y}_1;X'Y')|^{+}+\frac{\epsilon'}{2}}}}}.
\end{align*}
Let $R'_{\two}\defineqq \frac{\log{\inp{N_{\two}-1}}}{n} =\frac{\log{\inp{2^{nR_{\two}}-1}}}{n} $. For $i\in [1:N_{\two}]$ and $\tilde{\vecy}\in T^n_{\tilde{Y}_1|{X}'Y'}(\vecx',\vecy')$, define events $\tilde{\cE}_{2}^{i, \tilde{\vecy}}$ and $\tilde{\cE}_{2,\dagger}^{i, \tilde{\vecy}}$ as 
\begin{align*}
\tilde{\cE}_{2}^{i, \tilde{\vecy}} &= \inb{\sum_{j = 1, j\neq i}^{N_{\two}} \tilde{h}^{\tilde{\vecy}}_j(\vecY_1, \vecY_2, \ldots, \vecY_j)>\exp{\inb{n\inp{|R_{\two}-I(\tilde{Y}_2;\tilde{Y}_1X'Y')|^{+}+\frac{\epsilon'}{2}}}}}\\
\tilde{\cE}_{2,\dagger}^{i, \tilde{\vecy}} &= \inb{\sum_{j = 1, j\neq i}^{N_{\two}} \tilde{h}^{\tilde{\vecy}}_j(\vecY_1, \vecY_2, \ldots, \vecY_j)>\exp{\inb{n\inp{|R'_{\two}-I(\tilde{Y}_2;\tilde{Y}_1X'Y')|^{+}+\frac{\epsilon'}{2}}}}}
\end{align*}
Note that 
\begin{align*}
&\left|\inb{(\tilde{m}_{\two 1}, \tilde{m}_{\two 2}):\tilde{m}_{\two 1}\neq  \tilde{m}_{\two 2}\text{ and }(\vecx', \vecY_{\tilde{m}_{\two 1}}, \vecY_{\tilde{m}_{\two 2}}, \vecy')\in T^n_{X'\tilde{Y}_1\tilde{Y}_2Y'}}\right|\\
&\qquad \qquad \qquad= \sum_{i = 1}^{N_{\two}}\tilde{g}_i(\vecY_1, \vecY_2, \ldots, \vecY_i)\inp{\sum_{j = 1, j\neq i}^{N_{\two}} \tilde{h}^{\vecY_i}_j(\vecY_1, \vecY_2, \ldots, \vecY_j)}.
\end{align*}
We see that 
\begin{align*}
\tilde{\cE} &\subseteq \inp{\cup_{i\in 2^{nR_{\two}}}\cup_{\tilde{\vecy}\in T_{\tilde{Y}_1|X'Y'}(\vecx', \vecy')}\tilde{\cE}_2^{i,\tilde{\vecy}}}\cup \tilde{\cE_1}. 
\end{align*} Notice that we need to only consider union over $\tilde{\vecy}\in T_{\tilde{Y}_1|X'Y'}(\vecx', \vecy')$ and not  $\tilde{\vecy}\in \cY^n$. 
Since, $\tilde{\cE}_{2}^{i, \tilde{\vecy}}\subseteq \tilde{\cE}_{2,\dagger}^{i, \tilde{\vecy}}$ for all $i\in [1:N_{\two}]$,
\begin{align}
\tilde{\cE}&\subseteq \inp{\cup_{i\in 2^{nR_{\two}}}\cup_{\tilde{\vecy}\in T_{\tilde{Y}_1|X'Y'}(\vecx', \vecy')}\tilde{\cE}_{2,\dagger}^{i,\tilde{\vecy}}}\cup \tilde{\cE_1}.\label{eq:intermediate:appA:1}
\end{align}

We apply Lemma~\ref{lemma:A1} to \eqref{func:tg} with $(\vecY_1, \ldots, \vecY_{N_{\two}})$ as the random variables $(Z_1, \ldots, Z_N)$. We can show that $a= (n+1)^{|\cY|}\exp{\inp{-nI(\tilde{Y}_1;X'Y')}}$ satisfies \eqref{eq:A1}. We take $t =\frac{1}{N_{\two}}\exp{\inb{n\inp{|R_{\two}-I(\tilde{Y}_1;X'Y')|^{+}+\frac{\epsilon'}{2}}}}$ and $n\geq n_1(\epsilon')$ (recall that $n_1(\epsilon') = \min{\inb{n:(n+1)^{|\cY|}\log{e}<\frac{1}{2}\exp(\frac{n\epsilon'}{2})}}$). This gives $N_{\two}(t-a\log{e})\geq (1/2)\exp(\frac{n \epsilon'}{2})$ which, when plugged into \eqref{eq:A2}, gives 
\begin{align}\label{eq:te1}
\bbP\inp{\tilde{\cE}_1} \leq \exp\inb{-\frac{1}{2}\exp\inb{\frac{n\epsilon'}{2}}}.
\end{align}

Similarly, for $i\in [1:N_{\two}]$, we can apply Lemma~\ref{lemma:A1} to \eqref{func:th} with $(\vecY_1, \ldots,\vecY_{i-1}, \vecY_{i+1}, \vecY_{N_{\two}})$ as the random variables $(Z_1, \ldots, Z_N)$. We can show that $a= (n+1)^{|\cY|}\exp{\inp{-nI(\tilde{Y}_2;\tilde{Y}_1X'Y')}}$ satisfies \eqref{eq:A1}. \\Choose, $t =\frac{1}{N_{\two}-1}\exp{\inb{n\inp{|R'_{\two}-I(\tilde{Y}_2;\tilde{Y}_1X'Y')|^{+}+\frac{\epsilon'}{2}}}}$ and $n\geq n_1(\epsilon')$ to obtain

\begin{align}\label{eq:te2}
\bbP\inp{\tilde{\cE}_{2,\dagger}^{i, \tilde{\vecy}}} \leq \exp\inb{-\frac{1}{2}\exp\inb{\frac{n\epsilon'}{2}}}, \,\text{ for } \tilde{\vecy}\in T_{\tilde{Y}_1|X'Y'}(\vecx', \vecy').
\end{align}
Using \eqref{eq:intermediate:appA:1}, \eqref{eq:te1} and \eqref{eq:te2}, we see that
\begin{align*}
\bbP\inp{\tilde{\cE}}\leq \inp{2^{nR_{\two}}|T_{\tilde{Y}_1|X'Y'}(\vecx', \vecy')|+1}\exp\inb{-\frac{1}{2}\exp\inb{\frac{n\epsilon'}{2}}}.
\end{align*}
When $\tilde{m}_{\two 1}= \tilde{m}_{\two 2}$ and $P_{\tilde{Y}_1\tilde{Y}_2}$ is such that $\tilde{Y}_1\neq \tilde{Y}_2$,
\begin{align*}
\left|\inb{(\tilde{m}_{\two 1}, \tilde{m}_{\two 2}):(\vecx', \vecY_{\tilde{m}_{\two 1}}, \vecY_{\tilde{m}_{\two 2}}, \vecy')\in T^n_{X'\tilde{Y}_1\tilde{Y}_2Y'}}\right| = 0 \text{ w.p. }1.
\end{align*} When $\tilde{m}_{\two 1}= \tilde{m}_{\two 2}$ and $P_{\tilde{Y}_1\tilde{Y}_2}$ is such that $\tilde{Y}_1= \tilde{Y}_2$,
\begin{align*}
&\bbP\inb{\left|\inb{(\tilde{m}_{\two 1}, \tilde{m}_{\two 2}):(\vecx', \vecY_{\tilde{m}_{\two 1}}, \vecY_{\tilde{m}_{\two 2}}, \vecy')\in T^n_{X'\tilde{Y}_1\tilde{Y}_2Y'}}\right|> \exp\inb{n\inp{|R_{\two}- I(\tilde{Y}_2;X'Y'\tilde{Y}_1)|^{+}+|R_{\two}-I(\tilde{Y}_1;X'Y')|^{+}+\epsilon'}}}\\
&\leq \bbP\inb{\left|\inb{\tilde{m}_{\two 1}:(\vecx',  \vecY_{\tilde{m}_{\two 1}}, \vecy')\in T^n_{X'\tilde{Y}_1Y'}}\right|> \exp\inp{n\inp{|R_{\two}-I(\tilde{Y}_1;X'Y')|^{+}+\epsilon'}}}\\
&\leq \exp\inb{-\frac{1}{2}\exp(n\epsilon')}.
\end{align*}
The last  inequality follows from \eqref{eq:intermediate:app1:2} where we substitute $X$ with $\tilde{Y}_1$, $Y$ with $(X',Y')$, $\vecy$ with $(\vecx', \vecy')$, $R_{\one}$ with $R_{\two}$ and $\epsilon/2$ with $\epsilon'$. 
Thus, 
\begin{align}
&\bbP\inb{\left|\inb{(\tilde{m}_{\two 1}, \tilde{m}_{\two 2}):(\vecx', \vecY_{\tilde{m}_{\two 1}}, \vecY_{\tilde{m}_{\two 2}}, \vecy')\in T^n_{X'\tilde{Y}_1\tilde{Y}_2Y'} }\right|
> \exp\inb{n\inp{|R_{\two}-I(\tilde{Y}_1;X'Y')|^{+}+|R_{\two}-I(\tilde{Y}_2;\tilde{Y}_1 X'Y')|^{+}+\epsilon'}}}\nonumber\\
&\leq \inp{2^{nR_{\two}}|T_{\tilde{Y}|X'Y'}(\vecx', \vecy')|+1}\exp\inb{-\frac{1}{2}\exp\inb{\frac{n\epsilon'}{2}}} + \exp\inb{-\frac{1}{2}\exp(n\epsilon')}\nonumber\\
&\leq \inp{2^{nR_{\two}}|T_{\tilde{Y}|X'Y'}(\vecx', \vecy')|+1}\exp\inb{-\frac{1}{2}\exp\inb{\frac{n\epsilon}{4}}} + \exp\inb{-\frac{1}{2}\exp\inb{\frac{n\epsilon}{2}}}.\label{eq:prob2}
\end{align}
This completes the analysis of \eqref{codebook:5}.\\

\noindent\underline{\em Analysis of \eqref{codebook:4}}\\
Let $A$ be the set of indices $(j,k)\in [1:N_{\two}]\times[1:N_{\two}]$  such that $(\vecy_j, \vecy_k)\in T^{n}_{\tilde{Y}_{1}\tilde{Y}_2|Y'}(\vecy')$ provided \\$|A|\leq \exp{\inb{n\inp{|R_{\two}- I(\tilde{Y}_2;\tilde{Y}_1Y')|^{+}+|R_{\two}-I(\tilde{Y}_1;Y')|^{+}} +\frac{\epsilon}{4}}}$. Otherwise, $A = \emptyset$.
Let 
\begin{align*}
\tilde{f}_i^{[\vecy_1,\vecy_2,\ldots,\vecy_{N_{\two}}]}\inp{\vecx_1,\vecx_2,\ldots, \vecx_i} = \begin{cases}
    1, & \text{if } \vecx_i\in \cup_{(j,k)\in A} T^n_{X'|\tilde{Y}_1\tilde{Y}_2Y'}(\vecy_j,\vecy_k,\vecy') \\
    0, & \text{otherwise.}
   \end{cases}
\end{align*}
\begin{align}
\bbP&\inb{\sum_{i=1}^{N_\one}\tilde{f}^{[\vecY_1,\vecY_2,\ldots,\vecY_{N_{\two}}]}_{i}\inp{\vecX_1,\vecX_2,\ldots, \vecX_i}\neq \left|\inb{i:\vecX_i\in T^n_{X|\tilde{Y}_1\tilde{Y}_2Y'}(\vecY_j,\vecY_k,\vecy') \text{ for some }j\text{ and } k \right|}}\nonumber\\
=&\bbP\Bigg\{\left|\inb{(\tilde{m}_{\two1}, \tilde{m}_{\two2}):(\vecY_{\tilde{m}_{\two 1}}, \vecY_{\tilde{m}_{\two 2}}, \vecy')\in T^n_{\tilde{Y}_1\tilde{Y}_2Y'}}\right|\\
&\qquad \qquad \qquad\qquad> \exp{\inb{n\inp{|R_{\two}- I(\tilde{Y}_1;\tilde{Y}_2Y')|^{+}+|R_{\two}-I(\tilde{Y}_1;Y')|^{+}} +\frac{\epsilon}{4}}}\Bigg\}\nonumber\\
\leq& \inp{2^{nR_{\two}}|T_{\tilde{Y}|Y'}(\vecy')|+1}\exp\inb{-\frac{1}{2}\exp\inb{\frac{n\epsilon}{8}}}+ \exp\inb{-\frac{1}{2}\exp\inb{{n\epsilon}}},\label{eq:conditioning2}
\end{align}
where last inequality follows from \eqref{eq:prob2} by replacing $(\vecx', \vecy')$ with $\vecy'$, $(\vecX',\vecY')$ with $\vecY'$ and $\frac{\epsilon}{2}$ (or $\epsilon'$) with $\frac{\epsilon}{4}$.

For $\vecy_1,\vecy_2,\ldots,\vecy_{N_{\two}}\in \cY^n$, we will apply Lemma~\ref{lemma:A1} on $\tilde{f}_i^{[\vecy_1,\vecy_2,\ldots,\vecy_{N_{\two}}]}$ with $(\vecX_1,\ldots,\vecX_{N_{\one}})$ as the random variables $(Z_1, \ldots, Z_N)$. We will first compute the value of $a$ in \eqref{eq:A1}. 

\begin{align*}
&E\insq{\tilde{f}^{[\vecy_1,\vecy_2,\ldots,\vecy_{N_{\two}}]}_{i}\inp{\vecX_1,\vecX_2,\ldots, \vecX_i}\Big|(\vecX_1,\vecX_2,\ldots, \vecX_{i-1})}\\
&=\bbP\inp{\vecX_i\in \cup_{(j,k)\in A} T^n_{X'|\tilde{Y}_1\tilde{Y}_2Y'}(\vecy_j,\vecy_k,\vecy')}\\
&\stackrel{\text{(a)}}{\leq} |A|\frac{\exp\inb{nH(X'|\tilde{Y}_1\tilde{Y}_2Y')}}{(n+1)^{-|\cX|}\exp(nH(X'))}\\
& \leq  (n+1)^{|\cX|}\exp\inb{n\inp{|R_{\two}- I(\tilde{Y}_2;\tilde{Y}_1Y')|^{+}+|R_{\two}-I(\tilde{Y}_1;Y')|^{+}} -I(X';\tilde{Y}_1\tilde{Y}_2Y')+\frac{\epsilon}{4}}.
\end{align*}
where (a) follows from \eqref{eq:type_property2},\,\eqref{eq:type_property3}, a union bound over $(j,k)\in A$ and by noting that $|T^n_{\one}| = |T^n_{X'}|$. 
If $I(X';\tilde{Y}_1\tilde{Y}_2Y') >|R_{\two}- I(\tilde{Y}_2;\tilde{Y}_1Y')|^{+}+|R_{\two}-I(\tilde{Y}_1;Y')|^{+} +\epsilon$ (which~\eqref{codebook:4} postulates), \eqref{eq:A1} holds with $a = (n+1)^{|\cX|}\exp\inb{-\frac{3}{4}n\epsilon}$. For $t = \exp\inb{\frac{-n\epsilon}{2}}$ and $n\geq n_2(\epsilon)$ (recall that $n_2(\epsilon) = \min{\inb{n:(n+1)^{|\cX|}\log e<\frac{1}{2}\exp{\inb{\frac{n\epsilon}{4}}}}},$ 
we get
\begin{align}
\bbP&\inb{\frac{1}{N_{\one}}\sum_{i=1}^{N_{\one}}\tilde{f}_i^{[\vecy_1,\vecy_2,\ldots,\vecy_{N_{\two}}]}\inp{\vecX_1,\vecX_2,\ldots, \vecX_i}>\exp\inb{\frac{-n\epsilon}{2}}}\nonumber \\
&\leq \exp\inb{-\frac{N_{\one}}{2}\exp\inb{-\frac{n\epsilon}{2}}}\nonumber\\
&\leq \exp\inb{-\frac{1}{2}\exp\inb{\frac{n\epsilon}{2}}} \nonumber
\end{align}
where the last inequality uses the assumption that $N_\one\geq \exp\inb{n\epsilon}$. Averaging over $(\vecY_1, \ldots, \vecY_{\two})$, we get
\begin{align}
\bbP&\inb{\frac{1}{N_{\one}}\sum_{i=1}^{N_{\one}}\tilde{f}_i^{[\vecY_1,\vecY_2,\ldots,\vecY_{N_{\two}}]}\inp{\vecX_1,\vecX_2,\ldots, \vecX_i}>\exp\inb{\frac{-n\epsilon}{2}}}\nonumber \\
&\leq \exp\inb{-\frac{1}{2}\exp\inb{\frac{n\epsilon}{2}}}.   \label{eq:tprob}
\end{align}
Let events $\tilde{\cF_1}$ and $\tilde{\cF_2}$ be defined as 
\begin{align*}
\tilde{\cF_1} &= \inb{\frac{1}{N_{\one}}\left|\inb{i:\vecX_i\in T^n_{X|\tilde{Y}_1\tilde{Y}_2Y'}(\vecY_j, \vecY_k, \vecy')\text{ for some } j\text{ and }k}\right|>\exp\inb{\frac{-n\epsilon}{2}}},\\
\tilde{\cF_2} &= \inb{\sum_{i=1}^{N_\one}\tilde{f}^{[\vecY_1,\vecY_2,\ldots,\vecY_{N_{\two}}]}_{i}\inp{\vecX_1,\vecX_2,\ldots, \vecX_i}\neq \left|\inb{i:\vecX_i\in T^n_{X|\tilde{Y}_1\tilde{Y}_2Y'}(\vecY_j,\vecY_k,\vecy') \text{ for some }j\text{ and }k}\right|},\\
\tilde{\cF_3} &= \inb{\sum_{i=1}^{N_\one}\tilde{f}^{[\vecY_1,\vecY_2,\ldots,\vecY_{N_{\two}}]}_{i}\inp{\vecX_1,\vecX_2,\ldots, \vecX_i}>\exp\inb{\frac{-n\epsilon}{2}}}.
\end{align*}
We are interested in $\bbP\inp{\tilde{\cF_1}}$. We see that 
\begin{align*}
\bbP\inp{\tilde{\cF}_1}& = \bbP\inp{\tilde{\cF}_1\cap \tilde{\cF}_2} +\bbP\inp{\tilde{\cF}_1\cap \tilde{\cF}^{c}_2}\\
& \leq \bbP\inp{\tilde{\cF}_2} +\bbP\inp{\tilde{\cF}_1\cap \tilde{\cF}^{c}_2}\\
& \leq \bbP\inp{\tilde{\cF}_2} +\bbP\inp{\tilde{\cF}_3}\\
&\stackrel{\text{(a)}}{\leq} \inp{2^{nR_{\two}}|T_{\tilde{Y}|Y'}(\vecy')|+1}\exp\inb{-\frac{1}{2}\exp\inb{\frac{n\epsilon}{8}}}+ \exp\inb{-\frac{1}{2}\exp\inb{{n\epsilon}}}+\exp\inb{-\frac{1}{2}\exp\inb{\frac{n\epsilon}{2}}}\\
& =\inp{2^{nR_{\two}}|T_{\tilde{Y}|Y'}(\vecy')|+3}\exp\inb{-\frac{1}{2}\exp\inb{\frac{n\epsilon}{8}}},
\end{align*}
where (a) follows from \eqref{eq:conditioning2} and \eqref{eq:tprob}.  Statements \eqref{codebook:1q}-\eqref{codebook:5q} are analogous to \eqref{codebook:1}-\eqref{codebook:5} with the roles of users \one and \two are interchanged. In particular, $(\vecx, \vecx', \vecy, \vecy')$ is replaced with $(\bar{\vecy}, \tilde{\vecy}, \bar{\vecx}, \tilde{\vecx})$, $(X, \tilde{X}, \tilde{Y}, Y)$ with $(\bar{Y}, \hat{Y}, \hat{X}, \bar{X})$, $(X', \tilde{Y}_1, \tilde{Y}_2, Y')$ with $(Y'', \hat{X}_1, \hat{X}_2, X'')$ and $R_{\one}$ with $R_{\two}$. Thus, we can follow a similar analysis and show that the probability that that any of the statements \eqref{codebook:1q}, \eqref{codebook:2bq}, \eqref{codebook:3bq}, \eqref{codebook:4q} or \eqref{codebook:5q} do not hold for  some $\tilde{\vecx}, \bar{\vecx},  \tilde{\vecy}, \bar{\vecy}$,\, $P_{\hat{X}\bar{X}\bar{Y}\hat{Y}}$ and $P_{X''\hat{X}_1\hat{X}_2Y''}$ also falls doubly exponentially. 
\end{proof}

\section{Proof of Lemma~\ref{lemma:disambiguity}}\label{app:disambiguity}
\begin{Lemma}
{Suppose $\alpha>0$. For a channel which is not \one-spoofable, for sufficiently small $\eta>0$,} there does not exist a distribution $P_{XY\tilde{X}Y'X'\tilde{Y}Z}\in \cP^n_{XY\tilde{X}Y'X'\tilde{Y}Z}$ with $\min_{x}P_X(x), \min_{\tilde{x}}P_{\tilde{X}}(\tilde{x}), \min_{\tilde{y}}P_{\tilde{Y}}(\tilde{y})\geq \alpha$  which  satisfies the following:
\begin{enumerate}[label=(\Alph*)]
	\item $P_{XYZ}\in \blue{\cD_{\eta}}$,
	\item $P_{\tilde{X}Y'Z}\in \blue{\cD_{\eta}}$,
	\item $P_{X'\tilde{Y}Z}\in \blue{\cD_{\eta}}$,
	\item $I(\tilde{X}\tilde{Y};XZ|Y)<\eta$,
	\item $I(X\tilde{Y};\tilde{X}Z|Y')<\eta$ and
	\item $I(X\tilde{X};\tilde{Y}Z|X')<\eta$.
\end{enumerate}
{Similarly, for a channel which is not \two-\spoofable, there does not exist a distribution $P_{X'_1\tilde{Y}_1X'_2\tilde{Y}_2XYZ}\in \cP^n_{X'_1\tilde{Y}_1X'_2\tilde{Y}_2XYZ}$ with $\min_{x}P_X(x), \min_{\tilde{y}_1}P_{\tilde{Y}_1}(\tilde{y}_1), \min_{\tilde{y}_2}P_{\tilde{Y}_2}(\tilde{y}_2)\geq \alpha>0$ which satisfies the following:
\begin{enumerate}[label=(\Alph*)]
\setcounter{enumi}{6}
	\item $P_{XYZ}\in \blue{\cD_{\eta}}$,
	\item $P_{X'_1\tilde{Y}_1Z}\in \blue{\cD_{\eta}}$,
	\item $P_{X'_2\tilde{Y}_2Z}\in \blue{\cD_{\eta}}$,
	\item $I(\tilde{Y}_1\tilde{Y}_2;XZ|Y)<\eta$,
	\item $I(X\tilde{Y}_2;\tilde{Y}_1Z|X'_1)<\eta$ and
	\item $I(X\tilde{Y}_1;\tilde{Y}_2Z|X'_2)<\eta$.
\end{enumerate}}
\end{Lemma}
\begin{proof}
Suppose for a channel which is not \one-\spoofable, there exists $P_{XY\tilde{X}Y'X'\tilde{Y}Z}\in \cP^n_{XY\tilde{X}Y'X'\tilde{Y}Z}$ which satisfies \ref{disamb:1}-\ref{disamb:6}.
Using~\ref{disamb:1} and \ref{disamb:4}, we obtain that
{
\begin{align}
2\eta&> D(P_{XYZ}||P_XP_YW_{Z|XY})+I(\tilde{X}\tilde{Y};XZ|Y)\\
&=D(P_{XYZ}||P_XP_YW_{Z|XY})+D(P_{XY\tilde{X}\tilde{Y}Z}||P_YP_{\tilde{X}\tilde{Y}|Y}P_{XZ|Y})\\
&=\sum_{x,y, \tilde{x}, \tilde{y}, z}P_{XY\tilde{X}\tilde{Y}Z}(x,y, \tilde{x}, \tilde{y}, z)\left(\log{\left\{\frac{P_{XYZ}(x, y, z)}{P_X(x)P_Y(y)W_{Z|XY}(z|x,y)}\right\}} + \log{\left\{\frac{P_{XY\tilde{X}\tilde{Y}Z}(x, y, \tilde{x}, \tilde{y}, z)}{P_Y(y)P_{\tilde{X}\tilde{Y}|Y}(\tilde{x}, \tilde{y}|y)P_{XZ|Y}(x,z|y)}\right\}}\right)\\
&=\sum_{x,y, \tilde{x}, \tilde{y}, z}P_{XY\tilde{X}\tilde{Y}Z}(x,y, \tilde{x}, \tilde{y}, z)\left(\log{\left\{\frac{P_{XYZ}(x, y, z)\times P_{XY\tilde{X}\tilde{Y}Z}(x, y, \tilde{x}, \tilde{y}, z)}{P_X(x)P_Y(y)W_{Z|XY}(z|x,y)\times P_Y(y)P_{\tilde{X}\tilde{Y}|Y}(\tilde{x}, \tilde{y}|y)P_{XZ|Y}(x,z|y)}\right\}}\right)\\
&=\sum_{x,y, \tilde{x}, \tilde{y}, z}P_{XY\tilde{X}\tilde{Y}Z}(x,y, \tilde{x}, \tilde{y}, z)\left(\log{\left\{\frac{ P_{XY\tilde{X}\tilde{Y}Z}(x, y, \tilde{x}, \tilde{y}, z)}{P_X(x)P_Y(y)W_{Z|XY}(z|x,y) P_{\tilde{X}\tilde{Y}|Y}(\tilde{x}, \tilde{y}|y)}\right\}}\right)\\
&=\sum_{x,y, \tilde{x}, \tilde{y}, z}P_{XY\tilde{X}\tilde{Y}Z}(x,y, \tilde{x}, \tilde{y}, z)\left(\log{\left\{\frac{ P_{XY\tilde{X}\tilde{Y}Z}(x, y, \tilde{x}, \tilde{y}, z)}{P_X(x) W_{Z|XY}(z|x,y) P_{Y\tilde{X}\tilde{Y}}(y, \tilde{x}, \tilde{y})}\right\}}\right)\\
& = D(P_{XY\tilde{X}\tilde{Y}Z}||P_XP_{\tilde{X}\tilde{Y}}P_{Y|\tilde{X}\tilde{Y}}W_{Z|XY})\\
&\stackrel{(a)}{\geq} D(P_{X\tilde{X}\tilde{Y}Z}||P_XP_{\tilde{X}\tilde{Y}}V^1_{Z|X\tilde{X}\tilde{Y}})\label{eq:lemma15},
\end{align}
}where $V^{(1)}_{Z|X\tilde{X}\tilde{Y}}(z|x,\tilde{x},\tilde{y}) \defineqq \sum_{y}P_{Y|\tilde{X}\tilde{Y}}(y|\tilde{x},\tilde{y}) W_{Z|XY}(z|x,y)$ and $(a)$ follows from the log sum inequality. 

Before we proceed, recall that the total variation distance between distributions $P_X$ and $Q_X$  distributed on an alphabet $\cX$, $d_{TV}(P_X, Q_X)$ is defined as
\begin{align*}
d_{TV}(P_X, Q_X) \defineqq \frac{1}{2}\sum_{x\in \cX}|P_{X}(x)-Q_{X}(x)|.
\end{align*}
The total variation distance and KL divergence between two distributions are related by the Pinsker's inequality as stated below. 
\begin{align}
D(P_X||Q_X)\geq \frac{1}{\ln(2)}\inp{d_{TV}(P_X, Q_X)}^2.\label{eq:pinsker}
\end{align}
With this, using Pinsker's inequality \eqref{eq:pinsker} on \eqref{eq:lemma15}, we obtain
\begin{align}
d_{TV}\inp{P_{X\tilde{X}\tilde{Y}Z},P_XP_{\tilde{X}\tilde{Y}}V^{(1)}_{Z|X\tilde{X}\tilde{Y}}}< c\sqrt{\eta}.\label{disambeq:1}
\end{align} where $c = \sqrt{2\ln{2}}$.
Similarly, using~\ref{disamb:2} and \ref{disamb:5}, we obtain 

\begin{align}
d_{TV}\inp{P_{\tilde{X}X\tilde{Y}Z},P_{\tilde{X}}P_{{X}\tilde{Y}}V^{(2)}_{Z|X\tilde{X}\tilde{Y}}}< c\sqrt{\eta}\label{disambeq:2}
\end{align}
where $V^{(2)}_{Z|X\tilde{X}\tilde{Y}}(z|x,\tilde{x},\tilde{y}) \defineqq \sum_{y'}P_{Y'|{X}\tilde{Y}}(y'|{x},\tilde{y}) W_{Z|XY}(z|\tilde{x},y')$. Finally, using~\ref{disamb:3} and \ref{disamb:6}, we get
\begin{align}
d_{TV}\inp{P_{X\tilde{X}\tilde{Y}Z},P_{X\tilde{X}}P_{\tilde{Y}}V^{(3)}_{Z|X\tilde{X}\tilde{Y}}}< c\sqrt{\eta}\label{disambeq:3}
\end{align}
where $V^{(3)}_{Z|X\tilde{X}\tilde{Y}}(z|x,\tilde{x},\tilde{y})\defineqq \sum_{x'}P_{X'|X\tilde{X}}(x'|x,\tilde{x})W_{Z|XY}(z|x',\tilde{y})$. 
Next, we note that
\begin{align*}
&2d_{TV}\inp{P_XP_{\tilde{X}\tilde{Y}}V^{(1)}_{Z|X\tilde{X}\tilde{Y}},P_XP_{\tilde{X}}P_{\tilde{Y}}V^{(1)}_{Z|X\tilde{X}\tilde{Y}}}\\
&=\sum_{x,\tilde{x},\tilde{y},z}\left|P_X(x)P_{\tilde{X}\tilde{Y}}(\tilde{x},\tilde{y})V^{(1)}_{Z|X\tilde{X}\tilde{Y}}(z|x,\tilde{x},\tilde{y})-P_X(x)P_{\tilde{X}}(\tilde{x})P_{\tilde{Y}}(\tilde{y})V^{(1)}_{Z|X\tilde{X}\tilde{Y}}(z|x,\tilde{x},\tilde{y})\right|\\
& = \inp{\sum_xP_X(x)}\inp{\sum_{z}V^{(1)}_{Z|X\tilde{X}\tilde{Y}}(z|x,\tilde{x},\tilde{y})}\sum_{\tilde{x},\tilde{y}}\left|P_{\tilde{X}\tilde{Y}}(\tilde{x},\tilde{y})-P_{\tilde{X}}(\tilde{x})P_{\tilde{Y}}(\tilde{y})\right|\\
& = 2d_{TV}\inp{P_{\tilde{X}\tilde{Y}},P_{\tilde{X}}P_{\tilde{Y}}} <2c\sqrt{\eta} \text{ by~\eqref{disambeq:2}}.
\end{align*}
Using this and \eqref{disambeq:1},
\begin{align}
2c\sqrt{\eta} &> d_{TV}\inp{P_{X\tilde{X}\tilde{Y}Z},P_XP_{\tilde{X}\tilde{Y}}V^{(1)}_{Z|X\tilde{X}\tilde{Y}}} + d_{TV}\inp{P_XP_{\tilde{X}\tilde{Y}}V^{(1)}_{Z|X\tilde{X}\tilde{Y}},P_XP_{\tilde{X}}P_{\tilde{Y}}V^{(1)}_{Z|X\tilde{X}\tilde{Y}}}\\
&\stackrel{\text{(a)}}{\geq} d_{TV}\inp{P_{X\tilde{X}\tilde{Y}Z},P_XP_{\tilde{X}}P_{\tilde{Y}}V^{(1)}_{Z|X\tilde{X}\tilde{Y}}},
\end{align} 
where (a) is the triangle inequality.
Thus, 
\begin{align}
d_{TV}\inp{P_{X\tilde{X}\tilde{Y}Z},P_XP_{\tilde{X}}P_{\tilde{Y}}V^{(1)}_{Z|X\tilde{X}\tilde{Y}}} < 2c\sqrt{\eta}.\label{eq:sym1}
\end{align}
Similarly, using \eqref{disambeq:1} to show that $
d_{TV}\inp{P_{\tilde{X}}P_{X}P_{\tilde{Y}}V^{(2)}_{Z|X\tilde{X}\tilde{Y}},P_{\tilde{X}}P_{{X}\tilde{Y}}V^{(2)}_{Z|X\tilde{X}\tilde{Y}}}<c\sqrt{\eta}$ and \eqref{disambeq:2}, we obtain
\begin{align}
d_{TV}\inp{P_{X\tilde{X}\tilde{Y}Z},P_XP_{\tilde{X}}P_{\tilde{Y}}V^{(2)}_{Z|X\tilde{X}\tilde{Y}}}< 2c\sqrt{\eta},\label{eq:sym2}
\end{align} and using \eqref{disambeq:1} to show that $d_{TV}\inp{P_{X}P_{\tilde{X}}P_{\tilde{Y}}V^{(3)}_{Z|X\tilde{X}\tilde{Y}},P_{X\tilde{X}}P_{\tilde{Y}}V^{(3)}_{Z|X\tilde{X}\tilde{Y}}}<c\sqrt{\eta}$ and \eqref{disambeq:3}, we obtain
\begin{align}
d_{TV}\inp{P_{X\tilde{X}\tilde{Y}Z},P_XP_{\tilde{X}}P_{\tilde{Y}}V^{(3)}_{Z|X\tilde{X}\tilde{Y}}}< 2c\sqrt{\eta}.\label{eq:sym3}
\end{align}
 We use \eqref{eq:sym1} and \eqref{eq:sym3} to write the following:
\begin{align}
d_{TV}&\inp{P_XP_{\tilde{X}}P_{\tilde{Y}}V^{(1)}_{Z|X\tilde{X}\tilde{Y}}, P_XP_{\tilde{X}}P_{\tilde{Y}}V^{(3)}_{Z|X\tilde{X}\tilde{Y}}} \nonumber\\
&\leq d_{TV}\inp{P_{X\tilde{X}\tilde{Y}Z},P_XP_{\tilde{X}}P_{\tilde{Y}}V^{(1)}_{Z|X\tilde{X}\tilde{Y}}} +d_{TV}\inp{P_{X\tilde{X}\tilde{Y}Z},P_XP_{\tilde{X}}P_{\tilde{Y}}V^{(3)}_{Z|X\tilde{X}\tilde{Y}}}\nonumber \\
&< 4c\sqrt{\eta}\label{eq:intermediate1}.
\end{align}
Similarly, using \eqref{eq:sym2} and \eqref{eq:sym3}, we may write the following:
\begin{align}
d_{TV}&\inp{P_XP_{\tilde{X}}P_{\tilde{Y}}V^{(2)}_{Z|X\tilde{X}\tilde{Y}}, P_XP_{\tilde{X}}P_{\tilde{Y}}V^{(3)}_{Z|X\tilde{X}\tilde{Y}}} < 4c\sqrt{\eta}\label{eq:intermediate2},
\end{align}
and using \eqref{eq:sym1} and \eqref{eq:sym2}, we may write the following:
\begin{align}
d_{TV}&\inp{P_XP_{\tilde{X}}P_{\tilde{Y}}V^{(1)}_{Z|X\tilde{X}\tilde{Y}}, P_XP_{\tilde{X}}P_{\tilde{Y}}V^{(2)}_{Z|X\tilde{X}\tilde{Y}}} < 4c\sqrt{\eta}\label{eq:intermediate22}.
\end{align}
Let $Q_{Y|\tilde{X}\tilde{Y}}(y|\tilde{x}, \tilde{y})\defineqq \frac{P_{Y'|{X}\tilde{Y}}(y|\tilde{x},\tilde{y})+P_{Y|\tilde{X}\tilde{Y}}(y|\tilde{x},\tilde{y})}{2}$ for all $\tilde{x}\in \cX$, $y,\tilde{y}\in \cY$ and $Q_{X'|X\tilde{X}}(x'|x,\tilde{x}) = \frac{P_{X'|X\tilde{X}}(x'|\tilde{x},{x})+P_{X'|X\tilde{X}}(x'|x,\tilde{x})}{2}$ for all $x', x, \tilde{x}\in \cX$.
With this,  
\begin{align}
&\max_{x,\tilde{x}, \tilde{y}, z}2\alpha^3\left|\sum_{y}Q_{Y|\tilde{X}\tilde{Y}}(y|\tilde{x},\tilde{y}) W_{Z|XY}(z|x,y) -  \sum_{x'}Q_{X'|X\tilde{X}}(x'|x,\tilde{x})W_{Z|XY}(z|x',\tilde{y})\right|\\
&=\max_{x,\tilde{x}, \tilde{y}, z}\alpha^3\Bigg|\inp{\sum_{y}P_{Y'|{X}\tilde{Y}}(y|\tilde{x},\tilde{y}) W_{Z|XY}(z|x,y) -  \sum_{x'}P_{X'|X\tilde{X}}(x'|\tilde{x},x)W_{Z|XY}(z|x',\tilde{y})}\\
&\hspace{0.5cm}+\inp{\sum_{y}P_{Y|\tilde{X}\tilde{Y}}(y|\tilde{x},\tilde{y}) W_{Z|XY}(z|x,y) -  \sum_{x'}P_{X'|X\tilde{X}}(x'|x,\tilde{x})W_{Z|XY}(z|x',\tilde{y})}\Bigg|\\
&\leq \max_{x,\tilde{x}, \tilde{y}, z}\alpha^3\left|\inp{\sum_{y}P_{Y'|{X}\tilde{Y}}(y|\tilde{x},\tilde{y}) W_{Z|XY}(z|x,y) -  \sum_{x'}P_{X'|X\tilde{X}}(x'|\tilde{x},x)W_{Z|XY}(z|x',\tilde{y})}\right|\\
&\hspace{0.5cm}+\max_{x,\tilde{x}, \tilde{y}, z}\alpha^3\left|\inp{\sum_{y}P_{Y|\tilde{X}\tilde{Y}}(y|\tilde{x},\tilde{y}) W_{Z|XY}(z|x,y) -  \sum_{x'}P_{X'|X\tilde{X}}(x'|x,\tilde{x})W_{Z|XY}(z|x',\tilde{y})}\right|\\
&\stackrel{(a)}{=} \max_{\tilde{x},x, \tilde{y}, z}\alpha^3\left|\inp{\sum_{y}P_{Y'|{X}\tilde{Y}}(y|x,\tilde{y}) W_{Z|XY}(z|\tilde{x},y) -  \sum_{x'}P_{X'|X\tilde{X}}(x'|x,\tilde{x})W_{Z|XY}(z|x',\tilde{y})}\right|\\
&\hspace{0.5cm}+\max_{x,\tilde{x}, \tilde{y}, z}\alpha^3\left|\inp{\sum_{y}P_{Y|\tilde{X}\tilde{Y}}(y|\tilde{x},\tilde{y}) W_{Z|XY}(z|x,y) -  \sum_{x'}P_{X'|X\tilde{X}}(x'|x,\tilde{x})W_{Z|XY}(z|x',\tilde{y})}\right|\\
&\stackrel{(b)}{\leq} \max_{\tilde{x},x, \tilde{y}, z}P_X(x)P_{\tilde{X}}(\tilde{x})P_{\tilde{Y}}(\tilde{y})\left|\inp{\sum_{y}P_{Y'|{X}\tilde{Y}}(y|x,\tilde{y}) W_{Z|XY}(z|\tilde{x},y) -  \sum_{x'}P_{X'|X\tilde{X}}(x'|x,\tilde{x})W_{Z|XY}(z|x',\tilde{y})}\right|\\
&\hspace{0.5cm}+\max_{x,\tilde{x}, \tilde{y}, z}P_X(x)P_{\tilde{X}}(\tilde{x})P_{\tilde{Y}}(\tilde{y})\left|\inp{\sum_{y}P_{Y|\tilde{X}\tilde{Y}}(y|\tilde{x},\tilde{y}) W_{Z|XY}(z|x,y) -  \sum_{x'}P_{X'|X\tilde{X}}(x'|x,\tilde{x})W_{Z|XY}(z|x',\tilde{y})}\right|\\
&\leq \max_{x,\tilde{x}, \tilde{y}, z}\left|\inp{P_X(x)P_{\tilde{X}}(\tilde{x})P_{\tilde{Y}}(\tilde{y})\sum_{y}P_{Y'|{X}\tilde{Y}}(y|x,\tilde{y}) W_{Z|XY}(z|\tilde{x},y) -  P_X(x)P_{\tilde{X}}(\tilde{x})P_{\tilde{Y}}(\tilde{y})\sum_{x'}P_{X'|X\tilde{X}}(x'|x,\tilde{x})W_{Z|XY}(z|x',\tilde{y})}\right|\\
&\hspace{0.5cm}+\max_{x,\tilde{x}, \tilde{y}, z}\left|\inp{P_X(x)P_{\tilde{X}}(\tilde{x})P_{\tilde{Y}}(\tilde{y})\sum_{y}P_{Y|\tilde{X}\tilde{Y}}(y|\tilde{x},\tilde{y}) W_{Z|XY}(z|x,y) -  P_X(x)P_{\tilde{X}}(\tilde{x})P_{\tilde{Y}}(\tilde{y})\sum_{x'}P_{X'|X\tilde{X}}(x'|x,\tilde{x})W_{Z|XY}(z|x',\tilde{y})}\right|\\
&\leq \sum_{x,\tilde{x}, \tilde{y}, z}\left|\inp{P_X(x)P_{\tilde{X}}(\tilde{x})P_{\tilde{Y}}(\tilde{y})\sum_{y}P_{Y'|{X}\tilde{Y}}(y|x,\tilde{y}) W_{Z|XY}(z|\tilde{x},y) -  P_X(x)P_{\tilde{X}}(\tilde{x})P_{\tilde{Y}}(\tilde{y})\sum_{x'}P_{X'|X\tilde{X}}(x'|x,\tilde{x})W_{Z|XY}(z|x',\tilde{y})}\right|\\
&\hspace{0.5cm}+\sum_{x,\tilde{x}, \tilde{y}, z}\left|\inp{P_X(x)P_{\tilde{X}}(\tilde{x})P_{\tilde{Y}}(\tilde{y})\sum_{y}P_{Y|\tilde{X}\tilde{Y}}(y|\tilde{x},\tilde{y}) W_{Z|XY}(z|x,y) -  P_X(x)P_{\tilde{X}}(\tilde{x})P_{\tilde{Y}}(\tilde{y})\sum_{x'}P_{X'|X\tilde{X}}(x'|x,\tilde{x})W_{Z|XY}(z|x',\tilde{y})}\right|\\
&= 2d_{TV}\inp{P_XP_{\tilde{X}}P_{\tilde{Y}}V^{(2)}_{Z|X\tilde{X}\tilde{Y}}, P_XP_{\tilde{X}}P_{\tilde{Y}}V^{(3)}_{Z|X\tilde{X}\tilde{Y}}}+ 2d_{TV}\inp{P_XP_{\tilde{X}}P_{\tilde{Y}}V^{(1)}_{Z|X\tilde{X}\tilde{Y}}, P_XP_{\tilde{X}}P_{\tilde{Y}}V^{(3)}_{Z|X\tilde{X}\tilde{Y}}}\\
&\stackrel{(c)}< 16c\sqrt{\eta},
\end{align}
where $(a)$ is obtained by exchanging $x$ and $\tilde{x}$ in the first term, $(b)$ follows by recalling that $\min_{\tilde{x}}P_{\tilde{X}}(\tilde{x}), \min_{\tilde{y}}P_{\tilde{Y}}(\tilde{y})\geq \alpha$ and $(c)$ follows from \eqref{eq:intermediate1} and \eqref{eq:intermediate2}. Thus, we have shown that 
\begin{align}\label{eq:intermediate3}
\max_{x,\tilde{x}, \tilde{y}, z}\left|\sum_{y}Q_{Y|\tilde{X}\tilde{Y}}(y|\tilde{x},\tilde{y}) W_{Z|XY}(z|x,y) -  \sum_{x'}Q_{X'|X\tilde{X}}(x'|x,\tilde{x})W_{Z|XY}(z|x',\tilde{y})\right|< \frac{8c\sqrt{\eta}}{\alpha^3}.
\end{align}
Next, we will show the following similar statement: 
\begin{align}\label{eq:intermediate4}
\max_{x,\tilde{x}, \tilde{y}, z}\left|\sum_{y}Q_{Y|\tilde{X}\tilde{Y}}(y|\tilde{x},\tilde{y})W_{Z|XY}(z|x,y)-\sum_{y'}Q_{Y|\tilde{X}\tilde{Y}}(y'|{x},\tilde{y}) W_{Z|XY}(z|\tilde{x},y')\right|< \frac{8c\sqrt{\eta}}{\alpha^3}.
\end{align}
Consider,
\begin{align*}
& 2\alpha^3\max_{x,\tilde{x}, \tilde{y}, z}\left|\sum_{y}Q_{Y|\tilde{X}\tilde{Y}}(y|\tilde{x},\tilde{y})W_{Z|XY}(z|x,y)-\sum_{y'}Q_{Y|\tilde{X}\tilde{Y}}(y'|{x},\tilde{y}) W_{Z|XY}(z|\tilde{x},y')\right|\\
&=\max_{x,\tilde{x}, \tilde{y}, z}\alpha^3\Bigg|\sum_{y}\inp{P_{Y|\tilde{X}\tilde{Y}}(y|\tilde{x},\tilde{y})+P_{Y'|{X}\tilde{Y}}(y|\tilde{x},\tilde{y})}W_{Z|XY}(z|x,y)-\sum_{y'}\inp{P_{Y'|{X}\tilde{Y}}(y'|{x},\tilde{y})+P_{Y|\tilde{X}\tilde{Y}}(y'|{x},\tilde{y}))} W_{Z|XY}(z|\tilde{x},y')\Bigg| \\
&\leq \max_{x,\tilde{x}, \tilde{y}, z}\alpha^3\left|\sum_{y}P_{Y|\tilde{X}\tilde{Y}}(y|\tilde{x},\tilde{y}) W_{Z|XY}(z|x,y)-\sum_{y'}P_{Y'|{X}\tilde{Y}}(y'|{x},\tilde{y}) W_{Z|XY}(z|\tilde{x},y')\right| \\
&\qquad + \max_{x,\tilde{x}, \tilde{y}, z}\alpha^3\left|\sum_{y}P_{Y'|{X}\tilde{Y}}(y|\tilde{x},\tilde{y}) W_{Z|XY}(z|x,y) - \sum_{y'}P_{Y|\tilde{X}\tilde{Y}}(y'|{x},\tilde{y}) W_{Z|XY}(z|\tilde{x},y')\right|\\
&=2\max_{x,\tilde{x}, \tilde{y}, z}\alpha^3\left|\sum_{y}P_{Y|\tilde{X}\tilde{Y}}(y|\tilde{x},\tilde{y}) W_{Z|XY}(z|x,y)-\sum_{y'}P_{Y'|{X}\tilde{Y}}(y'|{x},\tilde{y}) W_{Z|XY}(z|\tilde{x},y')\right|\\
&\leq2\max_{x,\tilde{x}, \tilde{y}, z}P_X(x)P_{\tilde{X}}(\tilde{x})P_{\tilde{Y}}(\tilde{y})\left|\sum_{y}P_{Y|\tilde{X}\tilde{Y}}(y|\tilde{x},\tilde{y}) W_{Z|XY}(z|x,y)-\sum_{y'}P_{Y'|{X}\tilde{Y}}(y'|{x},\tilde{y}) W_{Z|XY}(z|\tilde{x},y')\right|\\
&\leq2\sum_{x,\tilde{x}, \tilde{y}, z}\left|P_X(x)P_{\tilde{X}}(\tilde{x})P_{\tilde{Y}}(\tilde{y})\sum_{y}P_{Y|\tilde{X}\tilde{Y}}(y|\tilde{x},\tilde{y}) W_{Z|XY}(z|x,y)-P_X(x)P_{\tilde{X}}(\tilde{x})P_{\tilde{Y}}(\tilde{y})\sum_{y'}P_{Y'|{X}\tilde{Y}}(y'|{x},\tilde{y}) W_{Z|XY}(z|\tilde{x},y')\right|\\
&=4d_{TV}\inp{P_XP_{\tilde{X}}P_{\tilde{Y}}V^{(1)}_{Z|X\tilde{X}\tilde{Y}}, P_XP_{\tilde{X}}P_{\tilde{Y}}V^{(2)}_{Z|X\tilde{X}\tilde{Y}}}\\
&\stackrel{\text{(a)}}{<} 16c\sqrt{\eta},
\end{align*}
where (a) follows from \eqref{eq:intermediate22}. With this, we have shown \eqref{eq:intermediate4}. 
Using \eqref{eq:intermediate3} and \eqref{eq:intermediate4}, we have shown that for the given channel there exist conditional distributions $Q_{Y|\tilde{X}\tilde{Y}}$ and $Q_{X'|X\tilde{X}}$ such that
\begin{align}
&\max\Bigg\{\max_{x,\tilde{x}, \tilde{y}, z}\left|\sum_{y}Q_{Y|\tilde{X}\tilde{Y}}(y|\tilde{x},\tilde{y}) W_{Z|XY}(z|x,y) -  \sum_{x'}Q_{X'|X\tilde{X}}(x'|x,\tilde{x})W_{Z|XY}(z|x',\tilde{y})\right|,\nonumber\\
&\hspace{1cm}\max_{x,\tilde{x}, \tilde{y}, z}\Bigg|\sum_{y}Q_{Y|\tilde{X}\tilde{Y}}(y|\tilde{x},\tilde{y})W_{Z|XY}(z|x,y)-\sum_{y'}Q_{Y|\tilde{X}\tilde{Y}}(y'|{x},\tilde{y}) W_{Z|XY}(z|\tilde{x},y')\Bigg|\Bigg\}< \frac{8c\sqrt{\eta}}{\alpha^3}. \label{eq:intermediate5}
\end{align}
Suppose the channel is not \one-\spoofable (i.e. \eqref{eq:spoof1} does not hold), then for every pair $\inp{Q_{Y|\tilde{X} \tilde{Y}},Q_{X'|X\tilde{X}}}$, there exists \\$\zeta\inp{Q_{Y|\tilde{X} \tilde{Y}},Q_{X'|X\tilde{X}}}>0$ such that  the following holds:
\begin{align}
&\max\Bigg\{\max_{x,\tilde{x}, \tilde{y}, z}\left|\sum_{y}Q_{Y|\tilde{X}\tilde{Y}}(y|\tilde{x},\tilde{y}) W_{Z|XY}(z|x,y) -  \sum_{x'}Q_{X'|X\tilde{X}}(x'|x,\tilde{x})W_{Z|XY}(z|x',\tilde{y})\right|,\\
&\hspace{1cm}\max_{x,\tilde{x}, \tilde{y}, z}\Bigg|\sum_{y}Q_{Y|\tilde{X}\tilde{Y}}(y|\tilde{x},\tilde{y})W_{Z|XY}(z|x,y)-\sum_{y'}Q_{Y|\tilde{X}\tilde{Y}}(y'|{x},\tilde{y}) W_{Z|XY}(z|\tilde{x},y')\Bigg|\Bigg\}> \zeta\inp{Q_{Y|\tilde{X} \tilde{Y}},Q_{X'|X\tilde{X}}}. 
\end{align}This contradicts \eqref{eq:intermediate5} for $\eta$ small enough such that $\zeta\inp{Q_{Y|\tilde{X} \tilde{Y}},Q_{X'|X\tilde{X}}}>8c\sqrt{\eta}/\alpha^3$. This completes the proof for a channel which is not \one \spoofable. The proof for a channel which is not \two \spoofable is along the similar lines. It can be obtained by interchanging the roles of users \one and \two and making the following replacements in the above proof: $ \tilde{X}\rightarrow \tilde{Y}_1, \, {Y}'\rightarrow X'_1, \, X'\rightarrow X'_2, \text{ and }\tilde{Y}\rightarrow \tilde{Y}_2$.

\end{proof}


\section{Proof of Theorem~\ref{thm:inner_bd}}\label{sec:inner_bd_proof}
\begin{proof}
\blue{Fix some $P_{\one}$ and $P_{\two}$ satisfying $\min_{x\in \cX}P_{\one}(x)\geq\alpha$ and $\min_{y\in \cY}P_{\two}(y)\geq\alpha$ respectively for some $\alpha>0$ and $\eta>0$ which is a function of $\alpha$ and the channel (See Lemma~\ref{lemma:disambiguity}). We also fix  $\eta/3>\epsilon>0$ and  $n>n_0(\epsilon)$.
For a rate pair $(R_{\one}, R_{\two})\in \cR_{2}(P_{\one}, P_{\two})$} given by \eqref{eq:inner_bd_2} with $R_{\one}, R_{\two}>\epsilon$, consider the codebook given by Lemma~\ref{lemma:codebook}. \\
\noindent {\em Encoding.} 
 Let $N_{\one}= 2^{nR_{\one}}, \, N_{\two}=2^{nR_{\two}}, \, \cM_{\one} = \inb{1, \ldots, N_{\one}}$ and $\cM_{\two} =\inb{1, \ldots, N_{\two}}$. For $\mo \in \cM_{\one}$, $f_{\one}(\mo) = \vecx_{\mo}$ and for $\mt \in \cM_{\two}$, $f_{\two}(\mt) = \vecy_{\mt}$.\\
{\em Decoding.} 
Let $\cD_{\eta}$ be the set of joint distributions defined as
\begin{align*}
\blue{\cD_{\eta}} \defineqq  \inb{P_{XYZ}\in \cP^n_{\cX\times \cY \times \cZ}:\, D\inp{P_{XYZ}||P_XP_YW}\leq\eta }.
\end{align*}
For a given codebook, the parameter $\eta$ and the received channel output sequence $\vecz$,   decoding happens in five steps. In the first step, we populate sets $A_1$ and $B_1$ containing candidate messages for user \one and \two respectively. In steps $2$-$5$, we prune these sets by sequentially removing the candidates which do not satisfy certain conditions. 
\begin{description}
\item  [{\em Step 1}:]\label{step:1} Let $A_1 = \{m_{\one}\in\mathcal{M}_{\one}:\inp{f_{\one}(m_{\one}), \bar{\vecy}, \vecz} \in T^{n}_{XYZ}$ for some $\bar{\vecy}\in \cY^n$ such that $P_{XYZ}\in \cD_{\eta}$\} and\\
$B_1 = \{m_{\two}\in\mathcal{M}_{\two}:\inp{\bar{\vecx}, f_{\two}(m_{\two}), \vecz} \in T^{n}_{XYZ}$ for some $\bar{\vecx}\in \cX^n$ such that $P_{XYZ}\in \cD_{\eta}$\}.

\item [{\em Step 2}:] \label{step:2} Let $C_1 = \{\mo\in A_1:$ {for every } $\tilde{m}_{\two 1},\, \tilde{m}_{\two 2}\in B_1$ and $\bar{\vecy}\in \cY^n$ with $\inp{f_{\one}(\mo), \bar{\vecy},f_{\two}(\tilde{m}_{\two 1}), f_{\two}(\tilde{m}_{\two 2}),\vecz}\in T^n_{XY\tilde{Y}_1\tilde{Y}_2Z}$
 that satisfies $P_{XYZ}\in\cD_{\eta}$, we have $I(\tilde{Y}_1\tilde{Y}_2;XZ|Y)> \eta\}$. Let $A_2 = A_1\setminus C_1$.

\item [{\em Step 3}:]\label{step:3} Let $C_2 = \{\mt\in B_1:$ {for every } $\tilde{m}_{\one 1},\, \tilde{m}_{\one 2}\in A_2$ and $\bar{\vecx}\in \cX^n$ with $\inp{\bar{\vecx}, f_{\two}(\mt),  f_{\one}(\tilde{m}_{\one 1}),f_{\one}(\tilde{m}_{\one 2}),\vecz}\in T^n_{XY\tilde{X}_1\tilde{X}_2Z}$ that satisfies $P_{XYZ}\in \cD_{\eta}$, we have $I(\tilde{X}_1\tilde{X}_2;YZ|X)> \eta\}$. Let $B_2 = B_1\setminus C_2$.

\item [{\em Step 4}:]\label{step:4} Let $C_3 = \{\mo\in A_2:$ {for every } $(\tilde{m}_{\one},\, \tilde{m}_{\two})\in A_2\times B_2, \, \tilde{m}_{\one}\neq m_{\one}$ and $\bar{\vecy}\in \cY^n$ with $\inp{f_{\one}(\mo), \bar{\vecy}, f_{\one}(\tilde{m}_{\one}), f_{\two}(\tilde{m}_{\two}), \vecz}\in T^n_{XY\tilde{X}\tilde{Y}Z}$ such that $P_{XYZ}\in \cD_{\eta}$, we have $I(\tilde{X}\tilde{Y};XZ|Y)>\eta\}$. Let $A_3 = A_2\setminus C_3$.

\item [{\em Step 5}:]\label{step:5} Let $C_4 = \{\mt\in B_2:$ {for every } $(\tilde{m}_{\one},\, \tilde{m}_{\two})\in A_3\times B_2, \, \tilde{m}_{\two}\neq m_{\two}$ and $\bar{\vecx}\in \cX^n$ with $\inp{\bar{\vecx},f_{\two}(\mt), f_{\one}(\tilde{m}_{\one}), f_{\two}(\tilde{m}_{\two}), \vecz}\in T^n_{XY\tilde{X}\tilde{Y}Z}$ such that $P_{XYZ}\in \cD_{\eta}$, we have $I(\tilde{X}\tilde{Y};YZ|X)>\eta\}$. Let $B_3 = B_2\setminus C_4$.

\end{description}
After steps 1-5, the decoded output is as follows.
\begin{align*}
\phi(\vecz) = \begin{cases}(\mo,\mt) &\text{ if }A_3\times B_3 = \{(\mo,\mt)\},\\ \oneb &\text{ if }|A_3| = 0, \, |B_3| \neq 0,\\\twob &\text{ if }|A_3| \neq 0, \, |B_3| = 0\text{ and}\\(1,1)&\text{ otherwise.}
\end{cases}
\end{align*}
Similar to the decoder for achievability of Theorem~\ref{thm:main_result}, the last of the above cases ({\em i.e.} $\phi(\vecz) = (1,1)$) occurs when either of the following two events occur: (i) $|A_3|=|B_3|=0$ and (ii) $|A_3|, |B_3|\geq 1$ and $|A_3|+|B_3|\geq 3$. The first event will be shown to be an atypical event and hence will occur with vanishing probability. For small enough choices of $\eta>0$, Lemma~\ref{lemma:disambiguity} implies that the second event cannot happen in a non-spoofable channel. We consider such an $\eta$.

\begin{figure}[!h]
\begin{centering}
\begin{tikzpicture}[scale=1.2]
  \node (A1) at (-1.5, 0) [rectangle, draw] {$P_{e,\maltwo}(\vecy)$};
  \node (C) at (-1.5, -2) [rounded corners=3pt, draw] {Union bound};
  \node[xshift = 1.4cm, right of = C] {eq.~\eqref{eq:appendix_union_1} and \eqref{eq:appendix_union_2}};
  \node (D) at (1.5, -2) {};

  \node (E) at (-3, -4)  {\small $\substack{\text{\small  small}\\\text{ (atypical event)}}$};
\node (F) at (-1.5, -4.5)  {$\substack{ \text{\small small}\\\text{ (atypical event)}}$};
 \node (G) at (0, -4) [rounded corners=3pt, draw] {Union bound};
  \node[xshift = 0.8cm, right of = G] {eq.~\eqref{err:upperbound}};
  \node (G1) at (-1.5, -6) [rounded corners=3pt, draw] {Union bound}; 
  \node (G3) at (1.5, -6) [rounded corners=3pt, draw] {Union bound};
 
\node (H1) at (-5, -8.7) {};
\node (H2) at (5.7, -8.7)  {};

  \node (G11) at (-2.7, -9.5) [rounded corners=3pt, draw] {$\substack{\text{Case analysis}\\ \text{of \eqref{eq:upperbound2}}}$};
  \node (G111) at (-4.7, -11)  {\small small};
  \node (G112) at (-3.5, -11.5)  {$\substack{\text{\small  small } \\ \text{(by $R_{\one}$}\\\text{  bound in \eqref{eq:rbound1})}}$};
 \node (G113) at (-2.1, -11.5)  {$\substack{\text{\small  small } \\ \text{(by $R_{\two}$ bound}\\\text{ in \eqref{eq:rbound2})}}$}; 
  \node (G114) at (-0.6, -11.5)  {$\substack{\text{\small  small } \\ \text{(by $R_{\one}+ R_{\two}$ }\\\text{bound in \eqref{eq:rbound3})}}$};

  \node (G12) at (-0.9, -7.5)  {$\substack{\text{\small  small } \\ \text{(by codebook}\\\text{ property\eqref{codebook:2b})}}$};

 \node (G31) at  (0.3, -8.3) {$\substack{\text{\small  small } \\ \text{(by codebook}\\\text{ property\eqref{codebook:4})}}$}; 
  \node (G32) at (3, -9.2) [rounded corners=3pt, draw]{$\substack{\text{Case analysis}\\ \text{of \eqref{eq:upperbound3}}}$};
  \node (G321) at (-4.7+5.7, -11.3)  {\small small};
  \node (G322) at (-3.5+5.7, -11.3)  {$\substack{\text{\small  small } \\ \text{(by $R_{\two}$}\\\text{  bound in \eqref{eq:Rbound4})}}$};
 \node (G323) at (-2.1+5.7, -11.3)  {$\substack{\text{\small  small } \\ \text{(by $R_{\two}$ bound}\\\text{ in \eqref{eq:Rbound5})}}$}; 
  \node (G324) at (-0.6+5.7, -11.3)  {$\substack{\text{\small  small } \\ \text{(by $2R_{\two}$ }\\\text{bound in \eqref{eq:Rbound6})}}$};

 	\draw[->] (G) -- (G1) node[midway, right]{$\scriptstyle P_{\cE_{\mo,1}}(\vecy)$};
 	\draw[->] (G) -- (G3) node[midway, right]{$\scriptstyle P_{\cE_{\mo,2}}(\vecy)$};
 	\draw[->] (G1) -- (G12) node[midway, xshift = 0cm, yshift = 0.1cm,  right]{$ \stackrel{\text{\eqref{chapter:4:cond} }}{\text{\scriptsize holds}}$};
 	\draw[->] (G1) -- (G11) node[midway, left]{$\scriptstyle \eqref{eq:1} \text{ holds}$};
 	\draw[ dashed] (H1) -- (H2) ;

 	\draw[->] (G11) -- (G111) node[midway, left]{\small 1};
 	\draw[->] (G11) -- (G112) node[midway, left]{\small 2};
 	\draw[->] (G11) -- (G113) node[midway, left]{\small 3};
 	\draw[->] (G11) -- (G114) node[midway, left]{\small 4};

 	\draw[->] (G32) -- (G321) node[midway, left]{\small 1};
 	\draw[->] (G32) -- (G322) node[midway, left]{\small 2};
 	\draw[->] (G32) -- (G323) node[midway, left]{\small 3};
 	\draw[->] (G32) -- (G324) node[midway, left]{\small 4};

 	 	\draw[->] (G3) -- (G31) node[xshift = 0.2cm, yshift = 0.3cm,midway, left]{$ \stackrel{\text{\eqref{eq:2} does }}{\text{\scriptsize not hold}}$};
 	\draw[->] (G3) -- (G32) node[midway, xshift = 0.2cm, yshift = -0.2cm,  right]{$\scriptstyle \eqref{eq:2} \text{ holds}$};

 	\draw[->] (C) -- (E) node[midway, right]{$\scriptstyle P_1(\vecy)$};
 	\draw[->] (C) -- (F) node[yshift = -0.2cm, midway, right]{$\scriptstyle P_2(\vecy)$};
 	\draw[->] (C) -- (G) node[midway, right]{$\scriptstyle \scriptstyle P_3(\vecy)$};
  \draw[->] (A1) -- (C) node[midway, right]{};
\node at (7, -2.5)  {
\begin{tabular}{p{1.5cm}|p{6.5cm}} 
$P_{e,\maltwo}(\vecy)$& the average probability of error when malicious user \two sends $\vecy$\\
\hline
$P_1(\vecy)$ & the average probability that channel inputs are atypical\\
\hline	
$P_2(\vecy)$ & the average probability that the channel output is atypical \\
\hline	
$P_3(\vecy)$ & the average probability of error when channel inputs and output are typical\\
\hline	
$P_{\cE_{\mo,1}}(\vecy)$& condition \ref{check:2} in Definition~\ref{D_eta_z} does not hold\\
\hline	
$P_{\cE_{\mo,2}}(\vecy)$& condition ~\ref{check:3}  in Definition~\ref{D_eta_z} does not hold\\
\end{tabular}
};
\end{tikzpicture}
\caption{Flowchart depicting the flow of analysis of $P_{e,\maltwo}(\vecy)$, the average probability of error when user \two is malicious and sends $\vecy$. The details of the case analysis of \eqref{eq:upperbound3} are skipped in the proof as it is similar to the case analysis of \eqref{eq:upperbound2}. This flowchart is same as the Flowchart~\ref{fig:flowchart1_authcom} for the proof of feasibility till the dashed line.} \label{fig:flowchart1_authcom_appendix}
\end{centering}
\end{figure}

As noted previously (see \eqref{honest_error_ub}) $P_{e,\na}\leq P_{e,\malone}+ P_{e,\maltwo}$. Thus, it is sufficient to analyze the case when one of the users is adversarial.
We consider the case when user \two is malicious.
We will analyse $P_{e,\maltwo}$. Suppose a malicious user \two sends $\vecy$. Let $P_{e,\maltwo}(\vecy)$ denote the probability of error when user \two is malicious and sends $\vecy$. That is, for $\cE^{\one}_{\mo} \defineqq \inb{\vecz:\phi_{\one}(\vecz)\notin\{\mo,\twob\}}$,  
\begin{align*}
P_{e,\maltwo}(\vecy) = \frac{1}{N_{\one}}\sum_{m_{\one}\in \mathcal{M}_{\one}}W^n\inp{\cE^{\one}_{\mo}|f_{\one}(\mo), \vecy}
\end{align*}  and $$P_{e,\maltwo} = \max_{\vecy}P_{e,\maltwo}(\vecy).$$ We will show that $P_{e,\maltwo}(\vecy)$ is small for each $\vecy\in \cY^n$. The analysis follows the flowchart given in Figure~\ref{fig:flowchart1_authcom}.
For some $\epsilon$ satisfying $0<\epsilon<\eta/3$, let
\begin{align*}
\cH&\defineqq \inb{m_{\one}: (\vecx_{m_{\one}}, \vecy)\in T^n_{XY}\text{ such that } I(X;Y)> \epsilon}.
\end{align*}
Then,
\begin{align}
P_{e,\maltwo}(\vecy) &\leq \frac{1}{N_{\one}}|\cH| + \frac{1}{N_{\one}}\sum_{\mo\in \cH^c}W^n(\cE^{\one}_{\mo}|f_{\one}(\mo),\vecy)=: P_1(\vecy) +\tilde{P}(\vecy)\label{eq:appendix_union_1}.
\end{align} Here, $\cH^c$ denotes the complement of $\cH$. 
The first term on the RHS, 
\begin{align*}
P_1(\vecy)\leq |\cP^n_{\cX\times \cY}|\times\frac{|\inb{m_{\one}: (\vecx_{m_{\one}}, \vecy)\in T_{XY}^n, \, I(X;Y)> \epsilon}|}{N_{\one}}
\end{align*}
which goes to zero as $n\rightarrow \infty$ by ~\eqref{codebook:1} and noting that there are only polynomially many types. 
In order to analyze the second term $\tilde{P}(\vecy)$, for $\vecy'\in \cY^n$, let $\cE_1(\vecy')$ be defined as
\begin{align*}
\cE_1(\vecy') = \{\vecz: (\vecx_{\mo}, \vecy', \vecz)\in T^n_{XYZ} \text{ for some }P_{XYZ}\in \blue{\cD_{\eta}}\}.
\end{align*}
Then,
\begin{align}
\tilde{P}(\vecy)&=   {\frac{1}{N_{\one}}\sum_{\mo\in \cH^c}W^n(\cE^{\one}_{\mo}|f_{\one}(\mo),\vecy)}\nonumber\\
&=\frac{1}{N_{\one}}\sum_{\mo\in \cH^c}W^n(\inp{\cE_1(\vecy)^c\cap\cE^{\one}_{\mo}}\cup\inp{\cE_1(\vecy)\cap\cE^{\one}_{\mo}}|f_{\one}(\mo),\vecy) \nonumber\\
&\leq \frac{1}{N_{\one}}\sum_{\mo\in \cH^c}W^n(\inp{\cE_1(\vecy)^c}|f_{\one}(\mo),\vecy) + W^n(\inp{\cE_1(\vecy)\cap\cE^{\one}_{\mo}}|f_{\one}(\mo),\vecy) \nonumber\\
&=\frac{1}{N_{\one}}\sum_{m_{\one}\in \cH^c}\inp{\sum_{P_{XYZ}\in \cD_{\eta}^c}\sum_{\vecz\in T^n_{Z|XY}(\vecx_{m_{\one}},\vecy)}W^n(\vecz|\vecx_{m_{\one}}, \vecy)} + \frac{1}{N_{\one}}\sum_{m_{\one}\in \cH^c}\inp{\sum_{P_{XYZ}\in \cD_{\eta}}\sum_{\vecz\in T^n_{Z|XY}(\vecx_{m_{\one}},\vecy)\cap\cE^{\one}_{\mo}}W^n(\vecz|\vecx_{m_{\one}}, \vecy)}\nonumber\\
&=:P_{2}(\vecy) +P_{3}(\vecy)\label{eq:appendix_union_2}.
\end{align} 
To analyze $P_{2}(\vecy)$, consider any  $\mo \in \cH^c$,
\begin{align*}
\sum_{\stackrel{P_{XYZ}\in \cD_{\eta}^c,}{I(X;Y)\leq \epsilon}}\sum_{\vecz\in T^n_{Z|XY}(\vecx_{m_{\one}},\vecy)}W^n(\vecz|\vecx_{m_{\one}}, \vecy) &\leq \sum_{\stackrel{P_{XYZ}\in \cD_{\eta}^c,}{I(X;Y)\leq \epsilon}}\exp{\inp{-nD(P_{XYZ}||P_{XY}W)}}\\
& = \sum_{\stackrel{P_{XYZ}\in \cD_{\eta}^c,}{I(X;Y)\leq \epsilon}}\exp{\inp{-n\inp{D(P_{XYZ}||P_{X}P_{Y}W) -I(X;Y)}}}\\
& = \sum_{\stackrel{P_{XYZ}\in \cD_{\eta}^c,}{I(X;Y)\leq \epsilon}}\exp{\inp{-n\inp{\eta-\epsilon}}}\\
& \leq |\cD_{\eta}^c|\exp{\inp{-n\inp{\eta-\epsilon}}}. 
\end{align*}
Thus, 
\begin{align*}
P_{2}(\vecy)&\leq \frac{|\cH^c|}{N_{\one}}|\cD_{\eta}^c|\exp{\inp{-n\inp{\eta-\epsilon}}}\\
&\rightarrow 0 \text{ as }\epsilon<\eta/3 \text{ and $|\cD_{\eta}^c|\leq (n+1)^{|\cX||\cY||\cZ|}$,}
\end{align*}

We are left to analyze 
\begin{align*}
P_{3}(\vecy) = \frac{1}{N_{\one}}\sum_{m_{\one}\in \cH^c}\inp{\sum_{P_{XYZ}\in \cD_{\eta}}\sum_{\vecz\in T^n_{Z|XY}(\vecx_{m_{\one}},\vecy)\cap\cE^{\one}_{\mo}}W^n(\vecz|\vecx_{m_{\one}}, \vecy)}.
\end{align*} 
In order to proceed, we define the following sets,
\begin{align*}
\cP_1^{\eta}& =  \{P_{X\tilde{X}\tilde{Y}YZ}\in \cP^n_{\cX\times\cX\times\cY\times\cY\times\cZ}: P_{XYZ}\in \cD_{\eta},\, P_{\tilde{X}Y'Z}\in \cD_{\eta} \text{ for some }P_{Y'|\tilde{X}Z},
 P_{X'\tilde{Y}Z}\in \cD_{\eta} \\&\qquad  \,\text{ for some }P_{X'|\tilde{Y}Z}, P_{X}=P_{\tilde{X}}=P_{\one}, P_{\tilde{Y}} = P_{\two}, \, I(\tilde{Y};X)\leq 2\eta, \, I(\tilde{Y};\tilde{X})\leq 2\eta\text{ and }I(\tilde{X}\tilde{Y};XZ|Y)\geq\eta\}\\
\cP_2^{\eta}& = \{P_{X\tilde{Y}_1\tilde{Y}_2YZ}\in \cP^n_{\cX\times\cY\times\cY\times\cY\times\cZ}: P_{XYZ}\in \cD_{\eta}, P_{X'_1\tilde{Y}_1Z}\in \cD_{\eta} \text{ for some }P_{X'_1|\tilde{Y}_1Z}, P_{X'_2\tilde{Y}_2Z}\in \cD_{\eta}\\
&\qquad \, \text{ for some }P_{X'_2|\tilde{Y}_2Z}, P_{X}=P_{\one}, P_{\tilde{Y}_1}=P_{\tilde{Y}_2} = P_{\two}\text{ and }I(\tilde{Y}_1\tilde{Y}_2;XZ|Y)\geq\eta\}.
\end{align*}
For $P_{X\tilde{X}\tilde{Y}YZ}\in \cP_1^{\eta}$ and $P_{X\tilde{Y}_1\tilde{Y}_2YZ}\in \cP_2^{\eta} $, let
\begin{align*}
\cE_{\mo,1}(P_{X\tilde{X}\tilde{Y}YZ}) & = \big\{\vecz: \exists(\tilde{m}_{\one},\, \tilde{m}_{\two})\in \cM_{\one}\times \cM_{\two}, \, \tilde{m}_{\one}\neq m_{\one}, \,  \inp{\vecx_{\mo},\vecx_{\tilde{m}_{\one}}, \vecy,   \vecy_{\tilde{m}_{\two}}, \vecz}\in T^n_{X\tilde{X}Y\tilde{Y}Z} \big\} \text{ and }\\
\cE_{\mo,2}(P_{X\tilde{Y}_1\tilde{Y}_2YZ}) & = \big\{\vecz: \exists \tilde{m}_{\two 1},\, \tilde{m}_{\two 2}\in \cM_{\two}, \,\inp{\vecx_{\mo},  \vecy_{\tilde{m}_{\two 1}}, \vecy_{\tilde{m}_{\two 2}},\vecy,\vecz}\in T^n_{X\tilde{Y}_1\tilde{Y}_2YZ}\big\}.
\end{align*}
With these definitions, 
\begin{align}
\inb{\vecz\in T^n_{Z|XY}(\vecx_{\mo}, \vecy)\cap\cE^{\one}_{\mo}\text{ where }P_{XYZ}\in \cD_{\eta}}&\stackrel{(a)}{\subseteq} \{\vecz:\mo\in A_1\cap C_1\}\cup\{\vecz: \mo\in A_2\cap C_3\}\\
&\stackrel{(b)}{\subseteq}\inp{\cup_{P_{X\tilde{Y}_1\tilde{Y}_2YZ}\in \cP_2^{\eta}}\cE_{\mo,2}(P_{X\tilde{Y}_1\tilde{Y}_2YZ})}\cup\inp{\cup_{P_{X\tilde{X}\tilde{Y}YZ}\in \cP_1^{\eta}}\cE_{\mo,1}(P_{X\tilde{X}\tilde{Y}YZ})}\label{eq:assym_up} 
\end{align} where $(a)$ follows from the decoder definition (Note that $\vecz\in T^n_{Z|XY}(\vecx_{\mo}, \vecy)$ for $P_{XYZ}\in \cD_{\eta}$ implies that $\vecz\in A_1$ and $\vecz\in \cE^{\one}_{\mo}$ further implies that it was eliminated in {\em \textbf{Step 2}} or {\em \textbf{Step 4}}, i.e., $\vecz\in A_1\cap C_1$ or $\vecz\in A_2\cap C_3$). To see $(b)$, first notice that $\{\vecz:\mo\in A_1\cap C_1\}\subseteq \inp{\cup_{P_{X\tilde{Y}_1\tilde{Y}_2YZ}\in \cP_2^{\eta}}\cE_{\mo,2}(P_{X\tilde{Y}_1\tilde{Y}_2YZ})}$. 
Further, for $\mo\in C_3$, we only consider pairs  $(\tilde{m}_{\one}, \tilde{m}_{\two})$ belonging to $A_2\times B_2$ defined in {\em \textbf{Step 2}} and {\em \textbf{Step 3}}. Thus, for  $(\vecx_{\mo}, \vecy, \vecx_{\tilde{m}_{\one}}, \vecy_{\tilde{m}_{\two}}, \vecz)\in T^n_{XY\tilde{X}\tilde{Y}Z}$, we have $P_{XYZ}\in \cD_{\eta}$, $P_{\tilde{X}Y'Z}\in \cD_{\eta}$ for some $P_{Y'|\tilde{X}Z}$, $P_{X'\tilde{Y}Z}\in \cD_{\eta}$ for some $P_{X'|\tilde{Y}Z}$ and $I(X\tilde{X};\tilde{Y}Z|X')\leq \eta$ (see  {\em \textbf{Step 3}}).
For such distributions, the following lemma implies that $\{\vecz: \mo\in A_2\cap C_3\}\subseteq \inp{\cup_{P_{X\tilde{X}\tilde{Y}YZ}\in \cP_1^{\eta}}\cE_{\mo,1}(P_{X\tilde{X}\tilde{Y}YZ})}$.
\begin{lemma}\label{lemma:indep}
For a distribution $P_{XY\tilde{X}Y'X'\tilde{Y}Z}\in \cP^n_{\cX\times\cY\times\cX\times\cY\times\cX\times\cY\times\cZ}$ satisfying
\begin{enumerate}[label=(\Alph*)]
	\item $P_{XYZ}\in \blue{\cD_{\eta}}$
	\item $P_{\tilde{X}Y'Z}\in \blue{\cD_{\eta}}$
	\item $P_{X'\tilde{Y}Z}\in \blue{\cD_{\eta}}$, and
	\item $I(X\tilde{X};\tilde{Y}Z|X')<\eta$,
\end{enumerate}  
we have  $I(\tilde{Y};X)\leq 2\eta$ and $I(\tilde{Y};\tilde{X})\leq 2\eta$.
\end{lemma}
\noindent The proof of this lemma is given after the current proof. It follows by adding $(C)$ and $(D)$ and an application of log-sum inequality. 
Continuing the analysis of $P_{3}(\vecy)$, we use \eqref{eq:assym_up} to write
\begin{align}
P_{3}(\vecy) =& \frac{1}{N_{\one}}\sum_{m_{\one}\in \cH^c}\inp{\sum_{P_{XYZ}\in \cD_{\eta}}\sum_{\vecz\in T^n_{Z|XY}(\vecx_{m_{\one}},\vecy)\cap\cE^{\one}_{\mo}}W^n(\vecz|\vecx_{m_{\one}}, \vecy)}\nonumber\\
\leq&\frac{1}{N_{\one}}\sum_{\mo\in \cM_{\one}}\sum_{P_{X\tilde{X}\tilde{Y}YZ}\in \cP_1^{\eta}}  W^n\inp{\cE_{\mo,1}(P_{X\tilde{X}\tilde{Y}YZ})|\vecx_{\mo}, \vecy} \nonumber\\
&\qquad \qquad+ \frac{1}{N_{\one}}\sum_{\mo\in \cM_{\one}}\sum_{P_{X\tilde{Y}_1\tilde{Y}_2YZ}\in \cP_2^{\eta}}  W^n\inp{\cE_{\mo,2}(P_{X\tilde{Y}_1\tilde{Y}_2YZ})|\vecx_{\mo}, \vecy}\\
&=:P_{\cE_{\mo,1}}(\vecy) + P_{\cE_{\mo,2}}(\vecy).  \label{err:upperbound}
\end{align}
We see that $|\cP_1^{\eta}|$ and $|\cP_2^{\eta}|$ grow at most polynomial in $n$. So, it will  suffice to uniformly upper bound $W^n\inp{\cE_{\mo,1}(P_{X\tilde{X}\tilde{Y}YZ})|\vecx_{\mo}, \vecy}$ and $W^n\inp{\cE_{\mo,2}(P_{X\tilde{Y}_1\tilde{Y}_2YZ})|\vecx_{\mo}, \vecy}$ by a term exponentially decreasing in $n$ for all $P_{X\tilde{X}\tilde{Y}YZ}\in \cP_1^{\eta}$, $P_{X\tilde{Y}_1\tilde{Y}_2YZ}\in \cP_2^{\eta}$ and $\mo\in \cM_{\one}$. 
We start the analysis of $P_{\cE_{\mo,1}}(\vecy) $ by upper bounding $W^n\inp{\cE_{\mo,1}(P_{X\tilde{X}\tilde{Y}YZ})|\vecx_{\mo}, \vecy}$. By using~\eqref{codebook:2b}, we see that for $P_{X\tilde{X}\tilde{Y}YZ}\in \cP_1^{\eta}$ such that
\begin{align}
I\inp{X;\tilde{X}\tilde{Y}Y}> |R_{\one}- I(\tilde{X};\tilde{Y}Y)|^{+}+|R_{\two}-I(\tilde{Y};Y)|^{+}+\epsilon,\label{chapter:4:cond}
\end{align}we have,
\begin{align*}
\frac{\left|\inb{\mo:(\vecx_{\mo}, \vecx_{\tilde{m}_{\one}}, \vecy_{\mt}, \vecy)\in T^n_{X\tilde{X}\tilde{Y}Y} \text{ for some }\tilde{m}_{\one}\neq \mo\text{ and some }\mt} \right|}{N_{\one}} \leq \exp\inb{-n\epsilon/2}.
\end{align*}
So, for all $P_{X\tilde{X}\tilde{Y}YZ}\in \cP_1^{\eta}$ satisfying \eqref{chapter:4:cond},
\begin{align*}
&\frac{1}{N_{\one}}\sum_{\mo\in \cH^c} W^n\inp{\cE_{\mo,1}(P_{X\tilde{X}\tilde{Y}YZ})|\vecx_{\mo}, \vecy} \\
& = \frac{1}{N_{\one}}\sum_{\substack{\mo:(\vecx_{\mo}, \vecx_{\tilde{m}_{\one}}, \vecy_{\mt}, \vecy)\in T^n_{X\tilde{X}\tilde{Y}Y},\\ \tilde{m}_{\one}\in \cM_{\one},\tilde{m}_{\one}\neq \mo,\mt\in\cM_{\two}}}\,\, \sum_{\vecz\in T^{n}_{Z|X\tilde{X}Y\tilde{Y}}(\vecx_{\mo},\vecx_{\tilde{m}_{\one}},\vecy,\vecy_{\tilde{m}_{\two}})}W^n\inp{\vecz|\vecx_{\mo}, \vecy}\\
&\leq \exp\inb{-n\epsilon/2}.
\end{align*}
Thus, it is sufficient to consider distributions $P_{X\tilde{X}\tilde{Y}YZ}\in \cP_1^{\eta}$ for which 
\begin{align}
I\inp{X;\tilde{X}\tilde{Y}Y}\leq |R_{\one}- I(\tilde{X};\tilde{Y}Y)|^{+}+|R_{\two}-I(\tilde{Y};Y)|^{+}+\epsilon.\label{eq:1}
\end{align}
For $P_{X\tilde{X}\tilde{Y}YZ}\in \cP_1^{\eta}$ satisfying~\eqref{eq:1},
\begin{align}
&\sum_{\vecz\in \cE_{\mo,1}(P_{X\tilde{X}\tilde{Y}YZ})}W^n(\vecz|\vecx_{\mo}, \vecy)\nonumber\\
&\qquad=\sum_{\substack{\tilde{m}_{\one}, \tilde{m}_{\two}:\tilde{m}_{\one}\neq \mo\\(\vecx_{\mo}, \vecx_{\tilde{m}_{\one}}, \vecy_{\tilde{m}_{\two}},\vecy)\in T^{n}_{X\tilde{X}\tilde{Y}Y}}}\sum_{\vecz:(\vecx_{\mo}, \vecx_{\tilde{m}_{\one}}, \vecy_{\tilde{m}_{\two}},\vecy, \vecz)\in T^{n}_{X\tilde{X}\tilde{Y}YZ}}W^n(\vecz|\vecx_{\mo}, \vecy)\nonumber\\
&\qquad\stackrel{(a)}{\leq} \sum_{\substack{\tilde{m}_{\one}, \tilde{m}_{\two}:\tilde{m}_{\one}\neq \mo\\(\vecx_{\mo}, \vecx_{\tilde{m}_{\one}}, \vecy_{\tilde{m}_{\two}},\vecy)\in T^{n}_{X\tilde{X}\tilde{Y}Y}}}\frac{|T^{n}_{Z|X\tilde{X}\tilde{Y}Y}(\vecx_{\mo},\vecx_{\tilde{m}_{\one}}, \vecy_{\tilde{m}_{\two}}, \vecy)|}{|T^n_{Z|XY}(\vecx_{\mo},\vecy)|}\nonumber\\
&\qquad \stackrel{(b)}{\leq} \sum_{\substack{\tilde{m}_{\one}, \tilde{m}_{\two}:\tilde{m}_{\one}\neq \mo\\(\vecx_{\mo}, \vecx_{\tilde{m}_{\one}}, \vecy_{\tilde{m}_{\two}},\vecy)\in T^{n}_{X\tilde{X}\tilde{Y}Y}}}\frac{\exp\inp{nH(Z|X\tilde{X}\tilde{Y}Y)}}{(n+1)^{-|\cX||\cY||\cZ|}\exp\inp{nH(Z|XY)}}\nonumber\\
&\qquad \stackrel{(c)}\leq \sum_{\substack{\tilde{m}_{\one}, \tilde{m}_{\two}:\\(\vecx_{\mo}, \vecx_{\tilde{m}_{\one}}, \vecy_{\tilde{m}_{\two}},\vecy)\in T^{n}_{X\tilde{X}\tilde{Y}Y}}}\exp\inp{-n\inp{I(Z;\tilde{X}\tilde{Y}|XY)-\epsilon}} \nonumber\\
&\qquad\stackrel{\text{(d)}}{\leq}\exp\inp{n\inp{|R_{\one}- I(\tilde{X};\tilde{Y}XY)|^{+}+|R_{\two}-I(\tilde{Y};XY)|^{+}-I(Z;\tilde{X}\tilde{Y}|XY)+2\epsilon}}\label{eq:upperbound2}
\end{align}
where (a) follows by noting that whenever $\vecz$ belongs to $T^n_{Z|X\tilde{X}\tilde{Y}Y}(\vecx_{\mo}, \vecx_{\tilde{m}_{\one}}, \vecy_{\tilde{m}_{\two}}, \vecy)$, $\vecz$ also belongs to $T^n_{Z|XY}(\vecx_{\mo}, \vecy)$; and for each $\vecz\in T^n_{Z|XY}(\vecx_{\mo}, \vecy)$, the value of $W^n(\vecz|\vecx_{\mo}, \vecy)$ is the same and is upper bounded by $1/|T^n_{Z|XY}(\vecx_{\mo}, \vecy)|$. $(b)$ follows from \eqref{eq:type_property3}, $(c)$ holds for large $n$ and $(d)$ follows from \eqref{codebook:3b}. 
To analyze \eqref{eq:upperbound2}, we will separately consider the following cases which together cover all the possibilities.
\begin{enumerate}
	\item $R_{\one}\leq I(\tilde{X};\tilde{Y}Y)$ and $R_{\two}\leq I(\tilde{Y};Y)$ \label{case1}
	\item $  I(\tilde{X};\tilde{Y}Y)<R_{\one}$ and $R_{\two}\leq I(\tilde{Y};XY)$\label{case2}
	\item $R_{\one}\leq  I(\tilde{X};\tilde{Y}XY) $ and $I(\tilde{Y};Y)<R_{\two}$\label{case3}
	\item $I(\tilde{X};\tilde{Y}XY)<R_{\one} $ and $I(\tilde{Y};XY)<R_{\two}$\label{case4}
\end{enumerate}

\noindent\underline{Case \ref{case1}: $R_{\one}\leq I(\tilde{X};\tilde{Y}Y)$ and $R_{\two}\leq I(\tilde{Y};Y)$}\\
In this case, \eqref{eq:1} implies that $I(X;\tilde{X}\tilde{Y}Y) \leq \epsilon$ which implies that $I(X;\tilde{X}\tilde{Y}|Y) \leq \epsilon$.  Thus, using the condition $I(XZ;\tilde{X}\tilde{Y}|Y)\geq \eta$ from definition of $\cP_1^{\eta}$, we see that
\begin{align*}
I(Z;\tilde{X}\tilde{Y}|XY) &= I(XZ;\tilde{X}\tilde{Y}|Y)-I(X;\tilde{X}\tilde{Y}|Y)\\
&\geq \eta-\epsilon.
\end{align*} 
Using \eqref{eq:upperbound2}, this implies that
\begin{align*}
\sum_{\vecz\in \cE_{\mo,1}(P_{X\tilde{X}\tilde{Y}YZ})}W^n(\vecz|\vecx_{\mo}, \vecy)&\leq \exp\inp{-n\inp{\eta-3\epsilon}}\\
&\rightarrow 0\text{ because }\eta>3\epsilon.
\end{align*}

\noindent \underline{Case ~\ref{case2}: $  I(\tilde{X};\tilde{Y}Y)<R_{\one}$ and $R_{\two}\leq I(\tilde{Y};XY)$}\\
 Using~\eqref{eq:1}, we have
 \begin{align*}
 -|R_{\two}-I(\tilde{Y};Y)|^{+} &\leq R_{\one} - I(\tilde{X};\tilde{Y}Y)-I(X;\tilde{X}\tilde{Y}Y)+\epsilon\\
 & = R_{\one} - I(\tilde{X};\tilde{Y}Y)-I(X;\tilde{Y}Y)-I(X;\tilde{X}|\tilde{Y}Y)+\epsilon\\
 & = R_{\one} - I(\tilde{X};\tilde{Y}XY)-I(X;\tilde{Y}Y)+\epsilon.
 \end{align*}
 This implies that 
 \begin{align*}
 R_{\one} - I(\tilde{X};\tilde{Y}XY)+\epsilon\geq I(X;\tilde{Y}Y) -|R_{\two}-I(\tilde{Y};Y)|^{+}.
 \end{align*}
 We will argue that the RHS of the above inequality is non-negative. When $R_{\two}\leq I(\tilde{Y};Y)$, the RHS is $I(X;\tilde{Y}Y)$ which is non-negative. Otherwise, when $I(\tilde{Y};Y)<R_{\two}\leq I(\tilde{Y};XY)$,
 \begin{align*}
  I(X;\tilde{Y}Y) -|R_{\two}-I(\tilde{Y};Y)|^{+} &=I(X;\tilde{Y}Y) -R_{\two}+I(\tilde{Y};Y)\\
  & = I(X;Y)+I(X;\tilde{Y}|Y)-R_{\two} +I(\tilde{Y};Y)\\
  &=I(\tilde{Y};XY)-R_{\two}+I(X;Y)\geq 0.
  \end{align*} 
 So, again the RHS is non-negative. This implies that $ R_{\one} \geq  I(\tilde{X};\tilde{Y}XY)-\epsilon$. Hence $ |R_{\one} -  I(\tilde{X};\tilde{Y}XY)|^{+}$ $\leq R_{\one} -  I(\tilde{X};\tilde{Y}XY) +\epsilon$. Thus, from \eqref{eq:upperbound2},
 \begin{align}
 \sum_{\vecz\in \cE_{\mo,1}}W^n(\vecz|\vecx_{\mo}, \vecy)&\leq \exp\inp{n\inp{R_{\one}-I(\tilde{X};\tilde{Y}XY) -I(Z;\tilde{X}\tilde{Y}|XY)+3\epsilon}}\label{rate_Ra_1}\\
 &= \exp\inp{n\inp{R_{\one}-I(\tilde{X};\tilde{Y}XY) -I(Z;\tilde{Y}|XY)-I(Z;\tilde{X}|\tilde{Y}XY)+3\epsilon}}\nonumber \\
 \qquad \qquad&= \exp\inp{n\inp{R_{\one}-I(\tilde{X};Z\tilde{Y}XY) -I(Z;\tilde{Y}|XY)+3\epsilon}}\nonumber \\
 \qquad \qquad&\leq \exp\inp{n\inp{R_{\one}-I(\tilde{X};Z\tilde{Y}) +3\epsilon}} \label{eq:cf_1}\\
  \qquad \qquad&\leq \exp\inp{n\inp{R_{\one}-I(\tilde{X};Z|\tilde{Y}) +3\epsilon}} \nonumber\\
 &\rightarrow 0
 \text{ for }R_{\one}\leq \min_{P_{X\tilde{X}\tilde{Y}YZ}\in \cP_1^{\eta}}{I(\tilde{X};Z|\tilde{Y})-3\epsilon}.
 \end{align}

Taking limit $\cP^{\eta}_1\rightarrow \cP_1^{0}$ (recall that $\epsilon <\eta/3)$, we get the following rate bound
 \begin{align}
 R_{\one}\leq \min_{P_{X\bar{X}\bar{Y}YZ}\in \cP_1^0}{I(\bar{X};Z|\bar{Y})}.\label{eq:rbound1}
\end{align}
\noindent \underline{Case~\ref{case3}: $R_{\one}\leq  I(\tilde{X};\tilde{Y}XY) $ and $I(\tilde{Y};Y)<R_{\two}$}\\
Using~\eqref{eq:1}, we obtain that 
 \begin{align*}
-|R_{\one}-I(\tilde{X};\tilde{Y}Y)|^{+}&\leq R_{\two} - I(\tilde{Y};Y)-I(X;\tilde{X}\tilde{Y}Y)+\epsilon\\
& = R_{\two} - I(\tilde{Y};Y)-I(X;Y)-I(X;\tilde{Y}|Y)-I(X;\tilde{X}|\tilde{Y}Y)+\epsilon\\
& = R_{\two} - I(\tilde{Y};XY)-I(X;Y)-I(X;\tilde{X}|\tilde{Y}Y)+\epsilon
\end{align*}
This implies that 
\begin{align*}
R_{\two} - I(\tilde{Y};XY)+\epsilon\geq I(X;Y)+I(X;\tilde{X}|\tilde{Y}Y) -|R_{\one}-I(\tilde{X};\tilde{Y}Y)|^{+}.
\end{align*}
Similar to the previous case, we will argue that RHS of the above inequality is non-negative. When $R_{\one}\leq I(\tilde{X};\tilde{Y}Y)$, it is clearly true. Otherwise, when $I(\tilde{X};\tilde{Y}Y)<R_{\one} \leq I(\tilde{X};\tilde{Y}XY)$, then 
\begin{align*}
I(X;\tilde{X}|\tilde{Y}Y) -|R_{\one}-I(\tilde{X};\tilde{Y}Y)|^{+} &=I(X;\tilde{X}|\tilde{Y}Y) -R_{\one}+I(\tilde{X};\tilde{Y}Y)\\
&=I(\tilde{X};\tilde{Y}XY)-R_{\one} \geq 0.
\end{align*}
Thus, for $R_{\one}\leq  I(\tilde{X};\tilde{Y}XY) $ and $I(\tilde{Y};Y)<R_{\two}$, we have $R_{\two} - I(\tilde{Y};XY)+\epsilon\geq0$. This implies that $|R_{\two} - I(\tilde{Y};XY)|^+\leq R_{\two} - I(\tilde{Y};XY) +\epsilon$. So, from \eqref{eq:upperbound2},
 \begin{align}
 \sum_{\vecz\in \cE_{\mo,1}(P_{X\tilde{X}\tilde{Y}YZ})}W^n(\vecz|\vecx_{\mo}, \vecy)&\leq \exp\inp{n\inp{R_{\two}-I(\tilde{Y};XY)-I(Z;\tilde{X}\tilde{Y}|XY)+3\epsilon}}\label{rate_Rb_1}\\
 &= \exp\inp{n\inp{R_{\two}-I(\tilde{Y};XY)-I(Z;\tilde{Y}|XY)-I(Z;\tilde{X}|XY\tilde{Y})+3\epsilon}}\nonumber \\
 &=\exp\inp{n\inp{R_{\two}-I(\tilde{Y};XYZ)-I(Z;\tilde{X}|XY\tilde{Y})+3\epsilon}}\nonumber \\
 &\leq\exp\inp{n\inp{R_{\two}-I(\tilde{Y};XZ)+3\epsilon}}\label{eq:cf_2}\\
 &\leq\exp\inp{n\inp{R_{\two}-I(\tilde{Y};Z|X)+3\epsilon}}\nonumber \\
 &\rightarrow 0 \text{ if }R_{\two}<\min_{P_{X\tilde{X}\tilde{Y}YZ}\in \cP^{\eta}_1}I(\tilde{Y};Z|X)-3\epsilon.\nonumber
 \end{align}
Taking limit $\cP_1^{\eta}\rightarrow \cP_1^{0}$, we get the rate bound,
 \begin{align}
 R_{\two}\leq \min_{P_{X\bar{X}\bar{Y}YZ}\in \cP_1^0}{I(\bar{Y};Z|X)}.\label{eq:rbound2}
\end{align}

\noindent\underline{Case 4: $I(\tilde{X};\tilde{Y}XY)<R_{\one} $ and $I(\tilde{Y};XY)<R_{\two}$}\\
From \eqref{eq:upperbound2}, we have
 \begin{align}
 \sum_{\vecz\in \cE_{\mo,1}(P_{X\tilde{X}\tilde{Y}YZ})}W^n(\vecz|\vecx_{\mo}, \vecy)&\leq \exp\inp{n\inp{R_{\one}- I(\tilde{X};\tilde{Y}XY)+R_{\two}-I(\tilde{Y};XY)-I(Z;\tilde{X}\tilde{Y}|XY)+2\epsilon}}\label{eq:cf_3}\\
 &\leq \exp\inp{n\inp{R_{\one}- I(\tilde{X};XY|\tilde{Y})+R_{\two}-I(\tilde{Y};XY)-I(Z;\tilde{X}\tilde{Y}|XY)+2\epsilon}}\\
 \qquad \qquad&= \exp\inp{n\inp{R_{\one}+R_{\two}-I(\tilde{X}\tilde{Y};XYZ)+3\epsilon}} \\
  \qquad \qquad&\leq \exp\inp{n\inp{R_{\one}+R_{\two}-I(\tilde{X}\tilde{Y};Z)+3\epsilon}} \\
 \qquad\qquad&\rightarrow 0 \text{ if }R_{\one}+R_{\two}<\min_{P_{X\tilde{X}\tilde{Y}YZ}\in \cP^{\eta}_1}I(\tilde{X}\tilde{Y};Z)-3\epsilon.
 \end{align}
Taking limit $\cP_1^{\eta}\rightarrow \cP_1^{0}$, we get the rate bound,
 \begin{align}
 R_{\one}+R_{\two}\leq \min_{P_{X\bar{X}\bar{Y}YZ}\in \cP_1^0}{I(\bar{X}\bar{Y};Z)}\label{eq:rbound3}
\end{align}
We define 
\begin{align*}
\cP_1^{+}& \defineqq  \{P_{X\bar{X}\bar{Y}YZ}\in \cP^n_{\cX\times\cY\times\cX\times\cY\times\cZ}: P_{XYZ}\in \cD_{0},\, P_{\bar{X}Y'Z}\in \cD_{0} \text{ for some }P_{Y'|\bar{X}Z}, P_{X'\bar{Y}Z}\in \cD_{0}\\
&\qquad  \,  \text{ for some }P_{X'|\bar{Y}Z},\, P_{X}=P_{\bar{X}}=P_{\one}, P_{\bar{Y}} = P_{\two}\text{ and }I(\bar{X};\bar{Y})=0,\,I(X;\bar{Y})=0\}
\end{align*} We see that $\cP_1^{0}\subseteq \cP_1^{+}$. Using this and 
collecting \eqref{eq:rbound1}, \eqref{eq:rbound2} and \eqref{eq:rbound3}, we obtain, 
\begin{align}
R_{\one}&\leq \min_{P_{X\bar{X}\bar{Y}YZ}\in \cP_1^{+}}{I(\bar{X};Z|\bar{Y})}\label{eq:Rbound1}\\
R_{\two}&\leq \min_{P_{X\bar{X}\bar{Y}YZ}\in \cP_1^{+}}{I(\bar{Y};Z|X)}\label{eq:Rbound2}\\
R_{\one}+R_{\two}&\leq \min_{P_{X\bar{X}\bar{Y}YZ}\in \cP_1^{+}}{I(\bar{X}\bar{Y};Z)}\label{eq:Rbound3}
\end{align}

Now, we analyze $P_{\cE_{\mo,2}}(\vecy)$, the second term in the RHS of~\eqref{err:upperbound}. We will upper bound $W^n\inp{\cE_{\mo,2}(P_{X\tilde{Y}_1\tilde{Y}_2YZ})|\vecx_{\mo}, \vecy}$. We see that by using~\eqref{codebook:4}, it is sufficient to consider distribution $P_{X\tilde{Y}_1\tilde{Y}_2YZ}\in \cP_2^\eta$ for which 
\begin{align}
I\inp{X;\tilde{Y}_1\tilde{Y}_2Y}\leq|R_{\two}-I(\tilde{Y}_1;Y)|^{+}+|R_{\two}-I(\tilde{Y}_2;\tilde{Y}_1 Y)|^{+} +\epsilon.\label{eq:2}
\end{align}
For $P_{X\tilde{Y}_1\tilde{Y}_2YZ}\in \cP_2^\eta$ satisfying~\eqref{eq:2},

\begin{align}
&\sum_{\vecz\in \cE_{\mo, 2}(P_{X\tilde{Y}_1\tilde{Y}_2YZ})}W^n(\vecz|\vecx_{\mo}, \vecy)\nonumber\\
&\qquad\leq\sum_{\substack{\tilde{m}_{\two 1},  \tilde{m}_{\two 2}:\tilde{m}_{\two 1}\neq  \tilde{m}_{\two 2},\\(\vecx_{\mo}, \vecy_{\tilde{m}_{\two 1}}, \vecy_{\tilde{m}_{\two 2}},\vecy)\in T^{n}_{X\tilde{Y}_1\tilde{Y}_2Y}}}\sum_{\vecz:(\vecx_{\mo}, \vecy_{\tilde{m}_{\two 1}}, \vecy_{\tilde{m}_{\two 2}},\vecy, \vecz)\in T^{n}_{X\tilde{Y}_1\tilde{Y}_2YZ}}W^n(\vecz|\vecx_{\mo}, \vecy)\nonumber\\
&\qquad\leq \sum_{\substack{\tilde{m}_{\two 1}, \tilde{m}_{\two 2}:\tilde{m}_{\two 1}\neq  \tilde{m}_{\two 2},\\(\vecx_{\mo}, \vecy_{\tilde{m}_{\two 1}}, \vecy_{\tilde{m}_{\two 2}},\vecy)\in T^{n}_{X\tilde{Y}_1\tilde{Y}_2Y}}}\frac{|T^{n}_{Z|X\tilde{Y}_1\tilde{Y}_2Y}(\vecx_{\mo},\vecy_{\tilde{m}_{\two 1}}, \vecy_{\tilde{m}_{\two 2}}, \vecy)|}{|T^n_{Z|XY}(\vecx_{\mo},\vecy)|}\nonumber\\
&\qquad \leq \sum_{\substack{\tilde{m}_{\two 1}, \tilde{m}_{\two 2}:\tilde{m}_{\two 1}\neq  \tilde{m}_{\two 2},\\(\vecx_{\mo}, \vecy_{\tilde{m}_{\two 1}}, \vecy_{\tilde{m}_{\two 2}},\vecy)\in T^{n}_{X\tilde{Y}_1\tilde{Y}_2Y}}}\frac{\exp\inp{nH(Z|X\tilde{Y}_1\tilde{Y}_2Y)}}{(n+1)^{-|\cX||\cY||\cZ|}\exp\inp{nH(Z|XY)}}\nonumber\\
&\qquad \stackrel{(a)}{\leq} \sum_{\substack{\tilde{m}_{\two 1}, \tilde{m}_{\two 2}:\tilde{m}_{\two 1}\neq  \tilde{m}_{\two 2},\\(\vecx_{\mo}, \vecy_{\tilde{m}_{\two 1}}, \vecy_{\tilde{m}_{\two 2}},\vecy)\in T^{n}_{X\tilde{Y}_1\tilde{Y}_2Y}}}\exp\inp{-n\inp{I(Z;\tilde{Y}_1\tilde{Y}_2|XY)-\epsilon}} \nonumber\\
&\qquad\stackrel{(b)}{\leq}\exp\inp{n\inp{|R_{\two}-I(\tilde{Y}_1;XY)|^{+}+|R_{\two}-I(\tilde{Y}_2;\tilde{Y}_1 XY)|^{+}-I(Z;\tilde{Y}_1\tilde{Y}_2|XY)+2\epsilon}},\label{eq:upperbound3}
\end{align}
where (a) holds for large $n$ and (b) follows from \eqref{codebook:5}.
Note that, in the analysis of $W^n\inp{\cE_{\mo,1}(P_{X\tilde{X}\tilde{Y}YZ})|\vecx_{\mo}, \vecy}$ (see the steps leading to \eqref{eq:upperbound2}), if we replace $R_{\one}$ with $R_{\two}$, $\tilde{Y}$ with $\tilde{Y}_1$ and $\tilde{X}$ with $\tilde{Y}_2$, \eqref{eq:upperbound2} changes to \eqref{eq:upperbound3} and the condition \eqref{eq:1} on the distribution changes to \eqref{eq:2}. With these replacements, we see that~\eqref{eq:upperbound3} goes to zero when the following hold (cf. \eqref{eq:cf_1},\eqref{eq:cf_2},\eqref{eq:cf_3}):
\begin{align}
	R_{\two}&<I(\tilde{Y}_{2};Z\tilde{Y}_1)-3\epsilon\label{cf:11}\\
	R_{\two}&<I(\tilde{Y}_1;XZ)-3\epsilon\label{cf:22}\\
	2R_{\two}&< I(\tilde{Y}_2;\tilde{Y}_1XYZ) + I(\tilde{Y}_1;XYZ)-3\epsilon.\label{cf:33}
\end{align} Here, \eqref{cf:33} is obtained by noting that $I(\tilde{Y}_2;\tilde{Y}_1XY) + I(\tilde{Y}_1;XY) + I(Z; \tilde{Y}_1\tilde{Y}_2|XY) = I(\tilde{Y}_2;\tilde{Y}_1XY) + I(\tilde{Y}_1;XY) + I(Z; \tilde{Y}_1|XY)+I(Z; \tilde{Y}_2|XY\tilde{Y}_1)$.
We first note that the rate bounds in \eqref{cf:11}, \eqref{cf:22} and \eqref{cf:33} still hold when
\begin{align}
	R_{\two}&<I(\tilde{Y}_{2};Z)-3\epsilon\label{cf:1}\\
	R_{\two}&<I(\tilde{Y}_1;Z)-3\epsilon\label{cf:2}\\
	2R_{\two}&< I(\tilde{Y}_2;Z) + I(\tilde{Y}_1;Z)-3\epsilon.\label{cf:3}
\end{align}
With this, we proceed similar to the previous case and define \begin{align*}
\cP_2^+& \defineqq \{P_{X\tilde{Y}_1\tilde{Y}_2YZ}\in \cP^n_{\cX\times\cY\times\cY\times\cY\times\cZ}: P_{XYZ}\in \cD_{0},\, P_{X'_1\tilde{Y}_1Z}\in \cD_{0}\text{ for some }P_{X'_1|\tilde{Y}_1Z}, \\
&\qquad \, P_{X'_2\tilde{Y}_2Z}\in \cD_{0} \text{ for some }P_{X'_2|\tilde{Y}_2Z}, \,   P_{X}=P_{\one}, P_{\tilde{Y}_1}=P_{\tilde{Y}_2} = P_{\two}\}.
\end{align*} Note that $\cP_2^0\subseteq\cP_2^+$. Using this and 
taking limit $\cP_2^{\eta}\rightarrow \cP_2^{0}$ 
we get the following rate bounds, 
\begin{align}
R_{\two}&\leq \min_{P_{XY\bar{Y}_1\bar{Y}_2Z}\in \cP_2^0}{I(\bar{Y}_2;Z)}\label{eq:Rbound4}\\
R_{\two}&\leq \min_{P_{XY\bar{Y}_1\bar{Y}_2Z}\in \cP_2^0}{I(\bar{Y}_1;Z)}\label{eq:Rbound5}\\
2R_{\two}&\leq \min_{P_{XY\bar{Y}_1\bar{Y}_2Z}\in \cP_2^0}{I(\bar{Y}_1;Z)+I(\bar{Y}_2;Z)}\label{eq:Rbound6}
\end{align}
When user \one is malicious, error will occur either in {\em Step 1} or {\em Step 3} or {\em Step 5}. Similar to the current case, we can show that error will not happen in {\em Step 1} w.h.p. because of typicality. For {\em Step 3} and {\em Step 5}, we will get bounds of the form \eqref{eq:Rbound1}, \eqref{eq:Rbound2} and \eqref{eq:Rbound3}. This is because we only consider the candidates which have passed {\em Step 2}. Hence, we get independence conditions (as we got from Lemma~\ref{lemma:indep} here).
Thus, combining \eqref{eq:Rbound1}, \eqref{eq:Rbound2}, \eqref{eq:Rbound3}, \eqref{eq:Rbound4}, \eqref{eq:Rbound5},\eqref{eq:Rbound6} and bounds from the case when user \one is malicious, we get the following rate region:
\begin{align*}
R_{\one}&\leq \min_{P_{XY'X'YZ}\in \cP}I(X;Z|Y)\\
R_{\two}&\leq \min_{P_{XY'X'YZ}\in \cP}I(Y;Z)
\end{align*}where $\cP$ is the set of distributions 
\begin{align*}
\cP \defineqq \{P_{XY'X'YZ}: P_{XY'Z}=P_{\one}P_{Y'}W, \,P_{X'YZ}=P_{X'}P_{\two}W, \,X\indep Y\}.
\end{align*}
This gives us one corner point (given by \eqref{eq:inner_bd_2}) of the rate region, we get the other corner point (given by \eqref{eq:inner_bd_1}) by  changing the order of decoding by performing {\em Step 3} before {\em Step 2}. 

\begin{remark}\label{remark:new_code_pos1}It is not clear if the statement of Theorem~\ref{thm:inner_bd} implies the achievability direction of Theorem~\ref{thm:main_result} (i.e., whether the inner bound in Theorem~\ref{thm:inner_bd} has a non-empty interior for all non-spoofable channels). However, we argue that the code (codebook and decoder) here can also be used to achieve positive rates for both users of a non-spoofable channel, which is the forward direction of Theorem~\ref{thm:main_result}.  In the current proof, the probability of error is analyzed by using the upper bounds \eqref{eq:upperbound2} and \eqref{eq:upperbound3} for a class of joint types. To show the achievability of Theorem~\ref{thm:main_result}, from \eqref{eq:upperbound2} (and \eqref{eq:upperbound3} respectively), we can instead follow the analysis from \eqref{eq:upperbound2_p} to \eqref{eq:upperbound123} (and \eqref{eq:upperbound3_p} to \eqref{eq:upperbound123_p} respectively) in the proof of achievability of Theorem~\ref{thm:main_result} in Section~\ref{sec:proof_thm1}. Note that \eqref{eq:upperbound2} is identical to \eqref{eq:upperbound2_p} and \eqref{eq:upperbound3} is identical to \eqref{eq:upperbound3_p}. Moreover, in the proof of achievability of Theorem~\ref{thm:main_result}, \eqref{eq:upperbound2_p} (and \eqref{eq:upperbound3_p} respectively) is analyzed for a class of joint types which is a superset of the class of joint types needed in the analysis of \eqref{eq:upperbound2} (and \eqref{eq:upperbound3} respectively).
The difference between the two proofs is highlighted in Flowchart~\ref{fig:flowchart1_authcom_appendix}. It might be possible to pursue an analysis of the current proof which also gives the achievability of Theorem~\ref{thm:main_result}. However, we could not find a way to obtain a simple inner bound using such an approach. 
\end{remark}

\end{proof}

\begin{proof}[Proof of Lemma~\ref{lemma:indep}]
Adding $(C)$ and $(D)$, 
\begin{align}
2\eta&> D(P_{X'\tilde{Y}Z}||P_{X'}P_{\tilde{Y}}W)+I(X\tilde{X};\tilde{Y}Z|X')\\
&=D(P_{X'\tilde{Y}Z}||P_{X'}P_{\tilde{Y}}W)+D(P_{X'\tilde{Y}X\tilde{X}Z}||P_{X'}P_{X\tilde{X}|X'}P_{\tilde{Y}Z|X'})\\
&=\sum_{x',\tilde{y}, x, \tilde{x}, z}P_{X'\tilde{Y}X\tilde{X}Z}(x',\tilde{y}, x, \tilde{x}, z)\left(\log{\left\{\frac{P_{X'\tilde{Y}Z}(x',\tilde{y}, z)}{P_{X'}(x')P_{\tilde{Y}}(\tilde{y})W(z|x', \tilde{y})}\right\}} + \log{\left\{\frac{P_{X'\tilde{Y}X\tilde{X}Z}(x', \tilde{y}, x, \tilde{x}, z)}{P_{X'}(x')P_{X\tilde{X}|X'}(x, \tilde{x}|x')P_{\tilde{Y}Z|X'}(\tilde{y},z|x')}\right\}}\right)\\
&=\sum_{x',\tilde{y}, x, \tilde{x}, z}P_{X'\tilde{Y}X\tilde{X}Z}(x',\tilde{y}, x, \tilde{x}, z)\left(\log{\left\{\frac{P_{X'\tilde{Y}Z}(x',\tilde{y}, z)\times P_{X'\tilde{Y}X\tilde{X}Z}(x',\tilde{y}, x, \tilde{x}, z)}{P_{X'}(x')P_{\tilde{Y}}(\tilde{y})W(z|x', \tilde{y})\times P_{X'}(x')P_{X\tilde{X}|X'}(x, \tilde{x}|x')P_{\tilde{Y}Z|X'}(\tilde{y},z|x')}\right\}}\right)\\
&=\sum_{x',\tilde{y}, x, \tilde{x}, z}P_{X'\tilde{Y}X\tilde{X}Z}(x',\tilde{y}, x, \tilde{x}, z)\left(\log{\left\{\frac{ P_{X'\tilde{Y}X\tilde{X}Z}(x',\tilde{y}, x, \tilde{x}, z)}{P_{X'}(x')P_{\tilde{Y}}(\tilde{y})W(z|x', \tilde{y}) P_{X\tilde{X}|X'}(x, \tilde{x}|x')}\right\}}\right)\\
&=\sum_{x',\tilde{y}, x, \tilde{x}, z}P_{X'\tilde{Y}X\tilde{X}Z}(x',\tilde{y}, x, \tilde{x}, z)\left(\log{\left\{\frac{ P_{X'\tilde{Y}X\tilde{X}Z}(x',\tilde{y}, x, \tilde{x}, z)}{P_{\tilde{Y}}(\tilde{y})P_{X'X\tilde{X}}(x', x, \tilde{x}) W(z|x', \tilde{y}) }\right\}}\right)\\
& = D(P_{X'\tilde{Y}X\tilde{X}Z}||P_{\tilde{Y}}P_{X\tilde{X}}P_{X'|X\tilde{X}}W)\\
&\stackrel{(a)}{\geq} D(P_{\tilde{Y}X\tilde{X}}||P_{\tilde{Y}}P_{X\tilde{X}})\\
& = I(\tilde{Y};X\tilde{X}),
\end{align} where $(a)$ follows from the log-sum inequality. Note that $I(\tilde{Y};X\tilde{X}) = I(\tilde{Y};X)+ I(\tilde{Y};\tilde{X}|X) = I(\tilde{Y};\tilde{X})+ I(\tilde{Y};X|\tilde{X})$. Thus, $I(\tilde{Y};X\tilde{X})\leq 2\eta$ implies $I(\tilde{Y};\tilde{X})\leq 2\eta$ and $I(\tilde{Y};X)\leq 2\eta$ as mutual information is always non-negative.
\end{proof}

\section{Proof of Theorem~\ref{thm:capacity_equivalence}}\label{app:rand_reduc}
We first prove Lemma~\ref{thm:rand_reduc} which gives the randomness reduction argument. 
\begin{lemma}[Randomness reduction]\label{thm:rand_reduc}
There exists $n_0(\cdot):\mathbb{R}^+\rightarrow\mathbb{N}$ such that given any $(N_{\one}, N_{\two}, L_{\one}, L_{\two}, n)$ adversary identifying code $(F_{\one}, F_{\two}, \phi_{F_{\one}, F_{\two}})$ with $P^{\rand}_{e}$ denoting its average probability of error and $\epsilon>0$, if 
\begin{align}\label{eq:rand_reduc_param}
\epsilon>2\log(1+P^{\rand}_{e}), 
\end{align} and $n\geq n_0(\epsilon)$
 there exists an $(N_{\one}, N_{\two}, n^2, n^2, n)$ adversary identifying code  $(F'_{\one}, F'_{\two}, \phi_{F'_{\one}, F'_{\two}})$  where the distributions $p_{F'_{\one}}$ and $p_{F'_{\two}}$ are the uniform distributions over encoder sets $\codeset'_{\one}\subseteq \codeset_{\one}$ and $\codeset'_{\two}\subseteq \codeset_{\two}$ (with $|\codeset'_{\one}|=|\codeset'_{\two}|=n^2$) respectively,  and the average probability of error is at most $\epsilon$. That is,
\begin{align}
    \frac{1}{N_{\one}\cdot N_{\two}} 
        \sum_{\substack{\mo\in\mathcal{M}_{\one}\\\mt\in\mathcal{M}_{\two}}}
            \sum_{\substack{f'_{\one}\in\codeset'_{\one}\\f'_{\two}\in\codeset'_{\two}}}
                \frac{1}{n^2\times n^2}\sum_{\substack{\vecz: \phi_{f'_{\one}, f'_{\two}}(\vecz)\\\notin \{(m_{\one}, m_{\two})\}}}W^{n}\inp{\vecz|f'_{\one}(\mo), f'_{\two}(\mt)}<\epsilon,\label{eq:rand_red_1}
\end{align}
\begin{align}
        \max_{\substack{\vecx\in\cX^{n}\\f'_{\one}\in\codeset'_{\one}}} 
            \left(\frac{1}{N_{\two}}\sum_{m_{\two}\in \mathcal{M}_{\two}}
                \sum_{f'_{\two}\in\codeset'_{\two}}
                    \frac{1}{n^2}\sum_{\substack{\vecz: \phi_{f'_{\one}, f'_{\two}}(\vecz)\\\notin \inb{(\cM_{\one}\times \inb{m_{\two}})\cup\inb{\oneb}}}}W^{n}\inp{\vecz|\vecx, f'_{\two}(\mt)}\right)< \epsilon \text{ and }\label{eq:rand_red_2}
\end{align}
\begin{align}
    \max_{\substack{\vecy\in\cY^{n}\\f'_{\two}\in\codeset'_{\two}}}
        \left(\frac{1}{N_{\one}}\sum_{m_{\one}\in \mathcal{M}_{\one}}
                \sum_{f'_{\one}\in\codeset'_{\one}}
                    \frac{1}{n^2}\sum_{\substack{\vecz: \phi_{f'_{\one}, f'_{\two}}(\vecz)\\\notin \inb{(\inb{m_{\one}}\times\cM_{\two})\cup\inb{\twob}}}}W^{n}\inp{\vecz|f'_{\one}(\mo), \vecy}\right)< \epsilon.\label{eq:rand_red_3}
\end{align}

\end{lemma}
\begin{proof}
The proof is along the lines of \cite[Lemma 12.8]{CsiszarKorner} and Jahn~\cite[Theorem~1]{Jahn81}. 
Let $\inb{F_{\one, i}}, i=1,\ldots, n^2$ be i.i.d. according to $p_{F_{\one}}$. Similarly, let $\inb{F_{\two, j}}, j=1,\ldots, n^2$ be i.i.d. according to $p_{F_{\two}}$. Further, let $\inb{F_{\one, i}}_{j=1}^{n^2}$ and $\inb{F_{\two, j}}_{j=1}^{n^2}$ be independent. 


Define $e_{\one}(f_{\one}, f_{\two}, \vecx)$ to be the error probability for fixed encoding maps $f_{\one}$ for user \one and $f_{\two}$ for user \two,  and the channel inputs chosen by the adversarial user $\one$ as $\vecx$, {\em i.e.},
\begin{align*}
&e_{\one}(f_{\one}, f_{\two}, \vecx)\defineqq\frac{1}{N_{\two}}\sum_{m_{\two}\in \cM_{\two}}\sum_{\substack{\vecz: \phi_{f_{\one}, f_{\two}}(\vecz)\\\notin (\cM_{\one}\times \inb{m_{\two}})\cup\inb{\oneb}}}W^n\inp{\vecz|\vecx, f_{\two}(m_{\two})}.
\end{align*} Similarly,  for adversarial user $\two$ with input $\vecy\in \cY^n$, let $e_{\two}(f_{\one}, f_{\two}, \vecy)$ be defined as 
\begin{align*}
&e_{\two}(f_{\one}, f_{\two}, \vecy)\defineqq\frac{1}{N_{\one}}\sum_{m_{\one}\in \cM_{\one}}\sum_{\substack{\vecz: \phi_{f_{\one}, f_{\two}}(\vecz)\\\notin (\inb{m_{\one}}\times\cM_{\two})\cup\inb{\twob}}}W^n\inp{\vecz| f_{\one}(m_{\one}), \vecy}. 
\end{align*}
When both users are honest, we define
\begin{align*}
&e(f_{\one}, f_{\two})\defineqq\frac{1}{N_{\one}N_{\two}}\sum_{(m_{\one}, m_{\two})\in \cM_{\one}\times\cM_{\two}}\,\,\,\sum_{\substack{\vecz: \phi_{f_{\one}, f_{\two}(\vecz)}\\\neq (m_{\one}, m_{\two})}}W^n\inp{\vecz| f_{\one}(m_{\one}), f_{\two}(m_{\two})}.
\end{align*} For any $j, j'\in [1:n^2]$, note that $e_{\one}(F_{\one, j}, F_{\two, j'}, \vecx), \, e_{\two}(F_{\one, j}, F_{\two, j'}, \vecy)$ and $e(F_{\one, j}, F_{\two, j'})$,  as functions of $F_{\one, j}$ and  $F_{\two, j'}$, are random variables. We will show that 
\begin{align}
\bbP\Bigg\{\Bigg(\frac{1}{n^2\times n^2}\sum_{j, j'\in[1:n^2]}e(F_{\one, j}, F_{\two, j'})\geq \epsilon\Bigg)&\bigcup\Bigg(\bigcup_{\substack{\vecx\in \cX^n, j\in [1:n^2]}}\bigg(\frac{1}{n^2}\sum_{j'\in[1:n^2]}e_{\one}(F_{\one, j}, F_{\two, j'}, \vecx)\geq \epsilon\bigg)\Bigg)\nonumber\\
&\bigcup\Bigg(\bigcup_{\vecy\in \cY^n, j'\in [1:n^2]}\bigg(\frac{1}{n^2}\sum_{j\in[1:n^2]}e_{\two}(F_{\one, j}, F_{\two, j'}, \vecy)\geq \epsilon\bigg)\Bigg)\Bigg\}\label{eq:rand_reduction_proof}
\end{align} is less than $1$. This will imply the existence of $n^2$ deterministic encoders $\codeset'_{\one}=\inb{f_{\one, j}:{j\in[1:n^2]}}$ and $\codeset'_{\two}=\inb{f_{\two, j}:j\in[1:n^2]}$ satisfying \eqref{eq:rand_red_1}, \eqref{eq:rand_red_2} and \eqref{eq:rand_red_3}.

For $j'\in n^2$, a fixed encoder $f_{\one}$ for user $\one$ and an input vector $\vecx$ for malicious user $\one$,  let $e^{*}_{\one}(f_{\one}, p_{F_{\two}}, \vecx)  \defineqq \bbE_{F_{\two, j'}}e_{\one}(f_{\one}, F_{\two, j'}, \vecx)$.  Note that for the given code $(F_{\one}, F_{\two}, \phi_{F_{\one}, F_{\two}})$, the average probability of error when user \one is malicious,\\   $P^{\rand}_{e,\malone} =
        \max_{\substack{\vecx\in\cX^n\\\fo\in\codeset_{\one}}} e^{*}_{\one}(f_{\one}, p_{F_{\two}}, \vecx)$.

\noindent For any $ j\in[1:n^2]$,
\begin{align*}
\bbP\inp{\frac{1}{n^2}\sum_{j'\in[1:n^2]}e_{\one}(F_{\one, j}, F_{\two, j'}, \vecx)\geq \epsilon}&= \bbP\inp{\exp\inp{\sum_{j'\in[1:n^2]}e_{\one}(F_{\one, j}, F_{\two, j'}, \vecx)}\geq \exp\inp{n^2\epsilon}}\\
&\leq \exp\inp{-n^2\epsilon}\bbE\insq{\exp\inp{\sum_{j'\in[1:n^2]}e_{\one}(F_{\one, j}, F_{\two, j'}, \vecx)}}\\
&=\exp\inp{-n^2\epsilon}\bbE\insq{\prod_{j'\in[1:n^2]}\exp\inp{e_{\one}(F_{\one, j}, F_{\two, j'}, \vecx)}}
\end{align*}
But,
\begin{align*}
\bbE\insq{\prod_{j'\in[1:n^2]}\exp\inp{e_{\one}(F_{\one, j}, F_{\two, j'}, \vecx)}} &= \bbE_{F_{\one, j}}\insq{\bbE_{(F_{\two, 1}, \ldots, F_{\two, n^2})|F_{\one, j}}\insq{\prod_{j'\in[1:n^2]}\exp\inp{e_{\one}(F_{\one, j}, F_{\two, j'}, \vecx)}}}\\
&=\sum_{f_{\one}\in \Gamma_{\one}}p_{F_{\one}}(f_{\one})\bbE_{(F_{\two, 1}, \ldots, F_{\two, n^2})|F_{\one, j}=f_{\one}}\insq{\prod_{j'\in[1:n^2]}\exp\inp{e_{\one}(F_{\one, j}, F_{\two, j'}, \vecx)}}\\
&\stackrel{(a)}=\sum_{f_{\one}\in \Gamma_{\one}}p_{F_{\one}}(f_{\one})\prod_{j'\in[1:n^2]}\bbE_{F_{\two, j'}}\insq{\exp\inp{e_{\one}(f_{\one}, F_{\two, j'}, \vecx)}}\\
&\stackrel{(b)}=\sum_{f_{\one}\in \Gamma_{\one}}p_{F_{\one}}(f_{\one})\inp{\bbE_{F_{\two, 1}}\insq{\exp\inp{e_{\one}(f_{\one}, F_{\two, 1}, \vecx)}}}^{n^2}\\
&\stackrel{(c)}\leq \sum_{f_{\one}\in \Gamma_{\one}}p_{F_{\one}}(f_{\one})\inp{\bbE_{F_{\two, 1}}\insq{1+e_{\one}(f_{\one}, F_{\two, 1}, \vecx)}}^{n^2}\\
&= \sum_{f_{\one}\in \Gamma_{\one}}p_{F_{\one}}(f_{\one})\inp{1+\bbE_{F_{\two, 1}}\insq{e_{\one}(f_{\one}, F_{\two, 1}, \vecx)}}^{n^2}\\
&= \sum_{f_{\one}\in \Gamma_{\one}}p_{F_{\one}}(f_{\one})\inp{1+e^{*}_{\one}(f_{\one}, p_{F_{\two}}, \vecx)}^{n^2}\\
&\leq  \sum_{f_{\one}\in \Gamma_{\one}}p_{F_{\one}}(f_{\one})\inp{1+P^{\rand}_{e,\malone}}^{n^2}\\
&\leq (1+P^{\rand}_{e})^{n^2},
\end{align*} 
where $(a)$ holds because $F_{\one, j}\indep \inp{F_{\two, 1}, \ldots, F_{\two, n^2}}$ and $\inb{\exp\inp{e_{\one}(f_{\one}, F_{\two, j'}, \vecx)}}_{j'=1}^{n^2}$ are i.i.d. random variables. $(b)$  also follows from the i.i.d. nature of $\inb{\exp\inp{e_{\one}(f_{\one}, F_{\two, j'}, \vecx)}}_{j'=1}^{n^2}$. To see inequality $(c)$, recall that $\exp\inp{\alpha}$ stands for $2^{\alpha}$ which is upper bounded by $(1+\alpha)$ for $0\leq \alpha\leq 1$. Thus, 

\begin{align}
&\bbP\inp{\frac{1}{n^2}\sum_{j'\in[1:n^2]}e_{\one}(F_{\one, j}, F_{\two, j'}, \vecx)\geq \epsilon \text{ for any }\vecx\in \cX^n\text{ and }j\in [1:n^2]}\\
&\hspace{7cm}\leq |\cX|^n\cdot n^2\cdot\exp\inp{-n^2\epsilon}(1+P^{\rand}_{e})^{n^2}\nonumber\\
&\hspace{7cm}=\exp\inb{-n^2\inp{\epsilon-\frac{\log{|\cX|}}{n}-\frac{2\log{n}}{n^2}-\log(1+P^{\rand}_{e})}}.\label{eq:rand_red_proof1}
\end{align}
Similarly, 
we have
\begin{align}
&\bbP\inp{\frac{1}{n^2}\sum_{j\in[1:n^2]}e_{\two}(F_{\one, j}, F_{\two, j'}, \vecy)\geq \epsilon \text{ for any }\vecy\in \cY^n\text{ and }j'\in [1:n^2]}\\
&\hspace{7cm}\leq |\cY|^n\cdot n^2 \cdot \exp\inp{-n^2\inp{\epsilon-\log(1+P^{\rand}_{e})}}\nonumber\\
&\hspace{7cm}=\exp\inb{-n^2\inp{\epsilon-\frac{\log{|\cY|}}{n}-\frac{2\log{n}}{n^2}-\log(1+P^{\rand}_{e})}}.\label{eq:rand_red_proof2}
\end{align}
Next, we will compute an  upper bound on  $\bbP\inp{\frac{1}{n^2\times n^2}\sum_{j\in[1:n^2]}\sum_{j'\in[1:n^2]}e(F_{\one, j}, F_{\two, j'})\geq \epsilon}$. 
To this end, let $\Sigma_{n^2}\defineqq\{\tau_i: i\in [0:n^2-1]\}$ be a set of permutations (in fact cyclic shifts) of $(1, 2, \ldots, n^2)$ such that $$\tau_i(j) = (i+j)\, \textsf{mod} \,n^2\mbox{ for all }j\in [1:n^2].$$ With this,
\begin{align*}
\bbP\inp{\frac{1}{n^2\times n^2}\sum_{j\in[1:n^2]}\sum_{j'\in[1:n^2]}e(F_{\one, j}, F_{\two, j'})\geq \epsilon} & = \bbP\inp{\frac{1}{n^2}\sum_{\sigma\in \Sigma_{n^2}}\inp{\frac{1}{n^2}\sum_{j\in[1:n^2]}e(F_{\one, j}, F_{\two, \sigma(j)})}\geq \epsilon}.
\end{align*}
For $\sigma\in \Sigma_{n^2}$, let $P_{\sigma}\defineqq\frac{1}{n^2}\sum_{j\in[1:n^2]}e(F_{\one, j}, F_{\two, \sigma(j)})$. 
\begin{align*}
\bbP\inp{\sum_{\sigma\in \Sigma_{n^2}}{P_{\sigma}\geq n^2\epsilon}}&\leq \bbP\inp{\cup_{\sigma\in \Sigma_{n^2}}\inp{P_{\sigma}\geq \epsilon}}\\
&\leq \sum_{\sigma\in \Sigma_{n^2}}\bbP\inp{P_{\sigma}\geq \epsilon}.
\end{align*}
Note that for $\sigma\in \Sigma_{n^2}$, $P_{\sigma}$ are identically distributed random variables. Thus,
$$\bbP\inp{\frac{1}{n^2\times n^2}\sum_{j\in[1:n^2]}\sum_{j'\in[1:n^2]}e(F_{\one, j}, F_{\two, j'})\geq \epsilon} \leq n^2\bbP\inp{P_{\tau_0}\geq \epsilon}.$$
But, 
\begin{align*}
\bbP\inp{P_{\tau_0}\geq \epsilon}& = \bbP\inp{\frac{1}{n^2}\sum_{j\in[1:n^2]}e(F_{\one, j}, F_{\two, \tau_0(j)})\geq \epsilon}\\
& =\bbP\inp{\frac{1}{n^2}\sum_{j\in[1:n^2]}e(F_{\one, j}, F_{\two, j})\geq \epsilon}\\
& =\bbP\inp{\sum_{j\in[1:n^2]}e(F_{\one, j}, F_{\two, j})\geq n^2\epsilon}\\
& =\bbP\inp{\exp\inp{\sum_{j\in[1:n^2]}e(F_{\one, j}, F_{\two, j})}\geq \exp\inp{n^2\epsilon}}\\
& \leq \exp\inp{-n^2\epsilon}\bbE\insq{\exp\inp{\sum_{j\in[1:n^2]}e(F_{\one, j}, F_{\two, j})}}.
\end{align*}
Note that $\inb{e(F_{\one, j}, F_{\two, j})}_{j=1}^{n^2}$ are i.i.d. random variables. Hence,
\begin{align*}
\bbE\insq{\exp\inp{\sum_{j\in[1:n^2]}e(F_{\one, j}, F_{\two, j})}}&=\bbE\insq{\prod_{j\in[1:n^2]}\exp\inp{e(F_{\one, j}, F_{\two, j})}}\\
&=\prod_{j\in[1:n^2]}\bbE\insq{\exp\inp{e(F_{\one, j}, F_{\two, j})}}\\
&=\inp{\bbE\insq{\exp\inp{e(F_{\one, 1}, F_{\two, 1})}}}^{n^2}\\
&\stackrel{(a)}{\leq}\inp{1+\bbE\inp{e(F_{\one, 1}, F_{\two, 1})}}^{n^2}\\
&\stackrel{(b)}{\leq}\inp{1+P^{\rand}_{e}}^{n^2}\\
\end{align*}
where $(a)$ holds because for $0\leq\alpha\leq 1$, $\exp(\alpha) = 2^{\alpha}\leq 1+\alpha$ and $(b)$ holds because for the given code $(F_{\one}, F_{\two}, \phi_{F_{\one}, F_{\two}})$, $P^{\rand}_{e,\na} = \bbE\inp{e(F_{\one, 1}, F_{\two, 1})}$. Thus, 
\begin{align}
\bbP\inp{\frac{1}{n^2\times n^2}\sum_{j\in[1:n^2]}\sum_{j'\in[1:n^2]}e(F_{\one, j}, F_{\two, j'})\geq \epsilon}&\leq n^2\exp\inp{-n^2\epsilon}\inp{1+P^{\rand}_{e}}^{n^2}\\
&=\exp\inp{2\log{n}-n^2\epsilon+n^2\log\inp{1+P^{\rand}_{e}}}\\
&=\exp\inp{-n^2\inp{\epsilon-\frac{2\log{n}}{n^2}-\log\inp{1+P^{\rand}_{e}}}}.\label{eq:rand_red_proof3}
\end{align}
From \eqref{eq:rand_red_proof1}, \eqref{eq:rand_red_proof2} and \eqref{eq:rand_red_proof3}, we note that by using a union bound, \eqref{eq:rand_reduction_proof} is upper bounded by 
\begin{align*}
&\exp\inb{-n^2\inp{\epsilon-\frac{\log{|\cX|}}{n}-\frac{2\log{n}}{n^2}-\log(1+P^{\rand}_{e})}}+ \exp\inb{-n^2\inp{\epsilon-\frac{\log{|\cY|}}{n}-\frac{2\log{n}}{n^2}-\log(1+P^{\rand}_{e})}}\\&+\exp\inp{-n^2\inp{\epsilon-\frac{2\log{n}}{n^2}-\log\inp{1+P^{\rand}_{e}}}},
\end{align*} which is less than $1$ for large enough $n (=: n_0(\epsilon))$ which depends only on the input alphabet sizes and $\epsilon$. 
\end{proof}

\section{Proof of Lemma~\ref{lemma:AV_MAC}}\label{proof_outer_bound}
\begin{proof}
\olive{Consider an $(N_{\one}, N_{\two}, n)$ adversary identifying code $(f_{\one},f_{\two}, \phi)$ such that $P_{e}(f_{\one},f_{\two}, \phi) \leq \epsilon$ 
For $i\in [1:n]$, let $(Q_{i,X'|X}, Q_{i,Y'|Y})$ be an arbitrary sequence of pairs of channel distributions such that each pair satisfies \eqref{eq:outer_bound} and define $\tilde{W}_{i}$ as 
\begin{align}
\tilde{W}_{i}(z|x,y) \defineqq \sum_{x'}Q_{i,X'|X}(x'|x)W(z|x',y) = \sum_{y'}Q_{i,Y'|Y}(y'|y)W(z|x,y') \label{eq:thm:outer_bd}
\end{align}
for all $x,y,z$. Let $Q_{\vecX'|\vecX}\defineqq \prod_{i=1}^{n}Q_{i,X'|X}$, $Q_{\vecY'|\vecY}\defineqq \prod_{i=1}^{n}Q_{i,Y'|Y}$ and $\tilde{W}^{(n)} \defineqq \prod_{i=1}^{n}\tilde{W}_{i}$. Recall that for $\mt \in \cM_{\two}$ and $\phi_{\two}$ as defined in \eqref{eq:feas_conv_dec2}, 
we define $\cE^{\two}_{\mt} = \inb{\vecz:\phi_{\two}(\vecz)\notin\{\mt,\oneb\}}$. Consider a malicious user-\one who chooses $M_{\one}$ uniformly from $[1:N_{\one}]$, passes $f_{\one}(M_{\one})$ over $Q_{\vecX'|\vecX}$, and transmits the resulting vector. We may conclude that  (see \eqref{eq:mal1}),
\begin{align}
P_{e, \malone} 
  &= \max_{\vecx\in\cX^n} \left(\frac{1}{N_{\two}}\sum_{\mt}W^n\inp{\cE^{\two}_{\mt}\Big|\vecx, f_{\two}(\mt)}\right)\notag\\
  &\geq \frac{1}{N_{\two}}\sum_{\mt} \sum_{\vecx}\inp{\frac{1}{N_{\one}}\sum_{\mo} Q_{\vecX'|\vecX}(\vecx|f_{\one}(\mo))} W^n\inp{\cE^{\two}_{\mt}\Big|\vecx, f_{\two}(\mt)}\notag\\
  &= \frac{1}{N_{\one}\cdot N_{\two}}\sum_{\mo, \mt}\sum_{\vecx}Q_{\vecX'|\vecX}(\vecx|f_{\one}(\mo))W^n\inp{\cE^{\two}_{\mt}\Big|\vecx,f_{\two}(\mt)}\notag\\
  &= \frac{1}{N_{\one}\cdot N_{\two}}\sum_{\mo, \mt}\tilde{W}^{(n)}\inp{\cE^{\two}_{\mt}\Big|f_{\one}(\mo),f_{\two}(\mt)}.\label{eq:outerbdconverse:malone}
\end{align}
Similarly, for $\phi_{\one}$ as defined in \eqref{eq:feas_conv_dec1}, $\mo\in\cM_{\one}$ and $\cE^{\one}_{\mo} = \inb{\vecz:\phi_{\one}(\vecz)\notin\{\mo,\twob\}}$,
\begin{align}
P_{e, \maltwo} \geq 
\frac{1}{N_{\one}\cdot N_{\two}}\sum_{\mo, \mt} \tilde{W}^{(n)}\inp{\cE^{\one}_{\mo}\Big|f_{\one}(\mo),f_{\two}(\mt)}.
\label{eq:outerbdconverse:maltwo}
\end{align}
First notice that for $(\mo,\mt)\in \cM_{\one}\times\cM_{\two}$ and $\vecz\in\cZ^n$,
\begin{align*}
&\inb{\vecz:\phi(\vecz)\neq(\mo, \mt)}\\
&=\inb{\vecz:\phi_{\one}(\vecz)\neq\mo}\cup \inb{\vecz:\phi_{\two}(\vecz)\neq\mt}\\
& \stackrel{(a)}{=} \inb{\vecz:\phi_{\one}(\vecz)\notin\{\mo,\twob\}}\cup \inb{\vecz:\phi_{\two}(\vecz)\notin\{\mt,\oneb\}}\\
&=\cE^{\one}_{\mo}\cup\cE^{\two}_{\mt}
\end{align*} where $(a)$ holds by noting from definitions \eqref{eq:feas_conv_dec1} and \eqref{eq:feas_conv_dec2} that $\phi_{\one} = \oneb$ (resp. $\twob$) if and only if $\phi_{\two} = \oneb$ (resp. $\twob$).
Thus,
\begin{align}
&\frac{1}{N_{\one}\cdot N_{\two}}\sum_{\mo, \mt}\tilde{W}^{(n)}\inp{\inb{\vecz:\phi(\vecz) \neq (\mo, \mt)}\Big|f_{\one}(\mo),f_{\two}(\mt)} \nonumber\\
&\qquad = \frac{1}{N_{\one}\cdot N_{\two}}\sum_{\mo, \mt}\tilde{W}^{(n)}\inp{\cE^{\two}_{\mt}\cup\cE^{\one}_{\mo}\Big|f_{\one}(\mo),f_{\two}(\mt)} \nonumber\\
&\qquad\stackrel{\text{(a)}}{\leq}  P_{e, \malone} + P_{e, \maltwo}\nonumber\\
&\qquad\leq 2\epsilon,\nonumber
\end{align}
where (a) follows from a union bound and \eqref{eq:outerbdconverse:malone}-\eqref{eq:outerbdconverse:maltwo}.
Recall that every pair $(Q_{X'|X}, Q_{Y'|Y})$ satisfying \eqref{eq:outer_bound} corresponds to an element in $\tilde{\cW}_W$, which is a convex set (see the discussion after Definition~\ref{defn:outerboundAVMAC}). Thus, for any $\epsilon>0$, an adversary identifying code for the MAC $W$ with an average probability of error $\epsilon$ is also a communication code for the AV-MAC $\tilde{\cW}_{W}$ with an average probability of error at most $2\epsilon$. So, the deterministic coding capacity region $\cC$ of $W$ is contained by the deterministic coding capacity region $\cC_{\AVMAC}(\tilde{\cW}_W)$ of the AV-MAC $\tilde{\cW}_{W}$. }
\end{proof} 

\section{Proof of Theorem~\ref{thm:capacity_equivalence}}
\begin{proof}[Proof of Theorem~\ref{thm:capacity_equivalence}]
We first restate the theorem below.
\begin{Theorem}
 $ \cC=\cC^{\rand}$ whenever $(R_{\one}, R_{\two})\in \cC$ for some $R_{\one}, R_{\two}>0$.
\end{Theorem}
Since deterministic codes are a subset of randomized codes, $\cC\subseteq \cC^{\rand}$. It only remains to show that $ \cC\supseteq \cC^{\rand}$ whenever $(R'_{\one}, R'_{\two})\in \cC$ for some $R'_{\one}, R'_{\two}>0$.

Let $\epsilon>0$  and large enough $n$ satisfying $n\geq n_0(\epsilon)$ (given by Lemma~\ref{thm:rand_reduc}). Consider an achievable rate pair $(R_{\one}, R_{\two})\in \cC^{\rand}$. This implies that for $\inp{N_{\one}, N_{\two}} = \inp{2^{nR_{\one}}, 2^{nR_{\two}}}$, there exists an $(N_{\one}, N_{\two}, L_{\one}, L_{\two}, n)$ adversary identifying code $(F_{\one}, F_{\two}, \phi_{F_{\one}, F_{\two}})$ with probability of error $P^{\rand}_{e}$ which vanishes with $n$. 
Choose $\epsilon>0$ such that 
\begin{align}
\epsilon>2\log(1+P^{\rand}_{e}).
\end{align} By Lemma~\ref{thm:rand_reduc}, there exists an $(N_{\one}, N_{\two}, n^2, n^2, n)$ adversary identifying code  $(F'_{\one}, F'_{\two}, \phi_{F'_{\one}, F'_{\two}})$  where the distributions $p_{F'_{\one}}$ and $p_{F'_{\two}}$ on the encoder sets $\codeset'_{\one}$ and $\codeset'_{\two}$ (with $|\codeset'_{\one}|=|\codeset'_{\two}|=n^2$) respectively, are uniform,  and the average probability of error is at most $\epsilon$. That is,
\begin{align}
    \frac{1}{N_{\one}\cdot N_{\two}} 
        \sum_{\substack{\mo\in\mathcal{M}_{\one}\\\mt\in\mathcal{M}_{\two}}}
            \sum_{\substack{l_{\one}\in n^2\\l_{\two}\in n^2}}
                \frac{1}{n^2\times n^2}\sum_{\substack{\vecz: \phi_{f_{\one, l_{\one}}, f_{\two, l_{\two}}}(\vecz)\\\notin \{(m_{\one}, m_{\two})\}}}W^{n}\inp{\vecz|f_{\one, l_{\one}}(\mo), f_{\two, l_{\two}}(\mt)}<\epsilon,\label{eq:rand_red_11}
\end{align}
\begin{align}
        \max_{\substack{\vecx\in\cX^{n}\\l_{\one}\in n^2}} 
            \left(\frac{1}{N_{\two}}\sum_{m_{\two}\in \mathcal{M}_{\two}}
                \sum_{l_{\two}\in n^2}
                    \frac{1}{n^2}\sum_{\substack{\vecz: \phi_{f_{\one, l_{\one}}, f_{\two, l_{\two}}}(\vecz)\\\notin \inb{(\cM_{\one}\times \inb{m_{\two}})\cup\inb{\oneb}}}}W^{n}\inp{\vecz|\vecx, f_{\two, l_{\two}}}\right)< \epsilon \text{ and }\label{eq:rand_red_22}
\end{align}
\begin{align}
    \max_{\substack{\vecy\in\cY^{n}\\l_{\two}\in n^2}}
        \left(\frac{1}{N_{\one}}\sum_{m_{\one}\in \mathcal{M}_{\one}}
                \sum_{l_{\one}\in n^2}
                    \frac{1}{n^2}\sum_{\substack{\vecz: \phi_{f_{\one, l_{\one}}, f_{\two, l_{\two}}}(\vecz)\\\notin \inb{(\inb{m_{\one}}\times\cM_{\two})\cup\inb{\twob}}}}W^{n}\inp{\vecz|f_{\one, l_{\one}}(\mo), \vecy}\right)< \epsilon.\label{eq:rand_red_33}
\end{align}

Further, since $(R'_{\one}, R'_{\two})\in \cC$ for $R'_{\one}, R'_{\two}>0$ ({\em i.e.} $(R'_{\one}, R'_{\two})$ is an achievable rate pair), there exists an $(n^2, n^2, k_n)$ code $(\hat{f}_{\one},\hat{f}_{\two}, \hat{\phi})$ where $k_n/n\rightarrow 0$ and 
\begin{align}P_{e}(\hat{f}_{\one},\hat{f}_{\two}, \hat{\phi})\leq \epsilon\label{eq:two}\end{align}
for large enough $n$. We choose sufficiently large $n$ such that \eqref{eq:rand_red_11}, \eqref{eq:rand_red_22}, \eqref{eq:rand_red_33} and \eqref{eq:two} hold. For a vector sequence $\tilde{\vecs}\in \cS^{k_n+n}$ for any alphabet $\cS$, we write $\tilde{\vecs}= ({\hat{\vecs},\vecs})$, where $\hat{\vecs}$ denotes the first $k_n$-length part of $\tilde{\vecs}$ and ${\vecs}$ denotes the last $n$-length part of the $\tilde{\vecs}$. 
Let $(\tilde{f}_{\one},\tilde{f}_{\two}, \tilde{\phi})$ be a new $(\tilde{N}_{\one}, \tilde{N}_{\two},\tilde{n})$ code where $\tilde{n}:=k_n+n$, message set for user-$i\in \{\one, \two\}$, $\tilde{\cM}_i=[1:\tilde{N}_i]:=\{1, 2, \ldots, n^2\}\times [1:{N}_i]$. Further, for $l_{\one}\in [1:n^2],\,    m_{\one}\in [1:{N}_{\one}]$, let $\tilde{m}_{\one} := (l_{\one},m_{\one})$. We define $\tilde{f}_{\one}(\tilde{m}_{\one}) = \tilde{f}_{\one}(l_{\one}, m_{\one}):= \inp{\hat{f}_{\one}(l_{\one}),f_{\one,l_{\one}}(m_{\one})}$. For $\tilde{\vecz}=(\hat{\vecz},\vecz)$, if $\hat{\phi}(\hat{\vecz})=(\hat{l}_{\one}, \hat{l}_{\two})$, we define $\tilde{\phi}(\tilde{\vecz})= \phi_{f_{\one,\hat{l}_{\one}},f_{\two,\hat{l}_{\two}} }(\vecz)$. Otherwise, if $\hat{\phi}(\hat{\vecz})\in\{\oneb, \twob\}$, $\tilde{\phi}(\tilde{\vecz})=\hat{\phi}(\hat{\vecz})$.
Then, for the new code, 

\begin{align}
P_{e,\maltwo} &= \max_{\tilde{\vecy}\in\cY^{\tilde{n}}}\frac{1}{\tilde{N}_{\one}}\sum_{\tilde{m}_{\one}\in \tilde{\mathcal{M}}_{\one}}\bbP\inp{\inp{\tilde{\phi}(\tilde{\vecZ})\notin {(\inb{\tilde{m}_{\one}}\times\tilde{\cM}_{\two})\cup\inb{\twob}}}\bigg|\tilde{\vecX} = \tilde{f}_{\one}(\tilde{m}_{\one}), \tilde{\vecY} = \tilde{\vecy}}\\
&=\max_{(\hat{\vecy}\times \vecy)\in\cY^{\tilde{n}}}\Bigg(\frac{1}{n^2{N}_{\one}}\sum_{(l_{\one}, {m}_{\one})\in \tilde{\mathcal{M}}_{\one}}\bbP\Big(\inp{\hat{\phi}(\hat{\vecZ})\notin \inb{(\inb{l_{\one}}\times[1:n^2])}\cup\inb{\twob}}\bigcup\\
& \qquad\inp{\hat{\phi}(\hat{\vecZ}) = (l_{\one}, l_{\two})\text{ for some }l_{\two}, \, \phi_{f_{l_{\one}},f_{l_{\two}}}(\vecZ)\notin {(\inb{{m}_{\one}}\times{\cM}_{\two})\cup\inb{\twob}}}\bigg|\tilde{\vecX} = \tilde{f}_{\one}(l_{\one}, {m}_{\one}), \tilde{\vecY} = (\hat{\vecy}\times \vecy)\Big)\Bigg)\\
&\leq \max_{\hat{\vecy}\in\cY^{k_{n}}}\frac{1}{n^2}\sum_{l_{\one}\in [1:n^2]}\bbP\Big(\inp{\hat{\phi}(\hat{\vecZ})\notin \inb{(\inb{l_{\one}}\times[1:n^2])}\cup\inb{\twob}}\bigg|\hat{\vecX} = \hat{f}_{\one}(l_{\one}), \hat{\vecY} = \hat{\vecy}\Big)\\
&\qquad+\max_{\stackrel{\vecy\in\cY^{n}}{l_{\two}\in [1:n^2]}}\frac{1}{n^2{N}_{\one}}\sum_{\stackrel{l_{\one}\in [1:n^2]}{{m}_{\one}\in {\mathcal{M}}_{\one}}}\bbP\inp{\phi_{f_{l_{\one}},f_{l_{\two}}}(\vecZ)\notin {(\inb{{m}_{\one}}\times{\cM}_{\two})\cup\inb{\twob}}}\bigg|{\vecX} = f_{\one, l_{\one}}( {m}_{\one}), {\vecY} =  \vecy\Big)\\
&\leq 2\epsilon
\end{align} where the last inequality follows from \eqref{eq:rand_red_33} and \eqref{eq:two}. Similarly, we can argue that $P_{e,\malone}\leq 2\epsilon$. Next, we note from \eqref{honest_error_ub} that $P_{e,\na}\leq P_{e,\malone} +P_{e,\maltwo}\leq 4\epsilon$.
\end{proof}

\newpage
\bibliography{refs} 
\bibliographystyle{ieeetr}


\end{document}